\documentclass[11pt]{article}
\usepackage{fullpage,amsthm,amssymb,amsmath}
\usepackage{enumerate}
\usepackage{hyperref}
\usepackage[normalem]{ulem}

\usepackage{tikz} 
\usetikzlibrary{shapes,fit}
\usetikzlibrary{calc}
\usetikzlibrary{decorations.pathmorphing}

\hypersetup{
    pdftitle=   {Generalized feedback vertex set problems on bounded-treewidth graphs: chordality is the key to single-exponential parameterized algorithms},
   pdfauthor=  {\'Edouard Bonnet, Nick Brettell, O{-}joung Kwon, and D\'aniel Marx}
}
\usepackage{authblk}
\usepackage{mathabx}
\usepackage{subcaption}
\newtheorem{theorem}{Theorem}[section]
\newtheorem{lemma}[theorem]{Lemma}
\newtheorem{proposition}[theorem]{Proposition}

\newtheorem{claim}{Claim}
\newenvironment{clproof}{\begin{list}{}{%
              \setlength{\leftmargin}{5mm}%
              } \item {\it Proof.} }{\hfill$\lozenge$\end{list}\medskip}

\usepackage{amssymb}
\usepackage{xspace}
\usepackage{algorithm}
\usepackage{algpseudocode2}

\usepackage{todonotes}
\usepackage{graphicx}
\usepackage{thm-restate}
\usepackage{comment}

\usepackage{hyperref}
\usepackage[capitalise,compress]{cleveref}
\newcommand\tw{\operatorname{tw}}

\newcommand\cC{\mathcal{C}}
\newcommand\cF{\mathcal{F}}

\newcommand\cR{\mathcal{R}}
\newcommand\cP{\mathcal{P}}

\newcommand\cB{\mathcal{B}}

\newcommand\cZ{\mathcal{Z}}
\newcommand\cY{\mathcal{Y}}
\newcommand\cA{\mathcal{A}}
\newcommand\cO{\mathcal{O}}

\newcommand\cM{\mathcal{M}}
\newcommand\cK{\mathcal{K}}
\newcommand\cU{\mathcal{U}}

\newcommand\cX{\mathcal{X}}

\newcommand{\YES}{\textsc{Yes}}
\newcommand{\NO}{\textsc{No}}

\newcommand{\Comp}{\operatorname{Comp}}
\newcommand{\Block}{\operatorname{Block}}
\newcommand{\Part}{\mathbf{Aux}}
\newcommand{\Parti}{\operatorname{Part}}

\newcommand{\Inc}{\mathbf{Inc}}

\newcommand{\FVS}{\textsc{Feedback Vertex Set}\xspace}
\newcommand{\COC}{\textsc{Component Order Connectivity}\xspace}
\newcommand{\MC}{\textsc{Multicolored Clique}\xspace}
\newcommand{\SI}{\textsc{Subgraph Isomorphism}\xspace}

\newcommand{\BPBVD}{\textsc{Bound\-ed $\mathcal{P}$-Block Vertex Deletion}\xspace}
\newcommand{\BPCVD}{\textsc{Bound\-ed $\mathcal{P}$-Component Vertex Deletion}\xspace}

\newcommand{\kkIS}{\textsc{$k \times k$ Independent Set}\xspace}
\newcommand{\PkkIS}{\textsc{Permutation $k \times k$ Independent Set}\xspace}
\newcommand{\PkkC}{\textsc{Permutation $k \times k$ Clique}\xspace}

\renewcommand{\leq}{\leqslant}
\renewcommand{\geq}{\geqslant}
\renewcommand{\le}{\leqslant}
\renewcommand{\ge}{\geqslant}

\newcommand\p[2]{\langle #1,#2 \rangle}
\newcommand\abs[1]{\lvert #1\rvert}

\newcommand{\Problem}[5]{
\ \\
\noindent
\begin{tabular}{|p{0.97\linewidth}|}
  \hline
  #1 \hfill \textbf{Parameter: }#2\\
  \textbf{Input: }#3\\
  \textbf{#5 }#4\\
  \hline
\end{tabular}\\
}

\newcommand{\ProblemQ}[4]{\Problem{#1}{#2}{#3}{#4}{Question:}}

\begin{document}
\title{Generalized feedback vertex set problems on bounded-treewidth graphs: chordality is the key to single-exponential parameterized algorithms\thanks{All authors were supported by ERC Starting Grant PARAMTIGHT (No. 280152).}}

\author[1]{\'Edouard Bonnet}
\author[2]{Nick Brettell}
\author[3]{O-joung Kwon\thanks{Supported by the National Research Foundation of Korea (NRF) grant funded by the Ministry of Education (No. NRF-2018R1D1A1B07050294).}}
\author[4]{D\'aniel Marx\thanks{Supported by ERC Consolidator Grant SYSTEMATICGRAPH (No. 725978).}}

\affil[1]{Univ Lyon, CNRS, ENS de Lyon, Universit\'e Claude Bernard Lyon 1, LIP UMR5668, France}
\affil[2]{Department of Mathematics and Computer Science, Eindhoven University of Technology, The Netherlands}
\affil[3]{Department of Mathematics, Incheon National University, Incheon, South Korea}
\affil[4]{Institute for Computer Science and Control, Hungarian Academy of Sciences, (MTA SZTAKI)}

\date\today
\maketitle

\makeatletter
\def\blfootnote{\gdef\@thefnmark{}\@footnotetext}
\makeatother

\blfootnote{E-mail addresses: \texttt{edouard.bonnet@dauphine.fr} (E.\ Bonnet), \texttt{nbrettell@gmail.com} (N.\ Brettell), \texttt{ojoungkwon@gmail.com} (O.\ Kwon), \texttt{dmarx@cs.bme.hu} (D.\ Marx) \\
 An extended abstract appeared in Proceedings of the 12th International Symposium on Parameterized and Exact Computations, 2017~\cite{BonnetBKM2017}. The corresponding author is O-joung Kwon.}

\begin{abstract}
  It has long been known that \textsc{Feedback Vertex Set} can
  be solved in time $2^{\mathcal{O}(w\log w)}n^{\mathcal{O}(1)}$ on $n$-vertex graphs of treewidth
  $w$, but it was only recently that this running time was improved to
  $2^{\mathcal{O}(w)}n^{\mathcal{O}(1)}$, that is, to single-exponential
  parameterized by treewidth. We investigate which generalizations of
  \textsc{Feedback Vertex Set} can be solved in a similar running time.
  Formally, for a class $\mathcal{P}$ of graphs, the \textsc{Bounded $\mathcal{P}$-Block Vertex Deletion} problem
  asks, given a graph~$G$ on $n$ vertices and positive integers~$k$
  and~$d$, whether $G$ contains a set~$S$ of at most $k$ vertices
  such that each block of $G-S$ has at most $d$
  vertices and is in $\mathcal{P}$. Assuming that $\mathcal{P}$ is
  recognizable in polynomial time and satisfies a certain natural
  hereditary condition, we give a sharp characterization of when
  single-exponential parameterized algorithms are possible for fixed
  values of $d$:
\begin{itemize}
\item if $\mathcal{P}$ consists only of chordal graphs, then the problem can be solved in time $2^{\mathcal{O}(wd^2)} n^{\mathcal{O}(1)}$, 
\item if $\mathcal{P}$ contains a graph with an induced cycle of length $\ell\ge 4$, then the problem is not solvable in time $2^{o(w\log w)} n^{\mathcal{O}(1)}$ 
even for fixed $d=\ell$, unless the ETH fails.
\end{itemize}
We also study a similar problem, called \textsc{Bound\-ed $\mathcal{P}$-Component Vertex Deletion}, where the target graphs have connected components of small size rather than blocks of small size, 
 and we present analogous results. For this problem, we also show that if $d$ is part of the input and $\mathcal{P}$ contains all chordal graphs, 
 then it cannot be solved in time $f(w)n^{o(w)}$ for some function $f$, unless the ETH fails.
  \end{abstract}

\section{Introduction}\label{sec:intro}

Treewidth is a measure of how well a graph accommodates a decomposition into a tree-like structure.
In the field of parameterized complexity, many NP-hard problems have been shown to have FPT algorithms when parameterized by treewidth; for example, \textsc{Coloring}, \textsc{Vertex Cover},  \textsc{Feedback Vertex Set}, and \textsc{Steiner Tree} (see \cite[Section~7]{Cygan15} for further examples). 
In fact, Courcelle~\cite{Courcelle1990} established a meta-theorem that says that every problem definable in MSO$_2$ logic can be solved in linear time on graphs of bounded treewidth. While Courcelle's Theorem is a very general tool for obtaining algorithmic results, for specific problems dynamic programming techniques usually give algorithms where the running time $f(w)n^{\cO(1)}$ has better dependence on treewidth $w$. There is some evidence that a careful implementation of dynamic programming (plus maybe some additional ideas) gives optimal dependence for some problems (see, e.g., \cite{LokshtanovMS11a}). 

For \textsc{Feedback Vertex Set}, standard dynamic programming techniques give $2^{\cO(w\log w)}n^{\cO(1)}$-time algorithms and it was considered plausible that this could be the best possible running time. Hence, it was a remarkable surprise when it turned out that $2^{\cO(w)}n^{\cO(1)}$-time algorithms are also possible for this problem by various techniques:
Cygan et al.~\cite{Cygan2011} obtained a $3^w n^{\mathcal{O}(1)}$-time randomized algorithm by using the so-called Cut \& Count technique, and Bodlaender et al.~\cite{BodlaenderCKN2015} showed there is a deterministic $2^{\mathcal{O}(w)} n^{\mathcal{O}(1)}$-time algorithm by using a rank-based approach and the concept of representative sets. 
This was also later shown in the more general setting of representative sets in matroids by Fomin et al.~\cite{FedorDS2014}.

\paragraph*{Generalized feedback vertex set problems.} In this paper, we explore the extent to which these results apply for generalizations of \textsc{Feedback Vertex Set}. The \textsc{Feedback Vertex Set} problem asks for a set $S$ of at most $k$ vertices such that $G-S$ is acyclic, or in other words, every block of $G-S$ is a single edge or a vertex. %
We consider generalizations where we allow the blocks to be some other type of small graph, such as triangles, small cycles, or small cliques; these generalizations were first studied in \cite{BonnetBKM2016}. 

Formally, we consider the following problem.
Let $\cP$ be a class of graphs. 

\ProblemQ{\BPBVD}{$d$, $w$}{A graph $G$ of treewidth at most $w$, and positive integers $d$ and $k$.}{Is there a set $S$ of at most $k$ vertices in $G$ such that each block of $G-S$ has at most $d$ vertices and is in $\cP$?}

	If $d=1$ or $\mathcal{P}=\{K_1\}$, then this problem is equivalent to the \textsc{Vertex Cover} problem.
	It is well known that \textsc{Vertex Cover} admits a $2^{\cO(w)}n^{\cO(1)}$-time algorithm; see \cite{Cygan15} for instance.
	Moreover, if either ($d=2$ and $\{K_1, K_2\}\subseteq \mathcal{P}$) or ($d\ge 3$ and $\mathcal{P}=\{K_1, K_2\}$), 
	then this problem is equivalent to the  \textsc{Feedback Vertex Set} problem.
	In this case, the result of Bodlaender et al.~\cite{BodlaenderCKN2015} implies that \BPBVD can be solved in time $2^{\mathcal{O}(w)} n^{\mathcal{O}(1)}$.
	Our main question is: when we regard $d$ as a fixed constant,
    for which graph classes $\cP$ can this problem be solved in time $2^{\mathcal{O}(w)}n^{\mathcal{O}(1)}$?

	To obtain a general result, we require some assumptions on the class 
$\cP$. First, in order to ensure that the solution can be checked in 
polynomial time, we assume that $\cP$ can be recognized in polynomial 
time. Second, for deletion problems, it is usually reasonable to assume 
that a superset of a solution $S$ is also a solution: deleting more 
vertices never hurts. If we define $\cC_{\cP}$ to be the class of graphs 
where every block is in $\cP$, then we want to consider deletion 
problems where $\cC_{\cP}$ is \emph{hereditary}; that is, for every 
graph $G\in \cC_{\cP}$ and every induced subgraph $H$ of $G$, we have 
$H\in \cC_{\cP}$. It is easy to see that if $\cP$ is hereditary, then 
$\cC_{\cP}$ is also hereditary. However, for technical reasons, in our 
setting it is more natural to consider a slightly weaker notion. Suppose 
that we want to express the problem "Delete $k$ vertices such that every 
block is a cycle or an edge." We can express this problem by letting 
$\cP$ be the class containing $K_1$, $K_2$, and every cycle. But this 
class is not hereditary: to make $\cP$ hereditary, we would need to add 
every path and disjoint union of paths; but clearly, these 
(non-biconnected) graphs are irrelevant for our problem. Therefore, it 
is natural to require $\cP$ to be \emph{block-hereditary} only: for 
every $G\in \cP$ and every biconnected induced subgraph $H$ of $G$, we 
have $H\in \cP$. The class consisting of $K_1$, $K_2$, and all cycles is 
block-hereditary.

	However, these two conditions are not sufficient to obtain single-exponential algorithms parameterized by treewidth.
	A graph is \emph{chordal} if it has no induced cycles of length at least $4$. 
	The main result of this paper is that the existence of single-exponential algorithms is closely linked to whether the graphs in $\cP$ we are allowing are all chordal or not.
	We show that if $\cP$ consists of all chordal graphs and satisfies the two previously mentioned conditions, then 
	\BPBVD can be solved in single-exponential time.

\begin{restatable}{theorem}{mainthmb}
\label{thm:main1b}
Let $\cP$ be a class of graphs that is block-hereditary, recognizable in polynomial time, and consists of only chordal graphs.
Then \BPBVD can be solved in time $2^{\mathcal{O}(wd^2)}k^2 n$ on graphs with $n$ vertices and treewidth $w$.
\end{restatable}
\noindent

	We complement this result by showing that if $\cP$ contains a graph that is not chordal, then single-exponential algorithms are not possible (assuming ETH), even for fixed $d$.
	Note that 
	if $\cP$ is block-hereditary and contains a graph that is not chordal, then this graph contains a chordless cycle on $\ell\ge 4$ vertices, and consequently the cycle graph on $\ell$ vertices is also in $\cP$.
	
\begin{theorem}
  	\label{thm:main1c1}
    Let $\cP$ be a block-hereditary class of graphs that is polynomial-time recognizable.
	If $\cP$ contains the cycle graph on $\ell \ge 4$ vertices, then \BPBVD is not solvable in time $2^{o(w\log w)}n^{\cO(1)}$ on graphs with $n$ vertices and treewidth at most $w$ 
	even for fixed $d=\ell$, unless the ETH fails.
\end{theorem}

	Baste, Sau, and Thilikos~\cite{BasteST17} recently studied the complexity of a similar problem, where the task is to find a set of vertices whose deletion results in a graph with no minor in a given collection of graphs $\mathcal{F}$, parameterized by treewidth.  When $\mathcal{F} = \{C_4\}$, this is equivalent to \BPBVD where $\cP = \{K_1, K_2,K_3\}$, and the complexity they obtain in this case is consistent with our result.
	
    Whether this lower bound of \cref{thm:main1c1} is best possible when $\cP$ contains a cycle on $\ell\ge 4$ vertices remains open.  
	However, as partial positive evidence towards this, we note that when $\cP$ contains all graphs, the result by 
	Baste, Sau, and Thilikos~\cite{BasteST17} implies that that \BPBVD can be solved in time $2^{\mathcal{O}(w\log w)}n^{\mathcal{O}(1)}$ when $d$ is fixed, 
	as the minor obstruction set $\mathcal{F}$ consists of $2$-connected graphs with $d+1$ vertices, 
	and contains a planar graph: the cycle graph of length $d+1$.

\paragraph*{Bounded-size components.} 
Using a similar technique, we can obtain analogous results for a simpler problem, which we call \BPCVD, where we want to remove at most $k$ vertices such that each connected component of the resulting graph has at most $d$ vertices and belongs to $\cP$. If we have only the size constraint (i.e., $\cP$ contains every graph), then this problem is known as \COC~\cite{DrangeDV2014}.

Let $\mathcal{P}$ be a class of graphs.

\ProblemQ{\BPCVD}{$d$, $w$}{A graph $G$ of treewidth at most $w$, and positive integers $d$ and $k$.}{Is there a set $S$ of at most $k$ vertices in $G$ such that each connected component of $G-S$ has at most $d$ vertices and is in $\cP$?}

Drange, Dregi, and van 't Hof~\cite{DrangeDV2014} studied the parameterized complexity of a weighted variant of the \COC problem; their results imply, in particular, that \COC can be solved in time $2^{\mathcal{O}(k\log d)}n$, but is $W[1]$-hard parameterized by only $k$ or $d$.
The corresponding edge-deletion problem, parameterized by treewidth, was studied by Enright and Meeks~\cite{Enright2015}.
For general classes $\cP$, we prove results that are analogous to those for \BPBVD.

\begin{restatable}{theorem}{mainthma}
\label{thm:main1a}
Let $\cP$ be a class of graphs that is hereditary, recognizable in polynomial time, and consists of only chordal graphs.
Then \BPCVD can be solved in time $2^{\mathcal{O}(wd^2)}k^2 n$ on graphs with $n$ vertices and treewidth $w$.
\end{restatable}

\begin{theorem}
  \label{thm:main1c2}
Let $\cP$ be a hereditary class of graphs that is polynomial-time recognizable.
If $\cP$ contains the cycle graph on $\ell \ge 4$ vertices, then \BPCVD is not solvable in time $2^{o(w\log w)}n^{\cO(1)}$ on graphs with $n$ vertices and treewidth at most $w$ even for fixed $d=\ell$, unless the ETH fails.
\end{theorem}
	
	Similar to \BPBVD, the result of Baste, Sau, and Thilikos~\cite{BasteST17} implies that
	when $\cP$ contains all graphs, \BPCVD can be solved in time $2^{\mathcal{O}(w\log w)}n^{\mathcal{O}(1)}$ when $d$ is fixed.

	When $d$ is not fixed, one might ask whether \BPCVD admits an $f(w)n^{\cO(1)}$-time algorithm; that is, an FPT algorithm parameterized only by treewidth.
	We provide a negative answer, showing that the problem is $W[1]$-hard when $\cP$ contains all chordal graphs, even parameterized by both treewidth and $k$.
	We further prove two stronger lower bound results assuming the ETH holds.

\begin{theorem}
  \label{hardness-intro}
  Let $\cP$ be a hereditary class containing all chordal graphs.  Then \BPCVD
  is $W[1]$-hard parameterized by the combined parameter $(w,k)$.
  Moreover, unless the ETH fails, this problem
  \begin{enumerate}
  \item has no $f(w)n^{o(w)}$-time algorithm; and 
  \item has no $f(k')n^{o(k'/\log k')}$-time algorithm, where $k' = w + k$.
  \end{enumerate}
\end{theorem}

\paragraph*{Techniques for positive results.} 
	We sketch the proof of Theorem~\ref{thm:main1b}.
	Let $\mathcal{P}$ be a class of graphs that is block-hereditary and consists of chordal graphs. 
    A pair $(G,S)$ consisting of a graph $G$ and a subset $S$ of its vertex set will be called a \emph{boundaried graph}.

	The key lemma can be briefly described as follows.
	Suppose there are two boundaried graphs $(G,S)$ and $(H,S)$ with $G[S]=H[S]$, and 
	we want to know whether 
	\begin{itemize}
	\item[$(\ast)$] the graph obtained from $G$ and $H$ by identifying vertices in $S$ has at most $d$ vertices and its blocks are in $\cP$.
	\end{itemize}
	In the dynamic programming algorithm, we consider one part $(G,S)$ as a partial solution, and $(H,S)$ has a role in the hypothetical complementary solution.
	We will show that we can guarantee the statement $(\ast)$ if 
	\begin{enumerate}[(i)]
	\item $G$ and $H$ each have at most $d$ vertices and their blocks are in $\cP$,
	\item for each non-trivial block $B$ of $G[S]$, the block of $G$ containing $B$ and the block of $H$ containing $B$ have no conflict near $B$ (we explain this below), and
	\item if we make an auxiliary bipartite graph with bipartition $(\mathcal{A}, \mathcal{B})$ where 
	\begin{itemize}
		\item $\mathcal{A}$ is the set of connected components of $G[S]$, 
		\item $\mathcal{B}$ is the union of the set of connected components of $G$ and the set of connected components of $H$,
		\item $X\in \mathcal{A}$ is adjacent to $Y\in \mathcal{B}$ if $X$ is contained in $Y$,
	\end{itemize}
	then this bipartite graph has no cycles.
	\end{enumerate}
	Section~\ref{sec:chordalsum} is devoted to showing a simplified version of this statement (Proposition~\ref{prop:sumofchordalgraphs2}).
	
	To establish the condition (ii), 
	we guess a graph $g(B)$ for each non-trivial block $B$ of $G[S]$, where $g(B)$ is the block containing $B$ after combining $G$ and $H$.
	Note that this target graph $g(B)$ must be a biconnected chordal graph with at most $d$ vertices. 
	So we consider $g(B)$ to be a biconnected chordal graph with distinct labels from $\{1, \ldots, d\}$.
	The necessary local information described in (ii) will be the set of labels of neighbors of $B$ (with fixed labels on $B$) in the block of $G$ containing $B$.
	We will store this as $h(B)$.
	The important point is that for a chordal graph $F$ and a connected vertex set $Z$, 
	there is an one-to-one correspondence between the connected components of $F-Z$  and the connected components of the neighborhood of $Z$ in $F$ (see Lemma~\ref{lem:chordalseparator}).
	Therefore, the neighbors of $B$ provide information about which connected components currently exist around $B$.
	The meaning of ``having no conflict'' in (ii) is that the neighbors of $B$ in the block of $G$ and in the block of $H$ have disjoint sets of labels.
	The pair $(g,h)$ will be considered as an index of the table of our dynamic programming algorithm.

	Once we have considered (i) and (ii), we need to deal with the auxiliary bipartite graph in (iii). 
	For the $(G,S)$ part, it is sufficient to know the auxiliary bipartite graph with components of $G$.
	This can be stored as a partition of the set of connected components of $G[S]$.
	As the size of $S$ corresponds to the treewidth of the given graph, 
	to obtain a single-exponential algorithm parameterized by treewidth, we need to efficiently deal with these partitions corresponding to partial solutions.
	This part can be dealt with in a similar manner to the single-exponential time algorithm for \FVS, using representative-set techniques.
	We recall the representative-set technique in Section~\ref{sec:representative}, and prove a variant that is fit for our case.
	
	In the algorithm, for each bag $B_t$ of the tree decomposition, 
	we guess a deletion set $X$ in $B_t$, and guess $(g,h)$ for blocks in $B_t\setminus X$.
	Whenever there is a partial solution corresponding to these information, 
	we keep the corresponding partition of the set of connected components on the boundary $B_t\setminus X$.
	As we take a representative set after partial solutions are updated, we can solve the problem in time $2^{\mathcal{O}(w)}n^{\mathcal{O}(1)}$.

\paragraph*{Lower bounds.}
Theorem~\ref{thm:main1c2} is obtained by a reduction from \PkkIS, the problem of finding an independent set of size $k$ in a graph with $k^2$ vertices and $O(k^4)$ edges.
One can think of those vertices as forming a $k$-by-$k$ grid, where one should select exactly one vertex per row and per column.
This problem cannot be solved in time $2^{o(k \log k)}k^{\cO(1)}$, unless the ETH fails \cite{LokshtanovMS11}.
The crucial point is that the treewidth of the equivalent instances of \BPCVD and \BPBVD should be in $\Theta(k)$.
We achieve this by stretching the information into a chain of $O(k^4)$ almost identical pieces, each encoding one edge of the initial graph.
The pieces are linked by small separators of size $2k$ that propagate the row and column indices of each of the $k$ choices for the independent set. 

For Theorem~\ref{hardness-intro}, we propose a reduction from \MC for the first item, and more or less the same reduction but from \SI for the second.
Again, the crux of the construction is obtaining an instance with low treewidth.
This time, we rely on an injective mapping of edges into integers, which is a folklore trick.
Vertices of the initial graph are encoded as a collection of candidate places where the constructed graph can be disconnected, regularly positioned on two \emph{paths}, one with a small weight and one with a larger weight.
The edge gadget is similarly realized with certain vertices that are candidates for removal, as they can disconnect the constructed graph, each corresponding to a specific edge.

\paragraph*{Organization.}
The paper is organized as follows. Section~\ref{sec:preliminaries} introduces the necessary notions including labelings, treewidth, and boundaried graphs. 
In Section~\ref{sec:chordalsum},
we prove structural lemmas about $S$-blocks, and in Section~\ref{sec:representative}, we discuss representative sets for acyclicity.
In Section~\ref{sec:BPBVD}, we prove Theorems~\ref{thm:main1b} and \ref{thm:main1a}.
Section~\ref{sec:lowerbound} shows that if $\cP$ contains the cycle graph on $d$ vertices, then both problems are not solvable in time $2^{o(w\log w)}n^{\cO(1)}$ on graphs of treewidth at most $w$,
unless the ETH fails.
In Section~\ref{sec:w1hardness}, we further show that if $d$ is not fixed and $\cP$ contains all chordal graphs, then 
\BPCVD is $W[1]$-hard when parameterized by both $k$ and $w$.

\section{Preliminaries}\label{sec:preliminaries}

Let $G$ be a graph.
We denote the vertex set and the edge set of $G$ by $V(G)$ and $E(G)$, respectively.
For a vertex $v$ in $G$, we denote by $G-v$ the graph obtained by removing $v$ and its incident edges, and
for $X\subseteq V(G)$, we denote by $G-X$ the graph obtained by removing all vertices in $X$ and their incident edges.
For $X\subseteq V(G)$, we denote by $G[X]$ the subgraph induced by the vertex set $X$.
A subgraph $H$ of $G$ is an \emph{induced subgraph} of $G$ if $H=G[X]$ for some vertex subset $X$ of $G$.
For two graphs $G_1$ and $G_2$,  $G_1\cup G_2$ is the graph with the vertex set $V(G_1) \cup V(G_2)$ and the edge set $E(G_1) \cup E(G_2)$, 
and $G_1\cap G_2$ is the graph with the vertex set $V(G_1) \cap V(G_2)$ and the edge set $E(G_1) \cap E(G_2)$.

For a vertex $v$ in $G$, we denote by $N_G(v)$ the set of neighbors of $v$ in $G$, 
and $N_G[v]:=N_G(v)\cup \{v\}$.
For $X\subseteq V(G)$, we let $N_G(X):=(\bigcup_{v \in X} N_G(v)) \setminus X$.

A vertex $v$ of $G$ is a {\em cut vertex} if the deletion of $v$ from $G$ increases the number of connected components. 
We say $G$ is \emph{biconnected} if it is connected and has no cut vertices.
Note that every connected graph on at most two vertices is biconnected.
A \emph{block} of $G$ is a maximal biconnected subgraph of $G$.
We say $G$ is \emph{$2$-connected} if it is biconnected and $\abs{V(G)}\ge 3$.

The \emph{length} of a path is the number of edges in the path.
Similarly, the \emph{length} of a cycle is the number of edges in the cycle.

An induced cycle of length at least four is called a \emph{chordless cycle}.
A graph is \emph{chordal} if it has no chordless cycles.
For a class of graphs $\cP$, 
a graph is called a \emph{$\cP$-block graph}
if each of its blocks is in $\cP$.

	For two integers $d_1,d_2$ with $d_1 \leq d_2$, 
	let $[d_1,d_2]$ be the set of all integers $i$ with $d_1\le i\le d_2$, 
	and for a positive integer $d$, let $[d]:=[1,d]$.
For a function $f:X\rightarrow Y$ and $X'\subseteq X$, 
the function $f':X'\rightarrow Y$ where $f'(x)=f(x)$ for all $x\in X'$
is called the \emph{restriction} of $f$ on $X'$, and is denoted $f|_{X'}$.
For such a pair of functions $f$ and $f'$, we also say that $f$ \emph{extends} $f'$ to the set $X$.

\subsection{Chordal graphs}

	We will use the following property of chordal graphs.
	\begin{lemma}\label{lem:chordalseparator}
	Let $G$ be a connected chordal graph and $X$ be a vertex subset such that $G[X]$ is connected.
	Then there is a bijection $f$ from the set of connected components of $G[N_G(X)]$ to the the set of connected components of $G-X$ such that a connected component $C$ of $G[N_G(X)]$ is contained in a connected component $H$ of $G-X$ if and only if $H=f(C)$.
	\end{lemma}
	\begin{proof}
	It is sufficient to show that no connected component of $G-X$ contains
	two connected components of $G[N_G(X)]$.
	Suppose for a contradiction that 
	there is a connected component $H$ of $G-X$ containing 
	at least two connected components of $G[N_G(X)]$.
	Let $P$ be a shortest path between two connected components of $G[N_G(X)]$ in $H$, 
	with endpoints $x_1$ and $x_2$.
	Let $Q$ be a shortest path from $N_G(x_1)\cap X$ to $N_G(x_2)\cap X$ in $G[X]$, 
	with endpoints $y_1\in N_G(x_1)\cap X$ and $y_2\in N_G(x_2)\cap X$.
	Then $x_1-y_1-Q-y_2-x_2-P-x_1$ is a chordless cycle, contradicting the fact that $G$ is a chordal graph.
	
	Since $G$ is connected, each connected component of $G-X$ contains exactly one connected component of $G[N_G(X)]$.
	Thus, the required bijection exists.
	\end{proof}

\subsection{Block $d$-labeling}

	For a graph $G$ where every block has at most $d$ vertices, 
	a \emph{block $d$-labeling} of $G$ is a function $L:V(G) \rightarrow [d]$ such that for each block $B$ of $G$, $L|_{V(B)}$ is an injection.
	If a graph is equipped with a block $d$-labeling $L$, then it is called a \emph{block $d$-labeled graph}, and we call $L(v)$ the \emph{label} of $v$.
	Two block $d$-labeled graphs $G$ and $H$ are \emph{label-isomorphic} if there is a graph isomorphism from $G$ to $H$ that is label preserving.
	For biconnected block $d$-labeled graphs $G$ and $H$, we say $H$ is \emph{partially label-isomorphic} to $G$ if $H$ is label-isomorphic to the subgraph of $G$ 
	induced by the vertices with labels in $H$.	
	Where there is no ambiguity, a block $d$-labeled graph will simply be called a $d$-labeled graph.

\subsection{Treewidth}

	A \emph{tree decomposition} of a graph $G$ is 
	a pair $(T,\cB)$ 
	consisting of a tree $T$
	and a family $\cB=\{B_t\}_{t\in V(T)}$ of sets $B_t\subseteq V(G)$,
	called \emph{bags},
	satisfying the following three conditions:
	\begin{enumerate}
	\item $V(G)=\bigcup_{t\in V(T)}B_t$,
	\item for every edge $uv$ of $G$, there exists a node $t$ of $T$ such that $u,v\in B_t$, and
	\item for $t_1,t_2,t_3\in V(T)$, $B_{t_1}\cap B_{t_3}\subseteq B_{t_2}$ whenever $t_2$ is on the path from $t_1$ to $t_3$ in $T$.
	\end{enumerate}
	The \emph{width} of a tree decomposition $(T,\cB)$ is $\max\{ \abs{B_{t}}-1:t\in V(T)\}$.	
	The \emph{treewidth} of $G$ is the minimum width over all tree decompositions of $G$. 
	A \emph{path decomposition} is a tree decomposition $(P,\mathcal{B})$ where $P$ is a path.
	The \emph{pathwidth} of $G$ is the minimum width over all path decompositions of $G$. 
	We denote a path decomposition $(P,\mathcal{B})$ as $(B_{v_1},\dotsc,B_{v_t})$, where $P$ is a path $v_1v_2\dotsb v_t$.

To design a dynamic programming algorithm, we use a convenient form of a tree decomposition known as a nice tree decomposition.
A tree $T$ is said to be \emph{rooted} if it has a specified node called the \emph{root}.
Let $T$ be  a rooted tree with root node $r$.
A node $t$ of $T$ is called a \emph{leaf} node if it has degree one and it is not the root.
For two nodes $t_1$ and $t_2$ of $T$, $t_1$ is a \emph{descendant} of $t_2$ if the unique path from $t_1$ to $r$ contains $t_2$. 
If a node $t_1$ is a descendant of a node $t_2$ and $t_1t_2\in E(T)$, 
then $t_1$ is called a \emph{child} of $t_2$.
 
A tree decomposition $(T,\cB=\{B_t\}_{t\in V(T)})$ is a \emph{nice tree decomposition} with root node $r\in V(T)$ if $T$ is a rooted tree with root node $r$, and every node $t$ of $T$ is one of the following:
\begin{enumerate}
  \item a \emph{leaf node}: $t$ is a leaf of $T$ and $B_t=\emptyset$;
  \item an \emph{introduce node}: $t$ has exactly one child $t'$ and $B_t=B_{t'}\cup \{v\}$ for some $v\in V(G)\setminus B_{t'}$;
  \item a \emph{forget node}: $t$ has exactly one child $t'$ and $B_t=B_{t'}\setminus \{v\}$ for some $v\in B_{t'}$; or
  \item a \emph{join node}: $t$ has exactly two children $t_1$ and $t_2$, and $B_t=B_{t_1}=B_{t_2}$.
\end{enumerate}

\begin{theorem}[Bodlaender et al.\ \cite{BodlaenderDDFP2016}]\label{thm:approxtw}
Given an $n$-vertex graph $G$ and a positive integer $k$, one can  either output
a tree decomposition of $G$ with width at most $5k+4$, or correctly answer that 
the treewidth of $G$ is larger than $k$, in time $2^{\mathcal{O}(k)} n$.
\end{theorem}

\begin{lemma}[folklore; see Lemma 7.4 in \cite{Cygan15}]\label{lem:nicetd}
Given a tree decomposition of an $n$-vertex graph $G$ of width $w$, 
one can construct a nice tree decomposition $(T, \cB)$ of width $w$ with $\abs{V(T)}=\mathcal{O}(wn)$  in time $\mathcal{O}(k^2\cdot \max (\abs{V(T)}, \abs{V(G)}))$.
\end{lemma}

\subsection{Boundaried graphs}

	For a graph $G$ and $S\subseteq V(G)$, the pair $(G, S)$ is called a \emph{boundaried graph}.
	When $G$ is a $d$-labeled graph, 
	we simply say that $(G,S)$ is a $d$-labeled graph.
	Two $d$-labeled graphs $(G,S)$ and $(H,S)$ are said to be \emph{compatible} 
	if $V(G-S)\cap V(H-S)=\emptyset$, $G[S]=H[S]$, and $G$ and $H$ have the same labels on $S$.
	For two compatible $d$-labeled graphs $(G,S)$ and $(H,S)$, 
	the \emph{sum} of two graphs is the graph obtained from the disjoint union of $G$ and $H$ 
	by identifying each vertex of $S$ in $G$ with the same vertex in $H$ and 
	removing an edge from multiple edges that appear in $S$.
	We denote the resulting graph by $(G,S)\oplus (H,S)$.
	See Figure~\ref{fig:sumgraphs} for an example.
	
	\begin{figure}
  \centering
  \begin{tikzpicture}[scale=0.7]
  \tikzstyle{w}=[circle,draw,fill=black,inner sep=0pt,minimum width=4pt]

	\draw (1, 1) node [w] (v1) {};
	\draw (2, 1) node [w] (v2) {};
	\draw (3, 1) node [w] (v3) {};
	\draw (1, 2.5) node [w] (v4) {};
	\draw (2, 2) node [w] (v5) {};
	\draw (2, 3) node [w] (v6) {};

	\draw(v1)--(v2)--(v3);
	\draw(v1)--(v4)--(v5)--(v1);
	\draw(v4)--(v6)--(v5)--(v2);
	\draw(v5)--(v3);

	\draw (1, -1) node [w] (w1) {};
	\draw (2, -1) node [w] (w2) {};
	\draw (3, -1) node [w] (w3) {};
	\draw (1, -2) node [w] (w4) {};
	\draw (2, -2) node [w] (w5) {};
	\draw (3, -2) node [w] (w6) {};
	\draw (2.5, -2.6) node [w] (w7) {};

	\draw(w1)--(w2)--(w3);
	\draw(w4)--(w1)--(w5)--(w2)--(w6)--(w3);
	\draw(w4)--(w5)--(w6)--(w7)--(w5);

	\draw[rounded corners] (0.7, 1)--(0.7,1.3)--(3.3,1.3)--(3.3,0.7)--(0.7, 0.7)--(0.7,1);
	\draw[rounded corners] (0.7, -1)--(0.7,-1.3)--(3.3,-1.3)--(3.3,-0.7)--(0.7, -0.7)--(0.7,-1);

	\node at (-1, 2) {$(G,S)$};
	\node at (-1, -2) {$(H,S)$};

	\draw (8+1, 1-1) node [w] (x1) {};
	\draw (8+2, 1-1) node [w] (x2) {};
	\draw (8+3, 1-1) node [w] (x3) {};
	\draw (8+1, 2.5-1) node [w] (x4) {};
	\draw (8+2, 2-1) node [w] (x5) {};
	\draw (8+2, 3-1) node [w] (x6) {};

	\draw(x1)--(x2)--(x3);
	\draw(x1)--(x4)--(x5)--(x1);
	\draw(x4)--(x6)--(x5)--(x2);
	\draw(x5)--(x3);

	\draw (8+1, -1+1) node [w] (y1) {};
	\draw (8+2, -1+1) node [w] (y2) {};
	\draw (8+3, -1+1) node [w] (y3) {};
	\draw (8+1, -2+1) node [w] (y4) {};
	\draw (8+2, -2+1) node [w] (y5) {};
	\draw (8+3, -2+1) node [w] (y6) {};
	\draw (8+2.5, -2.6+1) node [w] (y7) {};

	\draw(y1)--(y2)--(y3);
	\draw(y4)--(y1)--(y5)--(y2)--(y6)--(y3);
	\draw(y4)--(y5)--(y6)--(y7)--(y5);

	\node at (11-1, -2.6) {$(G,S)\oplus (H,S)$};
	
	\draw[rounded corners] (8+0.7, 1-1)--(8+0.7,1.3-1)--(8+3.3,1.3-1)--(8+3.3,0.7-1)--(8+0.7, 0.7-1)--(8+0.7,1-1);

   \end{tikzpicture}     \caption{An example of the sum $(G,S)\oplus (H,S)$.}\label{fig:sumgraphs}
\end{figure}
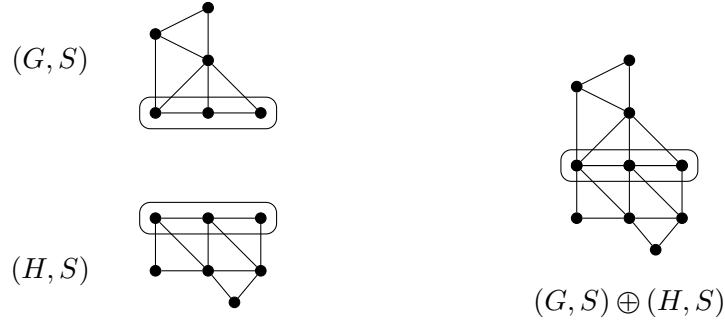
	
	We also denote by $L_G\oplus L_H$ the function from $V(G)\cup V(H)$ to $[d]$ where 
	for $v\in V(G)\cup V(H)$, $(L_G\oplus L_H)(v)=L_G(v)$ if $v\in V(G)$ and $(L_G\oplus L_H)(v)=L_H(v)$ otherwise.
	Notice that $L_G\oplus L_H$ is not necessarily a block $d$-labeling of $G\oplus H$.
	For two unlabeled boundaried graphs, we define the sum in the same way, but ignoring the label condition.

	A block of a graph is \emph{non-trivial} if it has at least two vertices.
	For a boundaried graph $(G,S)$, 
	a block $B$ of $G$ is called an \emph{$S$-block} if it contains an edge of $G[S]$.
	Note that every non-trivial block of $G[S]$ is contained in a unique $S$-block of $G$ because two distinct blocks share at most one vertex.

	Let $(G,S)$ be a boundaried graph. 
	We define $\Part (G,S)$ as the bipartite boundaried graph with bipartition $(\mathcal{X}, \mathcal{Y})$ and boundary $\mathcal{Y}$ such that
	\begin{enumerate}
	\item  $\mathcal{X}$ is the set of components of $G$, and $\mathcal{Y}$ is the set of components of $G[S]$, and
	\item  for $C_1\in \mathcal{X}$ and $C_2\in \mathcal{Y}$, $C_1C_2\in E(\Part(G,S))$ if and only if $C_2$ is contained in $C_1$.
	\end{enumerate}
	We remark that when $(G,S)$ and $(H,S)$ are two compatible $d$-labeled graphs, 
	$\Part (G,S)\oplus \Part(H,S)$ is well-defined, as $G$ and $H$ have the same set of components on $S$.
	We will use this notation to check, when we take the sum of two compatible $d$-labeled graphs $(G,S)$ and $(H,S)$, 
	whether the sum contains a chordless cycle through the cycle of 
	$\Part (G,S)\oplus \Part(H,S)$.
	
\section{Lemmas about chordal graphs and $S$-blocks}\label{sec:chordalsum}

	In this section, we present several lemmas regarding $S$-blocks. 
		
	For a biconnected $d$-labeled graph $Q$, 	
	we say that a $d$-labeled graph $(G, S)$ is \emph{block-wise partially label-isomorphic to $Q$} if 
	every $S$-block $B$ of $G$ is partially label-isomorphic to $Q$.
	A first result describes sufficient conditions for when, given a chordal labeled graph $Q$, 
	the sum of two given labeled graphs $(G,S)$ and $(H,S)$, each block-wise partially label-isomorphic to $Q$, is again block-wise partially label-isomorphic to $Q$.
	This argument will be used in the algorithm to decide whether the sum of two partial solutions is again a partial solution.

	To guarantee that the sum is again a block-wise partially label-isomoprhic to $Q$, 
	we need a compatibility condition.
	Informally, this condition arises due to the property of chordal graphs in Lemma~\ref{lem:chordalseparator}.
	Suppose $B$ is a block of $G[S]$. Then, for the sum to be label-isomorphic to $Q$, 
	if $B_1$ and $B_2$ are the $S$-blocks of $G$ and $H$ containing $B$, 
	then connected components of $B_1-V(B)$ and $B_2-V(B)$ have to indicate other components of $Q-X$, 
	where $X$ is the corresponding vertex set of $B$ in $Q$.	
	This can be checked by the labels of neighbors of $X$ in $Q$, 
	since there is a bijection between connected components of $Q-X$ and connected components of $Q[N_Q(X)]$.
	
	Formally, we define this compatibility condition as follows.
	For two compatible $d$-labeled graphs $(G,S)$ and $(H,S)$ with labelings $L_G$ and $L_H$ respectively, 
	we say that $(G,S)$ and $(H,S)$ are \emph{block-wise $Q$-compatible}
	if 
	\begin{enumerate}
	\item $(G,S)$ and $(H,S)$ are block-wise partially label-isomorphic to $Q$; and
	\item for every non-trivial block $B$ of $G[S]$, letting $B_1$ and $B_2$ be the $S$-blocks of $G$ and $H$ that contain $B$, respectively, 
	we have
	\begin{enumerate}
	\item $L_G(N_{B_1}(V(B))\setminus S) \cap L_H(N_{B_2}(V(B))\setminus S)=\emptyset$, and,
	\item for every $\ell_1\in L_G(N_{B_1}(V(B))\setminus S)$ and every $\ell_2\in L_H(N_{B_2}(V(B))\setminus S)$, the vertices in $Q$ with labels $\ell_1$ and $\ell_2$ are not adjacent.
	\end{enumerate}
	\end{enumerate}
	
	However, this local property is not sufficient to guarantee that the sum is again label-isomorphic to $Q$.
	The reason is that there might be a chordless cycle that is not captured by $S$-blocks of $(G,S)\oplus (H,S)$. 
	We provide such an example in Figure~\ref{fig:labelisomorphic}.
	Observe that, in that case, $\Part (G,S)\oplus \Part (H,S)$ has a cycle.
	On the other hand, we can show that 
	if we add the condition that 	$\Part (G,S)\oplus \Part (H,S)$ has no cycles, 
	then  the sum is indeed label-isomorphic to $Q$.

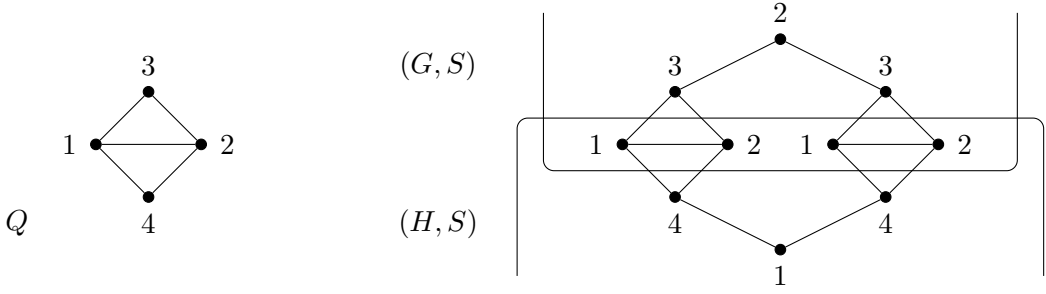
\begin{figure}
  \centering
  \begin{tikzpicture}[scale=0.7]
  \tikzstyle{w}=[circle,draw,fill=black,inner sep=0pt,minimum width=4pt]

	\draw (2-10,.5) node [w] (v1) {};
	\draw (4-10,.5) node [w] (v2) {};
	\draw (3-10,1.5) node [w] (v5) {};
	\draw (3-10,-0.5) node [w] (v6) {};

	\draw (v5)--(v1)--(v6)--(v2)--(v5);
	\draw(v1)--(v2);

	\node at (-9.5, -1) {$Q$};

	\node at (2-10-.5,.5)   {$1$};
	\node at (4-10+.5,.5)  {$2$};
	\node at (3-10,1.5+.5)  {$3$};
	\node at (3-10,-0.5-.5)  {$4$};

	\node at (2-.5,.5)   {$1$};
	\node at (4+.5,.5)  {$2$};
	\node at (3,1.5+.5)  {$3$};
	\node at (3,-0.5-.5)  {$4$};

	\node at (2+4-.5,.5)   {$1$};
	\node at (4+4+.5,.5)  {$2$};
	\node at (3+4,1.5+.5)  {$3$};
	\node at (3+4,-0.5-.5)  {$4$};

	\node at (5,3)  {$2$};
	\node at (5,-2)  {$1$};

	\draw[rounded corners] (0, -2)--(0,1)--(10,1)--(10,-2);
	\draw[rounded corners] (.5, 3)--(.5,0)--(10-.5,0)--(10-.5,3);
     \node at (-1.5, 2) {$(G,S)$};
     \node at (-1.5, -1) {$(H,S)$};

	\draw (2,.5) node [w] (w1) {};
	\draw (4,.5) node [w] (w2) {};
	\draw (6,.5) node [w] (w3) {};
	\draw (8,.5) node [w] (w4) {};

	\draw (3,1.5) node [w] (w5) {};
	\draw (3,-0.5) node [w] (w6) {};

	\draw (7,1.5) node [w] (w7) {};
	\draw (7,-0.5) node [w] (w8) {};

	\draw (5,2.5) node [w] (w9) {};
	\draw (5,-1.5) node [w] (w0) {};

	\draw (w5)--(w1)--(w6)--(w2)--(w5);
	\draw (w1)--(w2);

	\draw (w7)--(w3)--(w8)--(w4)--(w7);
	\draw (w3)--(w4);
	
	\draw (w5)--(w9)--(w7);
	\draw (w6)--(w0)--(w8);

   \end{tikzpicture}     \caption{An example where the sum of two labeled graphs $(G,S)$ and $(H,S)$, each partially label-isomorphic to $Q$, is not partially label-isomorphic to $Q$, since
   $\Part (G,S)\oplus \Part (H,S)$ has a cycle. }\label{fig:labelisomorphic}
\end{figure}

	\begin{proposition}\label{prop:sumofchordalgraphs2}
	Let $Q$ be a biconnected $d$-labeled chordal graph.
	Let $(G, S)$ and $(H,S)$ be two block-wise $Q$-compatible $d$-labeled graphs such that
	$\Part (G,S)\oplus \Part (H,S)$ has no cycles.
	Then $(G,S)\oplus (H,S)$ is block-wise partially label-isomorphic to $Q$.
	\end{proposition}

	The following lemma is an essential property of chordal graphs.
\begin{lemma}\label{lem:labelisomorphic}
	Let $F$ be a connected graph and $Q$ be a connected chordal graph. 
	Let $\mu:V(F)\rightarrow V(Q)$ be a function such that
	for every induced path $p_1 \cdots p_m$ in $F$ of length at most two, $\mu(p_1), \ldots, \mu(p_m)$ are pairwise distinct and $\mu(p_1)\cdots \mu(p_m)$ is an induced path of $Q$.
	Then $\mu$ is an injection and preserves the adjacency relation.
\end{lemma}

\begin{proof}
	We first show that $\mu$ is an injection.

	\begin{claim}\label{claim:injection} $F$ has no two vertices  $v$ and $w$ with $\mu(v)=\mu(w)$.
	\end{claim}
	\begin{clproof}
	Suppose $F$ has two distinct vertices $v$ and $w$ with $\mu(v)=\mu(w)$.
	Let $P=p_1p_2 \cdots p_x$ be a shortest path from $v=p_1$ to $w=p_x$ in $F$.
	Note that $P$ is an induced path, and by assumption, $x\ge 4$ and $\mu(p_1)\mu(p_2)\mu(p_3)$ is an induced path in $Q$.
	This further implies that $\mu(p_4)\neq \mu(p_i)$ for $i\in \{1,2,3\}$.
	Thus, we have $x\ge 5$.

	Let $y\in \{4, \ldots, x-1\}$ be the smallest integer such that $\mu(p_y)$ has a neighbor in $\{\mu(p_1), \ldots, \mu(p_{y-3})\}$.
	Such an integer exists as $\mu(p_1)=\mu(p_x)$, so $\mu(p_{x-1})$ is adjacent to $\mu(p_1)$, and 
	$\mu(p_i)\mu(p_{i+1})\mu(p_{i+2})$ is an induced path for each $1\le i\le x-2$.
	Let $\mu(p_z)$ be a neighbor of $\mu(p_y)$ with $z\in \{1, 2, \ldots, y-3\}$ and maximum $z$.
	Then $\mu(p_z)\mu(p_{z+1}) \cdots \mu(p_y)\mu(p_z)$ is an induced cycle of length at least $4$, which contradicts the assumption that
	$Q$ is chordal.
	\end{clproof}

Now, we show that 
$\mu$ preserves the adjacency relation.

\begin{claim} For each $v,w\in V(F)$, 
$vw\in E(F)$ if and only if $\mu(v)\mu(w)\in E(Q)$.
\end{claim}
\begin{clproof}
Suppose there are two vertices $v$ and $w$ in $F$ such that the adjacency relation between $v$ and $w$ in $F$ is different from the adjacency relation between $\mu(v)$ and $\mu(w)$ in $Q$.
When $vw\in E(F)$, $\mu(v)$ is adjacent to $\mu(w)$ in $Q$ by assumption.
Thus, $vw\notin E(F)$ and $\mu(v)\mu(w)\in E(Q)$.
We choose such vertices $v$ and $w$ with minimum distance in $F$.
Let $P=p_1p_2 \cdots p_x$ be a shortest path from $v=p_1$ to $w=p_x$ in $F$.
Observe that $x\ge 4$.
By the minimality of the distance,
each of $\mu(p_1)\mu(p_2) \cdots \mu(p_{x-1})$ and $\mu(p_2)\mu(p_3) \cdots \mu(p_x)$ is an induced path in $Q$.
Therefore, $\mu(p_1)\mu(p_2) \cdots \mu(p_x)\mu(p_1)$ is an induced cycle of length at least four in $Q$, 
contradicting the assumption that $Q$ is chordal. 
\end{clproof}

This completes the proof.
\end{proof}

	We need two more auxiliary lemmas to prove Proposition~\ref{prop:sumofchordalgraphs2}.

\begin{figure}
\centerline{
\includegraphics[scale=0.55]{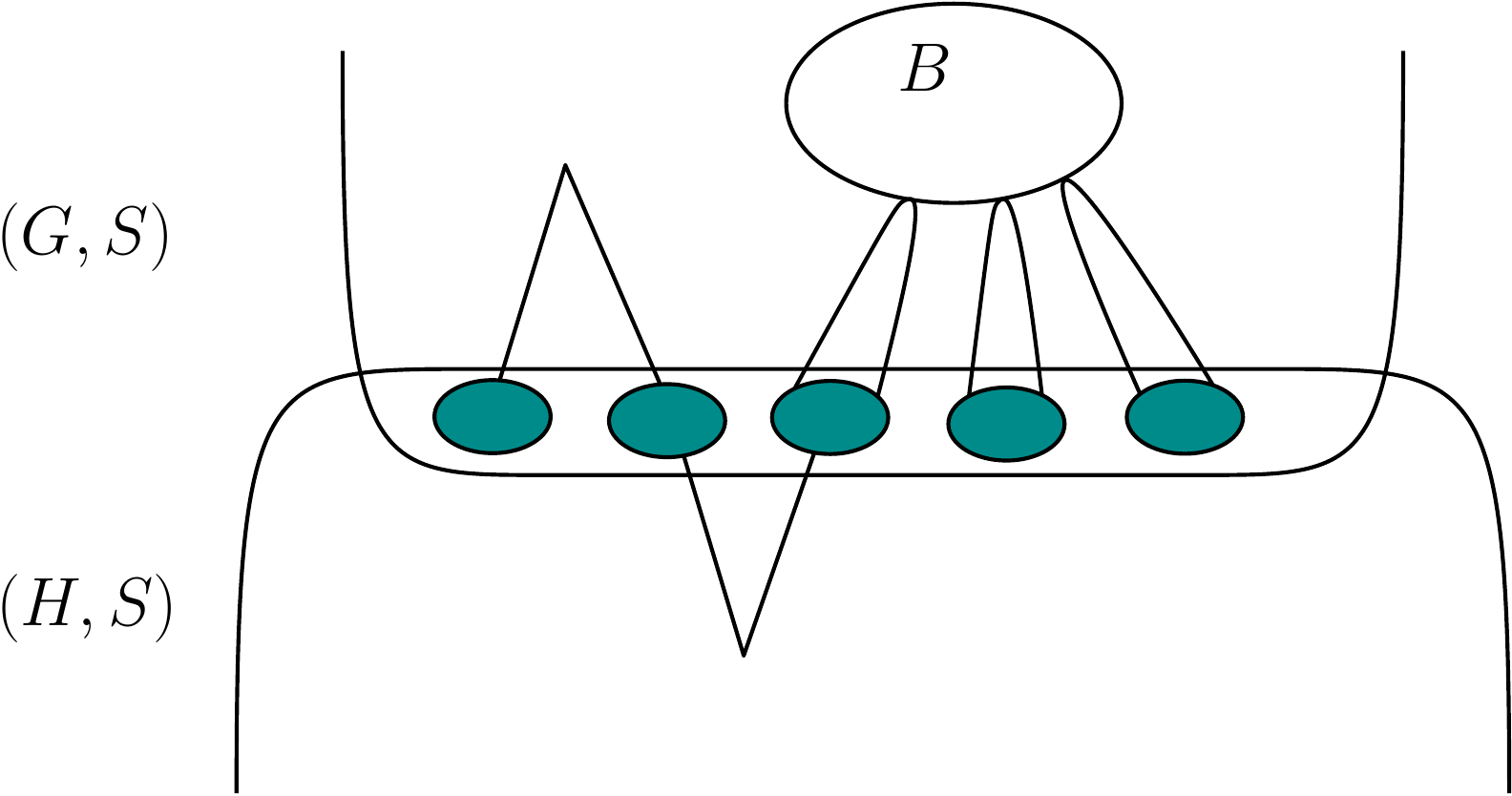}
}
 \caption{When the block $B$ is not an $S$-block in Lemma~\ref{lem:edgesblock}. For each connected component $X$ of $G[S]$, 
 there is a cut vertex of $G$ in $B$ separates $X$ from $B$ in $(G, S)\oplus (H,S)$.}\label{fig:edgesblock}
\end{figure}

	\begin{lemma}\label{lem:edgesblock}
	Let $(G, S)$ and $(H,S)$ be two compatible $d$-labeled graphs such that
	$\Part (G,S)\oplus \Part (H,S)$ has no cycles.
	If $F$ is an $S$-block of $(G,S)\oplus (H,S)$ and $uv$ is an edge in $F$, 
	then $uv$ is contained in some $S$-block of $G$ or $H$.
	\end{lemma}
	\begin{proof}
	We may assume that one of $u$ and $v$ is not contained in $S$, otherwise the block containing $uv$ in $G$ or $H$ is an $S$-block by definition.
	Without loss of generality, let us assume $v\in V(G)\setminus S$.	This implies that $u$ is also contained in $G$.

	Since $uv$ is an edge, there is a unique block of $G$ containing both $u$ and $v$.
	Let $C$ be the component of $G$ containing $u$ and $v$, and 
	let $B$ be the block of $G$ containing $u$ and $v$.
	If $B$ is an $S$-block, then we are done. Thus, we may assume that $B$ is not an $S$-block.

	For each vertex $w$ of $G$ contained in $B$, let $H_w$ be the subgraph of $G$
	induced by the union of $w$ and all components of $C-V(B)$ containing a neighbor of $w$.
	One can observe that if $H_w$ contains a vertex in a connected component of $G[S]$, 
	then $C-V(H_w)$ does not contain a vertex of that component; 
	otherwise, the existence of a cycle through $H_{w}$ and $B$ implies that $B$ is an $S$-block.
	See Figure~\ref{fig:edgesblock} for an illustration.
	This implies that for each connected component $X$ of $G[S]$ contained in $C$, 
	there is a vertex $w$ contained in $B$ such that 
	$w$ separates $B$ and $X$.
	Furthermore, 
	since $\Part (G,S)\oplus \Part (H,S)$ has no cycles, 
	for every connected component $X$ of $G[S]$, 
	there is a vertex $w$ of $G$ in $B$ such that 
	$w$ separates $B$ from $X$ in $(G,S)\oplus (H,S)$.

	As $F$ is an $S$-block of $(G,S)\oplus (H,S)$, $F$ contains an edge of $G[S]$, say $xy$. 
	Since $F$ contains $x,y$ and $v\notin S$, $F$ has at least $3$ vertices and thus it is $2$-connected.
	On the other hand, the conclusion in the previous paragraph implies that there is a vertex $w$ such that 
	$w$ separates $B$ and $\{x,y\}$ in $(G,S)\oplus (H,S)$.
	This contradicts the fact that $F$ is $2$-connected.
	
	We conclude that $B$ is an $S$-block.
	\end{proof}

	\begin{lemma}\label{lem:inducedpathsblock}
	Let $(G, S)$ and $(H,S)$ be two compatible $d$-labeled graphs such that
	each $S$-block of $G$ or $H$ is chordal, and 
	$\Part (G,S)\oplus \Part (H,S)$ has no cycles.
	If $F$ is an $S$-block of $(G,S)\oplus (H,S)$ and $uvw$ is an induced path in $F$ such that 
	$u$ and $w$ are not contained in the same $S$-block of $G$ or $H$, 
	then 
	\begin{enumerate}%
	\item $v\in S$, and 
	\item there is an induced path $q_1q_2 \cdots q_{\ell}$ from $u=q_1$ to $w=q_{\ell}$ in $F-v$
	such that each $q_i$ is a neighbor of $v$.
	\end{enumerate}
	\end{lemma}
	\begin{proof}
	Since $F$ contains at least $3$ vertices, $F$ is $2$-connected.
	Let $C$ be the component of $G$ containing $v$.
		
	\medskip
	(1) We verify that $v\in S$.
	Suppose $v\notin S$, and without loss of generality we assume $v\in V(G)\setminus S$.
	By Lemma~\ref{lem:edgesblock}, each of $uv$ and $vw$ is contained in some $S$-block of $G$.
	Moreover, since $u$ and $w$ are not contained in the same block, $v$ is a cut vertex of $G$.
	Let $H_1$ be the subgraph of $G$ induced by the union of $v$ and the component of $C-v$ containing $u$, 
	and let $H_2$ be the subgraph of $G$ induced by the union of $v$ and the component of $C-v$ containing $w$.
	Then $H_1$ and $H_2$ do not contain vertices from the same component of $G[S]$.
	This implies that $v$ separates $u$ and $w$ in $G$, and 
	since $\Part (G,S)\oplus \Part (H,S)$ has no cycles, 
	$v$ separates $u$ and $w$ in $(G, S)\oplus (H,S)$.
	This contradicts the assumption that $F$ is $2$-connected.
	Therefore, we have $v\in S$.

	\medskip
	(2) 
	Let $D$ be the component of $G[S]$ containing $v$.
	As $v\in V(D)$, for each $z\in \{u,w\}$, 
	we have either $z\in V(G)\setminus S$ or $z\in V(H)\setminus S$ or $z\in V(D)\setminus \{v\}$.
	
	\begin{claim}\label{claim:endpaths}
	For each $z\in \{u,w\}$, there is a path from $z$ to $V(D)\setminus \{v\}$ in $G-v$ or $H-v$.
	\end{claim}
	\begin{clproof}
	If $z\in V(D)\setminus \{v\}$, then this is clear.
	We assume $z\in V(G)\setminus S$; the symmetric argument works when $z\in V(H)\setminus S$. 
	Suppose for contradiction that there is no path from $z$ to $V(D)\setminus \{v\}$ in $G-v$.
	Then, $v$ is a cut vertex of $G$ separating $z$ from $D-v$.
	
	Let $H'$ be the component of $C-v$ containing $z$.
	If the other vertex in $\{u,w\}\setminus \{z\}$ is also contained in $H'$, then 
	there is a cycle formed with $v$ and a path from $u$ to $w$ in $H'$, 
	and thus $u,v,w$ are contained in the same block of $G$.
	Furthermore this block is an $S$-block by Lemma~\ref{lem:edgesblock}.
	This contradicts the assumption that $u$ and $w$ are not contained in the same $S$-block.
	Thus, $H'$ does not contain the other vertex in $\{u,w\}\setminus \{z\}$.
	
	Furthermore, since $\Part (G,S)\oplus \Part (H,S)$ has no cycles, 
	$v$ separates $u$ and $w$ in $(G,S)\oplus (H,S)$. 
	This contradicts the assumption that $F$ is $2$-connected.
	Therefore, there is a path from $z$ to $V(D)\setminus \{v\}$ in $G-v$.
	\end{clproof}
	
	Let $U'_1, \ldots, U'_p$ be the connected components of $D-v$, 
	and for each $i\in \{1, \ldots, p\}$, let $U_i:=G[V(U'_i)\cup \{v\}]$.
	Generally, we show the following.
	
\begin{figure}[t]
\centerline{
\includegraphics[scale=0.55]{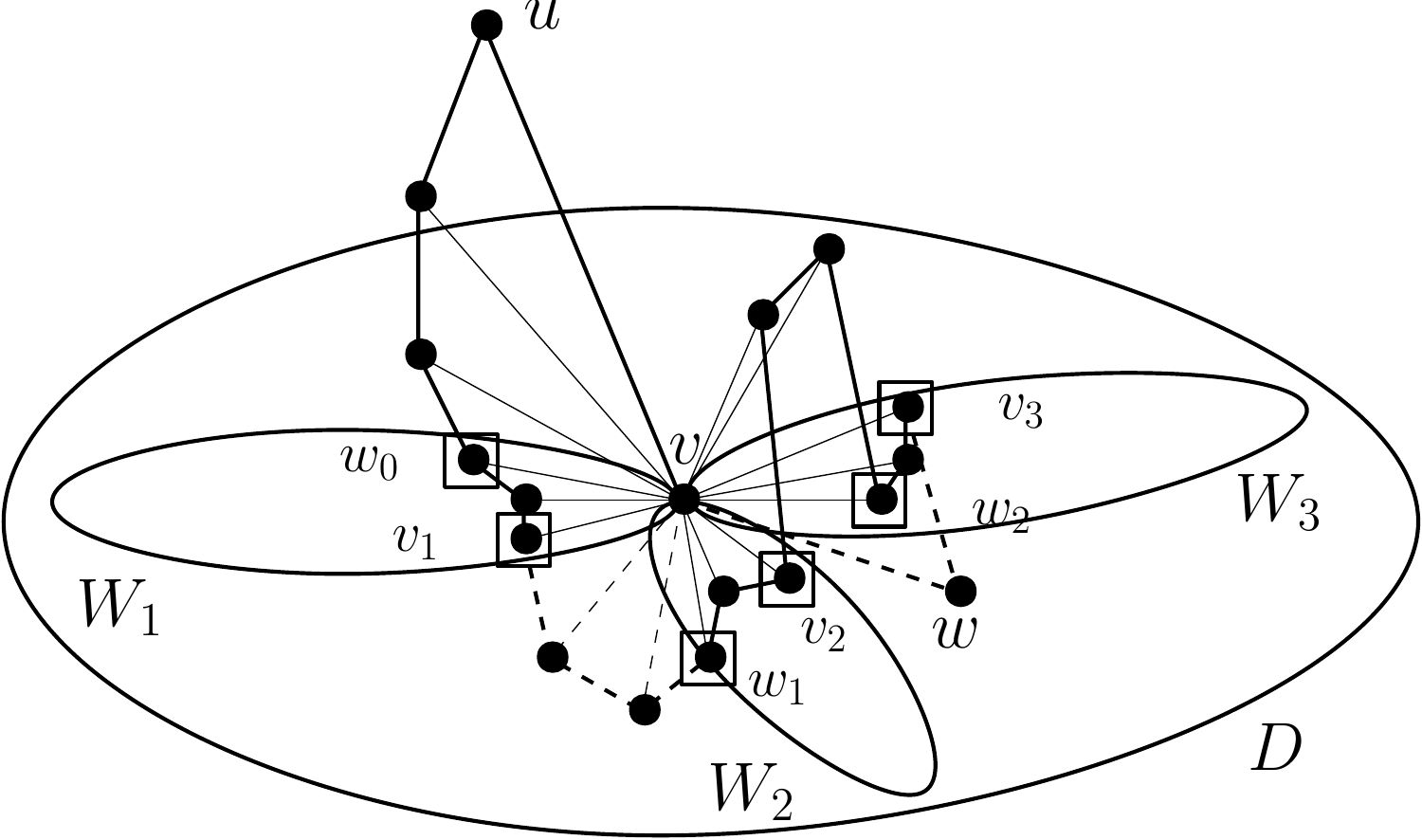}
}
\caption{The required path from $u$ to $w$ described in Lemma~\ref{lem:inducedpathsblock}. Dashed edges denote edges incident with vertices in $H-S$.}
\label{fig:blocksequence}
\end{figure}

	\begin{claim}\label{claim:blocksequence}
	There is a sequence $W_1-W_2- \cdots -W_m$ of distinct graphs in $\{U_1, \ldots, U_p\}$ such that
	\begin{itemize}
	\item there is a path from $u$ to $V(W_1)$ in $G-v$ or $H-v$, 
	\item there is a path from $w$ to $V(W_m)$ in $G-v$ or $H-v$, and
	\item if $m\ge 2$, then for each $i\in \{1, \ldots, m-1\}$, there is a path from $V(W_i)\setminus \{v\}$ to $V(W_{i+1})\setminus \{v\}$ in $G-v$ or $H-v$.
	\end{itemize}
	\end{claim}
\begin{clproof}
	Let $\mathcal{X}_u, \mathcal{X}_w\subseteq \{U_1, \ldots, U_p\}$ such that 
	\begin{itemize}
	\item for each $X\in \mathcal{X}_u$, there is a path from $u$ to $X$ in $(G, S)\oplus (H,S)-v$, 
	\item for each $X\in \mathcal{X}_w$, there is a path from $w$ to $X$ in $(G, S)\oplus (H,S)-v$. 
	\end{itemize}
 	By Claim~\ref{claim:endpaths}, $\mathcal{X}_u$ and $\mathcal{X}_w$ are non-empty.
	If $\mathcal{X}_u\cap \mathcal{X}_w\neq \emptyset$, then there is a required path.
	Suppose for contradiction that $\mathcal{X}_u\cap \mathcal{X}_w=\emptyset$.
	This implies that there is no path from components in $\mathcal{X}_u$ to components in $\mathcal{X}_w$ in $(G, S)\oplus (H,S)-v$, 
	and furthermore, there is no path from $u$ to $w$ in $(G, S)\oplus (H,S)-v$.
	This contradicts the fact that $F$ is $2$-connected.
\end{clproof}
			
	Now, we construct the required path. Fix a sequence $W_1-W_2- \cdots -W_m$ as obtained in Claim~\ref{claim:blocksequence}.
	Recall that the vertex set of each $W_i$ is contained in $S$.
	See Figure~\ref{fig:blocksequence} for an illustration.

	Let $P_0=z_1z_2 \cdots z_{\ell}$ be a path from $u=z_1$ to $w_0=z_{\ell}\in V(W_1)\setminus \{v\}$ in $G-v$ or $H-v$ such that
	\begin{enumerate}[(1)]
	\item $\ell$ is minimum,
	\item subject to (1), the distance from $w_0$ to $v$ in $W_1$ is minimum.
	\end{enumerate} 
	Let $R$ be a shortest path from $z_{\ell}$ to $v$ in $W_1$.
	As $G[V(P_0)\cup V(R)]$ is $2$-connected, it is contained in an $S$-block of $G$ or $H$, and by assumption, it is chordal.
	We claim that every vertex in $P_0$ is a neighbor of $v$.
	Suppose there exists $i\in \{2, \ldots, \ell-1\}$ such that $z_i$ is not adjacent to $v$.
	By the distance condition, there are no edges between $\{z_1, \ldots, z_{i-1}\}$ and $\{z_{i+1}, \ldots, z_{\ell}\}\cup (V(W_1)\setminus \{v\})$.
	Merging a shortest path from $z_i$ to $v$ in $G[\{z_1, \ldots, z_i\}\cup \{v\}]$ 
	and a shorest path from $z_i$ to $v$ in $G[\{z_i, \ldots, z_{\ell}\}\cup V(R)]$, 
	one can find a chordless cycle in $G[V(P_0)\cup V(R)]$; a contradiction.
	Therefore, every vertex in $V(P_0)\setminus \{z_{\ell}\}$ is a neighbor of $v$.
	Finally, by the assumption that the distance from $w_0$ to $v$ in $W_1$ is minimum, 
	$w_0$ is a neighbor of $v$; otherwise $G[\{z_{\ell-1}\}\cup V(R)]$ is a chordless cycle.
	Also, we can observe that every vertex in $P_0$ is in $F$.

	Similarly, let $P_m$ be a path from $w$ to $v_m\in V(W_m)\setminus \{v\}$ such that
	the length of $P_m$ is minimum, and subject to that, the distance from $v_m$ to $v$ in $W_m$ is minimum.
	Also, for each $i\in \{1,\ldots, m-1\}$, 
	let $P_i$ be the path from $v_i\in V(W_i)\setminus \{v\}$ to $w_i\in V(W_{i+1})\setminus \{v\}$ in $G-v$ or $H-v$ such that
	the length of $P_i$ is minimum, and subject to that,
	the sum of the distance from $v_i$ to $v$ in $W_i$ and the distance from $w_i$ to $v$ in $W_{i+1}$ is minimum.
	Lastly, 
	for each $i\in \{1, \ldots, m\}$, let $Q_i$ be a shortest path from $w_{i-1}$ to $v_i$ in $W_i-v$.
	Similar to $P_0$, 
	we can prove that every vertex of $Q_1\cup P_1\cup  \cdots \cup Q_m \cup P_m$ is a neighbor of $v$, 
	and is contained in $F$.
	Therefore, the shortest path from $u$ to $w$ in $P_0\cup Q_1\cup P_1\cup  \cdots \cup Q_m \cup P_m$ is the required path.
	\end{proof}

\begin{proof}[Proof of Proposition~\ref{prop:sumofchordalgraphs2}]
	Let $F$ be an $S$-block of $(G,S)\oplus (H,S)$. We need to show that $F$ is partially label-isomorphic to $Q$.
	If $F$ contains at most $2$ vertices, then it is contained in $G[S]$, and it is clearly partially label-isomorphic to $Q$.
	So we may assume $\abs{V(F)}\ge 3$, and thus $F$ is $2$-connected.

	Let $L_Q$ be the labeling of $Q$.
	Let $L_G$ and $L_H$ be labelings of $G$ and $H$, respectively, and $L:=L_G\oplus L_H$.
	By Lemma~\ref{lem:edgesblock}, every edge of $F$ is contained in some $S$-block of $G$ or $H$.
	This implies that for every edge $uv$ of $F$, we have $L(u)\neq L(v)$ and 
	the vertices with labels $L(u)$ and $L(v)$ are adjacent in $Q$.
	Moreover, since $(G,S)$ and $(H,S)$ are block-wise partially label-isomorphic to $Q$, 
	we have $L(V(F))\subseteq L_Q(V(Q))$.
	Let $\mu:V(F)\rightarrow V(Q)$ such that for each $v\in V(F)$, $L(v)=L_Q(\mu(v))$.

	To apply Lemma~\ref{lem:labelisomorphic}, it is sufficient to prove the following. Notice that we do not know yet whether $F$ is chordal or not.
	But since $Q$ is chordal, every $S$-block of $G$ is chordal, and also every $S$-block of $H$ is chordal.

	\begin{claim}\label{claim:inducedpath2}
	If $uvw$ is an induced path in $F$,
	then $L(u)\neq L(w)$ and $\mu(u)\mu(v)\mu(w)$ is an induced path in $Q$.
	\end{claim}
	\begin{clproof}
	First assume that 
	$u$ and $w$ are contained in an $S$-block of $G$ or $H$.
	We further assume that they are contained in an $S$-block of $G$, say $B_{uw}$.
	The symmetric argument holds when they are contained in an $S$-block of $H$.
	We claim that there is an $S$-block of $G$ or $H$ containing all of $u,v,w$.
	We divide into two cases.
	\begin{itemize}
		\item (Case 1. $v\in V(G)$.) 
		If $B_{uw}$ contains $v$, then we are done, so we may assume that $v\notin V(B_{uw})$.
		Let $P_{uw}$ be a path from $u$ to $w$ in $B_{uw}$. Note that 
		$P_{uw}$ and $v$ form a cycle of $G$.
		But this implies that $v$ is contained in $B_{uw}$; a contradiction. This proves the claim.
		
		\item (Case 2. $v\in V(H)\setminus S$.)
		In this case, $u$ and $w$ are contained in $S$. If $u$ and $w$ are contained in distinct connected components of $G[S]$, 
		then $\Part (G,S)\oplus \Part (H,S)$ contains a cycle of length $4$, because
		$u, w$ are contained in a connected component of each of $G$ and $H$.
		So, $u$ and $w$ are contained in the same connected component of $G[S]$.
		Let $P_{uw}$ be a path from $u$ to $w$ in $G[S]$.
		Then $P_{uw}$ and $v$ form a cycle in $H$, which implies that $u,v,w$ are contained in the same $S$-block of $H$.
	\end{itemize}
	Then, by the definition of partially label-isomorphic graphs, 
	$G[\{u,v,w\}]$ or $H[\{u,v,w\}]$ is isomorphic to $Q[\{ \mu(u), \mu(v), \mu(w)\}]$.
	This means that $\mu(u)\mu(v)\mu(w)$ is an induced path in $Q$ and the labels of $\mu(u)$ and $\mu(w)$ are distinct.

	Now, 
	we assume that $u$ and $w$ are not contained in the same $S$-block of $G$ or $H$.
	Recall that $\Part (G,S)\oplus \Part (H,S)$ contains no cycles, by the assumption.
	So, by Lemma~\ref{lem:inducedpathsblock}, $v\in S$ 
	and there is an induced path $q_1q_2 \cdots q_{\ell}$ from $u=q_1$ to $w=q_{\ell}$ in $F-v$
	such that each $q_i$ is a neighbor of $v$.

	We show that for each $i\in \{1, \ldots, \ell-2\}$, $L(q_i), L(q_{i+1}),  L(q_{i+2})$ are pairwise distinct, and $\mu(q_i)\mu(q_{i+1})\mu(q_{i+2})$ is an induced path of $Q$. 
	Let $i\in \{1, \ldots, \ell-2\}$.
	If all of $q_i, q_{i+1}, q_{i+2}$ are contained in $G$ or $H$, then they are contained in the same $S$-block with $v$, 
	and the claim follows.
	Thus, we may assume that one of $q_i$ and $q_{i+2}$ is contained in $G-S$, and the other one is contained in $H-S$.
	Then the $S$-block containing $q_i, q_{i+1}, v$ and the $S$-block containing $q_{i+1}, q_{i+2}, v$ share the edge $q_{i+1}v$.
	Since $(G,S)$ and $(H,S)$ are block-wise $Q$-compatible, 
	$L(q_i)\neq L(q_{i+2})$ and $\mu(q_i)$ is not adjacent to $\mu(q_{i+2})$ in $Q$. 

	We verify that $\mu(q_1)\mu(q_2)\cdots \mu(q_{\ell})$ is an induced path of $Q$.
	Suppose this is false, and choose $i_1, i_2\in \{1, 2, \ldots, \ell\}$ with $i_2-i_1>1$ and minimum $i_2-i_1$ 
	such that 
	$\mu(q_{i_1})$ is adjacent to $\mu(q_{i_2})$ in $Q$.
	By minimality, $\mu(q_{i_1})\cdots \mu(q_{i_2-1})$ and $\mu(q_{i_1+1})\cdots \mu(q_{i_2})$ are induced paths and have length at least $2$.
	Thus $\mu(q_{i_1})\cdots \mu(q_{i_2})$ is an induced cycle of length at least $4$, contradicting the assumption that $Q$ is chordal.
	Therefore, $\mu(q_1)\mu(q_2)\cdots \mu(q_{\ell})$ is an induced path of $Q$, and,
	in particular, $L(u)\neq L(w)$ and $\mu(u)$ and $\mu(w)$ are not adjacent in $Q$, as required.
	\end{clproof}	

	By Claim~\ref{claim:inducedpath2} and Lemma~\ref{lem:labelisomorphic},
	we conclude that $F$ is partially label-isomorphic to $Q$.
	\end{proof}

	Later, we will consider some information on non-trivial blocks of $G[S]$, 
	where two blocks in $G[S]$ contained in the same $S$-block of $G$ or $H$ have the same information.
	In Lemma~\ref{lem:gvlauesblock}, 
	we analyze when this property is preserved after taking the sum of $(G,S)$ and $(H,S)$.

\begin{lemma}\label{lem:gvlauesblock}
	Let $A$ be a set. 
	Let $(G, S)$ and $(H, S)$ be two compatible $d$-labeled graphs,
	$\cB$ be the set of non-trivial blocks in $G[S]$, and
	$g:\cB\rightarrow A$ be a function
	 such that
	\begin{itemize}
	\item each $S$-block of $G$ or $H$ is chordal, 
	\item $\Part (G,S)\oplus \Part(H,S)$ has no cycles, and
	\item  for every $B_1, B_2\in \cB$ where $B_1$ and $B_2$ are contained in an $S$-block of $G$ or $H$, $g(B_1)=g(B_2)$.
	\end{itemize}
	If $F$ is an $S$-block of $(G,S)\oplus (H,S)$ and $B_1, B_2\in \cB$ where $V(B_1), V(B_2)\subseteq V(F)$, then 
	$g(B_1)=g(B_2)$.
\end{lemma}
	\begin{proof}
	By Lemma~\ref{lem:edgesblock}, every edge of $F$ is contained in an $S$-block of $G$ or $H$.
	We define a function $g':E(F)\rightarrow A$ such that for each $vw\in E(F)$, 
	$g'(vw)=g(B)$ where $B\in \cB$ and $B$ is contained in the $S$-block of $G$ or $H$ containing $v$ and $w$.
	We claim that $g'(e)=g'(f)$ for all $e,f\in E(F)$.

	\begin{claim}\label{claim:propagate}
	$g'(e)=g'(f)$ for all $e,f\in E(F)$.
	\end{claim}
	\begin{clproof}
	Suppose towards a contradiction that there are $e,f\in E(F)$ such that $e$ and $f$ share a vertex and $g'(e)\neq g'(f)$.
	Let $e=uv$ and $f=vw$. Then $u,v,w$ are not contained in the same $S$-block of $G$ or $H$ as $g'(e)\neq g'(f)$.
	Also, this implies that $u$ is not adjacent to $w$.
	Thus by Lemma~\ref{lem:inducedpathsblock}, $v\in S$, 
	and there is an induced path $q_1q_2 \cdots q_{\ell}$ from $u=q_1$ to $w=q_{\ell}$ in $F-v$
	such that each $q_i$ is a neighbor of $v$.
	
	As $q_1, q_2, v$ are contained in the same $S$-block of $G$ or $H$, we observe that $g'(q_1q_2)=g'(q_1v)=g'(uv)$.
	Similarly, we have $g'(q_{\ell-1}q_{\ell})=g'(q_{\ell}w)=g'(vw)$.

	We claim that for each $i\in \{1, \ldots, \ell-2\}$, 
	$g'(q_iq_{i+1})=g'(q_{i+1}q_{i+2})$. 
	Let $i\in \{1, \ldots, \ell-2\}$.
	If $\{q_i, q_{i+1}, q_{i+2}\}\subseteq V(G)$ or $\{q_i, q_{i+1}, q_{i+2}\}\subseteq V(H)$, 
	then $q_i, q_{i+1}, q_{i+2}$ are contained in the same $S$-block with $v$, 
	and the claim follows.
	We may assume that one of $q_i$ and $q_{i+2}$ is contained in $G-S$ and the other one is contained in $H-S$.
	In this case, the $S$-block containing $q_i, q_{i+1}, v$ and the $S$-block containing $q_{i+1}, q_{i+2}, v$ share the edge $q_{i+1}v$, and we have $g'(q_iq_{i+1})=g'(q_{i+1}v)=g'(q_{i+1}q_{i+2})$.
	Therefore, $g'(uv)=g'(q_1q_2)=g'(q_{\ell-1}q_{\ell})=g'(vw)$, which is a contradiction.

	We conclude that $g'(e)=g'(f)$ for all $e,f\in E(F)$, as required.
	\end{clproof}
	
	Now, for each $i\in \{1,2\}$, we choose an edge $u_iv_i$ in $B_i$.
	By Claim~\ref{claim:propagate}, we have $g(B_1)=g'(u_1v_1)=g'(u_2v_2)=g(B_2)$.
		\end{proof}

	We also need the following lemma.
	\begin{lemma}\label{lem:nocycle}
	Let $(G, S)$ and $(H, S)$ be two compatible $d$-labeled graphs such that
	$\Part (G,S)\oplus \Part (H,S)$ has no cycles.
	If $F$ is an $S$-block of $(G,S)\oplus (H,S)$,
	then $\Part (F\cap G,S\cap V(F))\oplus \Part (F\cap H,S\cap V(F))$ has no cycles.
	\end{lemma}
	\begin{proof}
	Let $S_F:=S\cap V(F)$.
	Suppose towards a contradiction that $\Part (F\cap G,S_F)\oplus \Part (F\cap H,S_F)$ has a cycle $C_1-F_1-\cdots -C_m-F_m-C_1$, 
	where $C_1, \ldots, C_m$ are components of $F[S_F]$.
	
	First assume that there are two distinct components $C_i, C_j\in \{C_1, \ldots, C_m\}$ contained in the same component of $G[S]$.
	We choose such components $C_i, C_j$ such that 
	the distance between $C_i$ and $C_j$ in the cycle  $C_1-F_1-\cdots -C_m-F_m-C_1$ is minimum.
	By relabeling if necessary, we may assume that $i<j$ and in the sequence $C_i, C_{i+1}, \ldots, C_j$, there are no two components 
	contained in the same component of $G[S]$ except the pair $(C_i, C_j)$.

	We claim that all of $C_i, F_i, C_{i+1}, F_{i+1}, \ldots, C_j$ are contained in the same component of $G$ or $H$.
	Without loss of generality, we assume that $F_i$ is contained in $G$.

	Note that $\Part (G,S)\oplus \Part (H,S)$ has no cycles. 
	So if there is $C_{i_1}$ for some $i<i_1\le j$ where 
	$C_{i_1}$ and $C_i$ are not contained in the same component of $G$ or $H$, 
	then there exists $i_1<i_2\le j$ where $C_{i_2}$ and $C_{i_1}$ are contained in the same connected component of $G[S]$.
	But this contradicts the assumption that $C_i$ and $C_j$ are contained in the same connected component of $G[S]$ where the distance 
	between $C_i$ and $C_j$ in the cycle $C_1-F_1-\cdots -C_m-F_m-C_1$ is minimum.
	Also, if $F_{i'}$ is contained in $H$ for some $i<i'$, then there exists $i'<i''$ such that 
	$C_{i'}$ and $C_{i''}$ are contained in the same connected component of $G[S]$; a contradiction.
	Therefore, all of $C_i, F_i, C_{i+1}, F_{i+1}, \ldots, C_j$ are contained in the same component of $G$.

    This implies that $j=i+1$; because all these subgraphs are connected to each other in $F\cap G$.
	Let $P$ be a path from $V(C_i)$ to $V(C_{i+1})$ in $F_i$ with endpoints $x$ and $y$, 
	and $Q$ be a path from $x$ to $y$ in $G[S]$.
	Then $P\cup Q$ is a cycle containing $x$ and $y$, and the existence of this cycle implies that $V(P)\cup V(Q)\subseteq V(F)$, 
	as $F$ is a block of $(G,S)\oplus (H,S)$.
	But this implies that $C_i$ and $C_{i+1}$ are contained in the same connected component of $F[S_F]$; a contradiction.
	We conclude that there are no two distinct components $C_i$ and $C_j$ contained in the same component of $G[S]$.	

	We observe that all of $C_1, \ldots, C_m$ are contained in the same component of $G$ or $H$ since 
	there are no two distinct components $C_i$ and $C_j$ contained in the same component of $G[S]$.
	This implies that $C_1, \ldots, C_m$ are contained in the same component of $F\cap G$ or $F\cap H$. 
	This contradicts the assumption that 
	$C_1-F_1-\cdots -C_m-F_m-C_1$ is a cycle.
	\end{proof}

	Lastly, we show that when every $S$-block of $(G,S)\oplus (H,S)$ is chordal, 
	 $(G,S)\oplus (H,S)$ is chordal if and only if	$\Part (G,S)\oplus \Part (H,S)$ has no cycles.

\begin{proposition}\label{prop:nocircuit}
	Let $(G,S)$ and $(H,S)$ be two compatible graphs such that 
	every $S$-block of $(G,S)\oplus (H,S)$ is chordal.
	The following are equivalent:
\begin{enumerate}
  \item $(G, S)\oplus (H,S)$ is chordal.\label{propcase2}
  \item $\Part (G,S)\oplus \Part (H,S)$ has no cycles.\label{propcase3}
\end{enumerate}
\end{proposition}
\begin{proof}
Let $\cC$ be the set of components of $G[S]$.

\begin{figure}[t]
\centerline{
\includegraphics[scale=0.4]{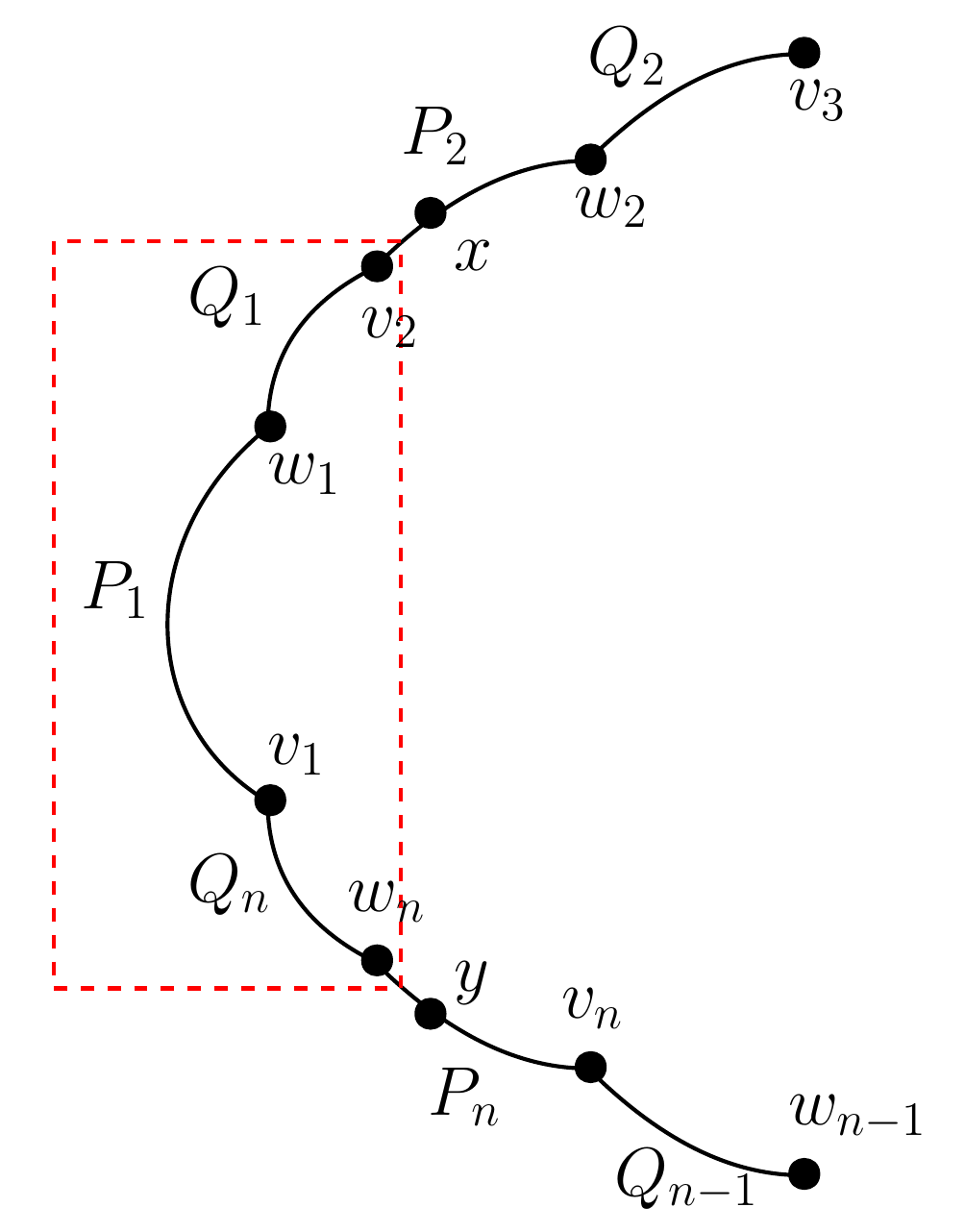}
}
\caption{Finding a chordless cycle in Proposition~\ref{prop:nocircuit}.}
\label{fig:findingcycle}
\end{figure}

\medskip
(\ref{propcase2} $\Rightarrow$ \ref{propcase3}). 
	Suppose that $\Part (G,S)\oplus \Part (H,S)$ has a cycle 
	$C_1-A_1-C_2-A_2-\cdots -C_n-A_n-C_1$
	where 
	$C_1, \ldots, C_n\in \cC$.
	For convenience, let $C_{n+1}:=C_1$ and $A_{n+1}:=A_1$.

	We construct an induced cycle of length at least $4$ in $(G,S)\oplus (H,S)$.
	For each $i\in \{1, \ldots, n\}$, 
	we define that 
	\begin{itemize}
	\item $P_i$ is the shortest path from $C_i$ to $C_{i+1}$ in $A_i$,
	\item $v_i, w_i$ are the end vertices of $P_i$ where $v_i\in V(C_i)$ and $w_i\in V(C_{i+1})$.
	\item $Q_i$ is the shortest path from $w_i$ to $v_{i+1}$ in $C_{i+1}$.
	\end{itemize}
	Note that $n\ge 2$. We consider two cases depending on whether $n=2$ or not.

	Suppose $n=2$. 
	Notice that $A_1$ and $A_2$ may share several components of $G[S]$.
	We choose $C_1,C_2, P_1, P_2, Q_1, Q_2$ such that the cycle $P_1\cup Q_1\cup P_2\cup Q_2$ passes 
	the minimum number of components of $G[S]$.
	This minimality implies that 
	$C_1$ and $C_2$ are the only components of $G[S]$ 
	that contain vertices of both $P_1$ and $P_2$, and there are no edges between the internal vertices of $P_1$ and the internal vertices of $P_2$.
	Therefore, $P_1\cup Q_1\cup P_2\cup Q_2$ contains a chordless cycle.

	Now, assume that $n\ge 3$.	
	In this case, $v_1-P_1-Q_1-P_2-Q_2-\cdots -P_n-Q_n-v_1$ is a cycle in $(G,S)\oplus (H,S)$, but is not necessarily a chordless cycle.
	Call this cycle $C$. We claim that $C$ contains a chordless cycle.
	Let $x$ be the vertex following $v_2$ in $P_2$, and $y$ be the vertex preceding $w_n$ in $P_n$.
	See \cref{fig:findingcycle} for an illustration.
	Take a shortest path $P$ from $x$ to $y$ in the path $y-Q_n-P_1-Q_1-x$. 
	Clearly $P$ has length at least $2$, as $x$ and $y$ are contained in distinct components of $Q$.
	Also, every internal vertex of $P$ has no neighbors in the other path of the cycle $v_1-P_1-Q_1-P_2-Q_2-\cdots -P_n-Q_n-v_1$ 
	between $x$ and $y$.
	So, if we take a shortest path $P'$ from $x$ to $y$ along the other part of the cycle $v_1-P_1-Q_1-P_2-Q_2-\cdots -P_n-Q_n-v_1$, 
	then $P\cup P'$ is a chordless cycle. This proves the claim. 

\medskip
(\ref{propcase3} $\Rightarrow$ \ref{propcase2}).
Suppose, towards a contradiction, that $(G, S)\oplus (H, S)$ contains a chordless cycle $C$.
Since $G$ and $H$ are chordal,
$C$ should contain a vertex of $G-S$ and a vertex of $H-S$. 
By assumption, we know that every $S$-block of $(G, S)\oplus (H, S)$ is chordal.
Thus, $C$ can contain at most one vertex from each $S$-block of $(G, S)\oplus (H,S)$.
Furthermore, we can observe that $\abs{V(C)\cap V(F)}\le 1$ for every component $F$ of $G[S]$;
otherwise one of $S$-blocks of $(G,S)\oplus (H,S)$ should contain all vertices of $C$, contradicting the fact that every $S$-block is chordal.

	Let $C_1-C_2- \cdots- C_n-C_1$ be the sequence of components of $G[S]$ such that
	\begin{enumerate}
	\item for each $v\in V(C)\cap V(C_i)$, one neighbor of $v$ in $C$ is contained in $G-S$ and the other is contained in $H-S$, and
	\item $C$ passes through the components of $G[S]$ in this order.
	\end{enumerate}
	As $C$ contains at least one vertex of $G-S$ and one vertex of $H-S$, 
	such a sequence exists, and $n\ge 2$.
	Without loss of generality, we may assume that the internal vertices in the path from $C_1$ to $C_2$ (corresponding to the first part of the sequence) are contained in $G$.
	Then, the internal vertices in the path from $C_2$ to $C_3$ are contained in $H$, 
	and we use parts of $G-S$ and $H-S$ alternately.
	For each $i$, pick $A_i\in V(\Part(G,S)\oplus \Part(H,S))\setminus \cC$ corresponding to a component of $G$ or $H$ 
	containing the internal vertices of the path from $C_i$ to $C_{i+1}$.	
	Then $C_1-A_1-C_2-A_2-\cdots -C_n-A_n-C_1$ contains a cycle of $\Part (G,S)\oplus \Part (H,S)$.
	\end{proof}

\section{Representative sets for acyclicity}\label{sec:representative}

	In our algorithm, we need to store auxiliary graphs $\Part (G,S)$ for boundaried graphs $(G,S)$.
	Instead of working with $\Part (G,S)$, we work with 
	the partition of the set $\cC$ of components of $G[S]$, where $C_1, C_2\in \cC$ are in the same part 
	if and only if they are contained in the same component of $G$.
	This formulation has the advantage that it is convenient for applying representative-set techniques.
	
	For a set $S$ and a family $\cX$ of subsets of $S$,
	we define $\Inc (S, \cX)$ as the bipartite graph on the bipartition $(S, \cX)$ such that 
	for $v\in S$ and $X\in \cX$ with $v\in X$, $v$ and $X$ are adjacent in $\Inc (S, \cX)$.
	Let $S$ be a set, and $\cA$ be a set of partitions of $S$.
	A subset $\cA'$ of $\cA$ is called a \emph{representative set} if 
	\begin{itemize}
	\item for every $\cX_1\in \cA$ and every partition $\cY$ of $S$ where $\Inc (S, \cX_1\cup \cY)$ has no cycles, 
	there exists a partition $\cX_2\in \cA'$ such that $\Inc (S, \cX_2\cup \cY)$ has no cycles.
	\end{itemize} 
	Computing a representative set for a family of partitions is an essential part of our algorithm.
	To apply the ideas in \cite{BodlaenderCKN2015}, 
	it is necessary to translate our problem to finding a pair of partitions $\cX_1, \cX_2$ where $\Inc (S, \cX_1\cup \cX_2)$ is connected.

For partitions $\cX_1$ and $\cX_2$ of a set~$S$, 
$\cX_1$ is a \emph{coarsening} of $\cX_2$ if every two elements in the same part of $\cX_2$ are in the same part of $\cX_1$.
We denote by $\cX_1\uplus \cX_2$ the common coarsening of $\cX_1$ and $\cX_2$ with the maximum number of parts.
For instance, if $\cX_1=\{\{1\}, \{2,3\}, \{4\}\}$ and $\cX_2=\{\{1,2\}, \{3\}, \{4\} \}$, then 
both $\{\{1,2,3\}, \{4\}\}$ and $\{\{1,2,3,4\}\}$ are common coarsenings of $\cX_1$ and $\cX_2$, 
and $\cX_1\uplus \cX_2=\{\{1,2,3\}, \{4\}\}$.

	\begin{lemma}\label{lem:connectedhypergraph}
	Let $S$ be a set and $\cX_1, \cX_2$ be two partitions of $S$
	such that $\Inc (S, \cX_1\cup \cX_2)$ is connected.	
	Then $\Inc (S, \cX_1\cup \cX_2)$ has no cycles if and only if
	$\abs{\cX_1}+\abs{\cX_2}=\abs{S}+1$.
	\end{lemma}

	\begin{proof}
	Let $H:=\Inc (S, \cX_1\cup \cX_2)$.
	The result follows from the fact that $\abs{V(H)}=\abs{S}+\abs{\cX_1}+\abs{\cX_2}$, $\abs{E(H)}=2\abs{S}$, and
	a connected graph $H$ has no cycles if and only if $\abs{E(H)}=\abs{V(H)}-1$. 
	\end{proof}

	For a set $S$ and a partition $\cX$ of $S$, 
	a partition $\cY$ of $S$ is called a \emph{$1$-coarsening} of $\cX$ if 
 	$\cY=\cX\setminus \{X_1, \ldots, X_m\}\cup  \{X_1\cup \cdots \cup X_m\}$ for some $X_1, \ldots, X_m\in \cX$. 
 	Notice that the partition $\cX$ itself is a $1$-coarsening of $\cX$.
	We will use the following observation. For two partitions $\cX_1, \cX_2$ of a set $S$, the following are equivalent:	
	\begin{itemize}
	\item $\Inc (S, \cX_1\cup \cX_2)$ has no cycles. 
	\item There exists a $1$-coarsening $\cX_1'$ of $\cX_1$ such that
	$\Inc (S, \cX_1'\cup \cX_2)$ is connected and has no cycles.
	\end{itemize}
	Such a $1$-coarsening $\cX_1'$ can be obtained by taking one part of $\cX_1$ for each component of  $\Inc (S, \cX_1\cup \cX_2)$
	and unifying them into one part. Since the vertex corresponding to the new part of $\cX_1'$ would be a cut vertex of $\Inc (S, \cX_1'\cup \cX_2)$,
	there will not be an additional cycle in $\Inc (S, \cX_1'\cup \cX_2)$ while it is connected.

\begin{theorem}[\cite{BodlaenderCKN2015}; See also Theorem 11.11 in \cite{Cygan15}]\label{thm:representativeset}
Given two families of partitions $\mathcal{A}$, $\mathcal{B}$ of a set $S$, 
one can, in time $\mathcal{A}^{\mathcal{O}(1)} 2^{\mathcal{O}(\abs{S})}$, find a set $\mathcal{A}'\subseteq \mathcal{A}$ of size at most $2^{\abs{S}-1}$ such that
for every $\cX_1\in \mathcal{A}$ and every $\cY\in \cB$ such that $\Inc (S, \cX_1\cup  \cY)$ is connected,
there exists $\cX_2\in \mathcal{A}'$ such that $\Inc (S, \cX_2\cup \cY)$ is connected. 
\end{theorem}

We explicitly describe a necessary subroutine, Algorithm~\ref{alg:reppartitions}.

\begin{algorithm}
  \caption{\textsc{RepPartitions($S, \mathcal{A}$)}}\label{alg:reppartitions}
\begin{algorithmic}[1]
\Statex \textbf{Input:} A set $S$ and a family $\mathcal{A}$ of partitions of $S$.
\Statex \textbf{Output:} A representative set $\mathcal{R}$ of $\mathcal{A}$ of size at most $\abs{S}\cdot 2^{\abs{S}-1}$.
\State We compute the family $\mathcal{A}'$ of all $1$-coarsenings of partitions in $\mathcal{A}$.
\State For each $1\le i\le \abs{S}$, 
set $\cA_i:=\{\cX\in \cA': \abs{\cX}=i\}$ and $\cB_i$ the set of all partitions of $S$ of size $i$.
\State For each $1\le i, j\le \abs{S}$ with $i+j=\abs{S}+1$, we compute a set $\cR_i$ from $\cA_i$ with respect to $\cB_j$ using Theorem~\ref{thm:representativeset}.
\State We take the set $\cR$ from $\bigcup_{1\le i\le \abs{S}}\cR_i$ by taking the original partition before taking a $1$-coarsening, and output $\cR$.
\end{algorithmic}
\end{algorithm}

\begin{proposition}\label{prop:repalgorithm}
Given a family $\cA$ of partitions of a set $S$,
Algorithm~\ref{alg:reppartitions} outputs a representative set of $\cA$ of size at most $\abs{S}\cdot 2^{\abs{S}-1}$ in time $\mathcal{A}^{\mathcal{O}(1)} 2^{\mathcal{O}(\abs{S})}$.
\end{proposition}

\begin{proof}
Let $\cR$ be the output of Algorithm~\ref{alg:reppartitions}. 
Clearly, $\cR\subseteq \cA$, because we take the original partitions of $\bigcup_{1\le i\le \abs{S}}\cR_i$ at the last step.
Thus, it is sufficient to show that
\begin{itemize}
\item for every $\cX_1\in \cA$ and every partition $\cY$ of $S$ where $\Inc (S, \cX_1\cup \cY)$ has no cycles, 
there exists a partition $\cX_2\in \cR$ such that $\Inc (S, \cX_2\cup \cY)$ has no cycles.
\end{itemize} 

To show this, let $\cX_1\in \cA$ and $\cY$ be partitions of $S$ such that $\Inc (S, \cX_1\cup \cY)$ has no cycles.
We know that there exists a $1$-coarsening $\cX_2$ of $\cX_1$ such that $\Inc (S, \cX_2\cup \cY)$ is connected and has no cycles.
This $1$-coarsening $\cX_2$ is obtained in Step 1.
In Step 3, we obtain $\cR_{\abs{\cX_2}}$, and there exists $\cX_3\in \cR_{\abs{\cX_2}}$ such that  $\Inc (S, \cX_3\cup \cY)$ is connected and has no cycles.
Let $\cX_4$ be the partition obtained from $\cX_3$ by taking the original partition before taking a $1$-coarsening.
We have that $\cX_4\in \cR$ and $\Inc (S, \cX_4\cup \cY)$ has no cycles, as required.
By Theorem~\ref{thm:representativeset}, $\abs{\cR}\le \sum_{1\le i\le \abs{S}} \abs{\cR_{i}}\le \abs{S}\cdot 2^{\abs{S}-1}$ and 
Algorithm~\ref{alg:reppartitions} runs in time $\mathcal{A}^{\mathcal{O}(1)} 2^{\mathcal{O}(\abs{S})}$.
\end{proof}

\section{\BPBVD}\label{sec:BPBVD}

In this section, we prove \cref{thm:main1b}, restated below.

\mainthmb*

	We provide an overview of our approach for \cref{thm:main1b}.

	\begin{enumerate}
	\item Let $(G,S)$ be a $d$-labeled $\cP$-block graph, which will be the graph that remains after removing some partial solution in the dynamic programming algorithm.
	We first focus on $S$-blocks of $(G,S)$.	
	For each non-trivial block of $G[S]$, we guess its final shape as a $d$-labeled biconnected graph, and
	store the labelings of the vertices and their neighbors in the $S$-block of $G$ containing it.
	Collectively, we call this information a \emph{characteristic} of $(G,S)$.
	\item Suppose $(H,S)$ is a $d$-labeled $\cP$-block boundaried graph compatible with $(G,S)$ such that  
	every $S$-block of $(G,S)\oplus (H,S)$ is a $d$-labeled $\cP$-block graph. 
	Note that $(G,S)\oplus (H,S)$ still may have a chordless cycle, and
	by Proposition~\ref{prop:nocircuit}, 
	$(G,S)\oplus (H,S)$ is chordal if and only if
	$\Part (G,S)\oplus \Part (H,S)$ has no cycles.
	If $(G,S)\oplus (H,S)$ is chordal, then it is easy to check that 
	for every block $B$ of $(G,S)\oplus (H,S)$,
	either $B$ is contained in one of $G$ and $H$, or it is an $S$-block.
    Thus, instead of storing $\Part (G,S)$, we will store the corresponding partition of the set of components of $G[S]$.  
	To avoid storing all such partitions, whose total size might be $2^{c\cdot w\log w}$ for some constant $c$, we use the representative set technique discussed in Section~\ref{sec:representative}.
	\item We formally describe and prove an equivalence between two boundaried graphs in Theorem~\ref{thm:equivalence2}.
	\end{enumerate}

	For convenience, we fix an integer $d\ge 2$ and a class $\cP$ of graphs that is block-hereditary, recognizable in polynomial time, 
	and consists of only chordal graphs.	
	Let $\cU_d$ be the set of all $d$-labeled biconnected $\cP$-block graphs.
	For a boundaried graph $(G,S)$, we denote by $\Block (G,S)$ the set of all non-trivial blocks in $G[S]$.

\subsection{Characteristics}\label{subsec:bcharacteristic}

	For a $d$-labeled graph $(G, S)$ with labeling $L$, 
	a \emph{characteristic} of $(G, S)$ is a pair $(g, h)$ of functions 
	$g:\Block (G,S) \rightarrow \cU_d$ and $h:\Block (G,S) \rightarrow 2^{[d]}$
	satisfying the following:
	for each $B\in \Block (G,S)$ and the unique $S$-block $X$ of $G$ containing $B$, 
	\begin{enumerate}[(a)]
	\item (label-isomorphism condition) $X$ is partially label-isomorphic to $g(B)$;
	\item (coincidence condition) for every $B'\in \Block (G,S)$ contained in $X$, $g(B')=g(B)$;
	\item (neighborhood condition) $h(B)=L(N_{X}(V(B))\setminus S)$; and
	\item (completeness condition) for every $w$ where $w\in V(X)\setminus S$ or $\{w\}=V(X)\cap V(C)$ for some component $C$ of $G[S]$, 
	$X[N_X[w]]$ is label-isomorphic to $g(B)[N_{g(B)}[z]]$ where $z$ is the vertex in $g(B)$ with label $L(w)$.
	\end{enumerate}
	The conditions (a) and (c) were motivated in the overview.
	Since we want that $g(B)$ is a final block containing $B$, 
	for other non-trivial block $B'$ of $G[S]$ already contained in the same block of $G$ with $B$, 
	it has to indicate the same final block. This is the condition (b).
	If we just say that (a) $X$ is partially label-isomorphic to $g(B)$,
	some vertex of $X$ may have unexpected neighbor. To avoid this problem, we impose the last condition (d).

	For a $d$-labeled $\cP$-block graph $(G,S)$ with characteristic $(g,h)$ and a $d$-labeled $\cP$-block graph $(H,S)$ compatible with $(G,S)$,
	the sum $(G,S)\oplus (H,S)$ \emph{respects} $(g, h)$  
	if for each $B\in \Block (G,S)$, the $S$-block of $(G,S)\oplus (H,S)$ containing $B$ 
	is label-isomorphic to $g(B)$.

	The following is the main combinatorial result regarding characteristics.
\begin{theorem}\label{thm:equivalence2}
Let $(G_1, S)$, $(G_2, S)$, and $(H, S)$ be $d$-labeled $\cP$-block graphs such that
\begin{itemize}
\item for each $i\in \{1,2\}$, $(G_i, S)$ is compatible with $(H, S)$, 
\item $(G_1, S)$ and $(G_2, S)$ have the same characteristic $(g,h)$, and 
\item $\Part (G_2,S)\oplus \Part (H,S)$ has no cycles. 
 \end{itemize}
If $(G_1, S)\oplus (H, S)$ is  a $d$-labeled $\cP$-block graph that respects $(g,h)$, 
then $(G_2, S)\oplus (H, S)$ is a $d$-labeled $\cP$-block graph that respects $(g,h)$.
\end{theorem}

\begin{proof}
	Suppose $(G_1, S)\oplus (H, S)$ is a $d$-labeled $\cP$-block graph that respects $(g,h)$.
	We first show $(G_2, S)\oplus (H, S)$ respects $(g,h)$.
	Choose a non-trivial block $B$ of $G_2[S]$, let $Q:=g(B)$, 
	and let $F$ be the $S$-block of $(G_2, S)\oplus (H, S)$ containing $B$.
	As a shortcut, set $S_{F}:=V(F)\cap S$.
	Let $L_F$ be the function from $V(F)$ to $[d]$ that sends each vertex to its label from $G_2$ or $H$.
	Let $L_Q$ be the labeling of $Q$.
	
	We first show that $L_F$ is a $d$-labeling of $F$, and $F$ is partially label-isomorphic to $Q$.
	We verify the conditions of Proposition~\ref{prop:sumofchordalgraphs2} by regarding $F$ as the sum of $(F\cap G_2, S_F)$ and $(F\cap H, S_F)$
	to show that $F$ is partially label-isomorphic to $Q$.
	We additionally show that $L_Q(V(Q))\subseteq L_F(V(F))$, in order to complete the proof.

	\begin{claim}\label{claim:equivalenceblock}
	For every non-trivial block $B'$ of $G_2[S]$ with $V(B')\subseteq V(F)$, $g(B')=Q$.
	\end{claim}
	\begin{clproof}
	Note that $\Part (G_2,S)\oplus \Part (H,S)$ has no cycles.
	Since $(g,h)$ is a characteristic of $(G_2, S)$, 
	for non-trivial blocks $B_1, B_2$ of $G_2[S]$ contained in the same $S$-block of $G_2$, $g(B_1)=g(B_2)$.
	Also, since $(G_1, S)\oplus (H,S)$ respects $(g,h)$, 
	for non-trivial blocks $B_1, B_2$ of $G_2[S]$ contained in the same $S$-block of $H$, $g(B_1)=g(B_2)$.
	Thus, the claim follows from Lemma~\ref{lem:gvlauesblock}.
	\end{clproof}

Since $\Part (G_2,S)\oplus  \Part (H,S)$ has no cycles,
by Lemma~\ref{lem:nocycle}, 
$\Part (F\cap G_2,S_F)\oplus \Part (F\cap H,S_F)$ has no cycles.
To apply Proposition~\ref{prop:sumofchordalgraphs2}, it remains to show that $(F\cap G_2,S_F)$ and $(F\cap H,S_F)$ are block-wise $Q$-compatible.

	\begin{claim}\label{claim:partial2} 
	$(F\cap G_2, S_F)$ and $(F\cap H, S_F)$ are block-wise $Q$-compatible.
	\end{claim}
	\begin{clproof}
	By Claim~\ref{claim:equivalenceblock} and the fact that $(g,h)$ is a characteristic of $(G_2, S)$, $F\cap G_2$ is block-wise partially label-isomorphic to $Q$.
	By Claim~\ref{claim:equivalenceblock} and the fact that $(G_1, S)\oplus (H,S)$ respects $(g,h)$, $F\cap H$ is block-wise partially label-isomorphic to $Q$.

	We now confirm the second condition of being block-wise $Q$-compatible.
	Let $B\in \Block (F, S_F)$.
	Let $B_1$ be the $S$-block of $G_2$ containing $B$,
	$B_2$ be the $S$-block of $H$ containing $B$, and 
	$B_1'$ be the $S$-block of $G_1$ containing $B$.
	
	Since $(G_1, S)\oplus (H,S)$ respects $(g,h)$, $N_{B_1'}(V(B))\setminus S$ and $N_{B_2}(V(B))\setminus S$ have disjoint sets of labels.
	As $(G_1, S)$ and $(G_2, S)$ have the same characteristic, 
	$N_{B_1'}(V(B))\setminus S$ and $N_{B_1}(V(B))\setminus S$ have the same set of labels, and thus
	$N_{B_1}(V(B))\setminus S$ and $N_{B_2}(V(B))\setminus S$ have disjoint sets of labels.
	Furthermore, for every $\ell_1\in L_F(N_{B_1}(V(B))\setminus S)$ and every $\ell_2 \in L_F(N_{B_2}(V(B))\setminus S)$, 
	the vertices in $Q$ with labels $\ell_1$ and $\ell_2$ are not adjacent because there are no edges between 
	$N_{B_1'}(V(B))\setminus S$ and $N_{B_2}(V(B))\setminus S$ in $(G_1, S)\oplus (H,S)$.
	\end{clproof}

	By Claim~\ref{claim:partial2} and Proposition~\ref{prop:sumofchordalgraphs2}, 
	$L_F$ is a $d$-labeling of $F$ and $F$ is partially label-isomorphic to $Q$.
	Lastly, we show that $F$ and $Q$ have the same set of labels.

	\begin{claim}
	$L_Q(V(Q))\subseteq L_F(V(F))$.
	\end{claim}
	\begin{clproof}
	Suppose there is a vertex $v$ in $Q$ such that $F$ has no vertex with label $L_Q(v)$.
	We choose such a vertex $v$ so that there exists $w\in V(Q)$ that is adjacent to $v$ in $Q$ where the label of $w$ appears in $F$.
	We can choose such vertices $v$ and $w$ because $Q$ is connected, $V(F)\neq \emptyset$, and $L_F(V(F))\subseteq L_Q(V(Q))$. 
	Let $w'$ be the vertex in $F$ with label $L_Q(w)$.

	First assume $w'\in V(F)\setminus S$.
	If $w'\in V(G_2)\setminus S$, 
	then by the completeness condition of the characteristic, 
	$U[N_U[w']]$ is label-isomorphic to $Q[N_Q[w]]$, where $U$ is the $S$-block of $G_2$ containing $w'$ and $V(U)\subseteq V(F)$.
	If $w'\in V(H)\setminus S$, 
	then since $(G_1, S)\oplus (H,S)$ respects $(g,h)$, 
	$U[N_U[w']]$ is label-isomorphic to $Q[N_Q[w]]$, where $U$ is the $S$-block of $H$ containing $w'$ and $V(U)\subseteq V(F)$.
	Thus, in these cases, $F$ contains a vertex with label $L_Q(v)$; a contradiction.	
	We may assume that $w'$ is contained in $S$.
	
	Next, we assume that $\{w'\}$ is the vertex set of some component of $F[S_F]$.
	In this case, $F$ has at least $3$ vertices, because $F$ contains some edge of $G_2[S]$. 
	Thus, $w'$ has a neighbor in $F$.
	We claim that $w'$ has neighbors in precisely one of $F\cap G_2$ and $F\cap H$.
	Towards a contradiction, suppose $w'$ has neighbors in both $F\cap G_2$ and $F\cap H$.
	Note that $F-w'$ is connected. We take a shortest path $P$ from $N_{F\cap G_2}(w')$ to $N_{F\cap H}(w')$.
	By construction, the end vertices of $P$ are not adjacent, and $w'$ is not adjacent to any internal vertices of $P$.
	Thus, $F[\{w'\}\cup V(P)]$ is a chordless cycle, contradicting the fact that $F$ is partially label-isomorphic to $Q$ and $Q$ is chordal.
	We conclude that $w'$ has neighbors in precisely one of $F\cap G_2$ and $F\cap H$.
	
	If $w'$ has a neighbor in $F\cap G_2$, then 
	by the completeness condition of the characteristic, 
	$U[N_U[w']]$ is label-isomorphic to $Q[N_Q[w]]$, where $U$ is the $S$-block of $G_2$ containing $w'$ and $V(U)\subseteq V(F)$.
	If $w'$ has a neighbor in $F\cap H$, then 
	since $(G_1, S)\oplus (H,S)$ respects $(g,h)$, 
	$U[N_U[w']]$ is label-isomorphic to $Q[N_Q[w]]$, where $U$ is the $S$-block of $H$ containing $w'$ and $V(U)\subseteq V(F)$.
	Thus, in these cases, $F$ contains a vertex with label $L_Q(v)$; a contradiction.	
	
	Finally, we may assume that there is a non-trivial block $B'$ of $F[S_F] $ containing $w'$.
	We observe that 
	the $S$-block of $(G_1, S)\oplus (H, S)$ containing $B'$ is label-isomorphic to $Q$.
	We also observe that every label appearing in the neighborhood of $w'$ in the $S$-block of $(G_1, S)\oplus (H, S)$ containing $B'$
	appears in the neighborhood of $w'$ in $(G_2, S)\oplus (H, S)$ as well, because $(G_1, S)$ and $(G_2, S)$ have the same characteristic.
	This contradicts the assumption that $F$ has no vertex with label $L_Q(v)$.
	We conclude that $L_Q(V(Q))\subseteq L_F(V(F))$.
	\end{clproof}

	We conclude that $F$ is label-isomorphic to $Q$. Since $B$ was arbitrarily chosen, this implies that $(G_2, S)\oplus (H,S)$ respects $(g,h)$.
	Lastly, we confirm that $(G_2,S)\oplus (H,S)$ is a $d$-labeled $\cP$-block graph.
	
	\begin{claim}\label{claim:blockgraph}
	The graph $(G_2,S)\oplus (H,S)$ is a $d$-labeled $\cP$-block graph.
	\end{claim}
	\begin{clproof}
	It is sufficient to show that every non $S$-block of $(G_2,S)\oplus (H,S)$ is fully contained in $G_2$ or $H$.
	We observe that since $\Part (G_2,S)\oplus \Part (H,S)$ has no cycles and every $S$-block of $(G_2, S)\oplus (H,S)$ is chordal, 
	by Proposition~\ref{prop:nocircuit}, 
	we have $(G_2, S)\oplus (H, S)$ is chordal.

	Suppose towards a contradiction that there is a non $S$-block $U$ of $(G_2,S)\oplus (H,S)$ intersecting both $G_2-S$ and $H-S$. 
We choose a triple $(v,w,D)$ such that 
	\begin{itemize}
	\item $v\in V(U)\cap (V(G_2)\setminus S)$, $w\in V(U)\cap (V(H)\setminus S)$, 	
	$D$ is a cycle containing $v$ and $w$ in $U$; and
	\item the length of $D$ is minimum.
	\end{itemize}
	Let $P_1$ and $P_2$ be the two paths from $v$ to $w$ in $D$.

	We claim that there are no edges between the internal vertices of $P_1$ and the internal vertices of $P_2$.
	Suppose there is an edge $p_1p_2$ for some $p_1\in V(P_1)\setminus \{v,w\}$ and $p_2\in V(P_2)\setminus \{v,w\}$.
	One of $p_1$ and $p_2$ is contained in $G_2-S$ or $H-S$, as
	$U$ can contain at most one vertex of each component of $G_2[S]$.
	Now, if $p_1$ and $p_2$ are contained in $G_2$, then we can replace $v$ with one of $p_1$ and $p_2$ that is in $G_2-S$, and obtain a cycle shorter than $D$; a contradiction.
	Similarly, if they are contained in $H$, then we obtain a cycle shorter than $D$.
	This implies that there are no edges between the internal vertices of $P_1$ and the internal vertices of $P_2$.
	Since $v$ is not adjacent to $w$, $D$ is a chordless cycle, which contradicts the fact that $(G_2,S)\oplus (H,S)$ is chordal.
	We conclude that every non $S$-block of $(G_2,S)\oplus (H,S)$ is fully contained in $G_2$ or $H$, and therefore $(G_2,S)\oplus (H,S)$ is a $d$-labeled $\cP$-block graph.
	\end{clproof}

This concludes the proof.
\end{proof}

\subsection{Main algorithm}\label{sec:mainalgo}

	Let $(G,S)$ be a boundaried graph, and $\cC$ be the set of components of $G[S]$.
 	For a partition $\cZ$ of $\cC$, 
	we write $\Inc(\cC, \cZ)\sim \Part (G,S)$ if
	\begin{itemize}
	\item two components of $G[S]$ are in the same part of $\cZ$ if and only if they are contained in the same component of $G$.
	\end{itemize}
	One can observe that there is an isomorphism from $\Inc(\cC, \cZ)$ to $\Part (G,S)$
	that maps each component of $\cC$ to the same component.

	\begin{proof}[Proof of \cref{thm:main1b}]
	Using Theorem~\ref{thm:approxtw} and Lemma~\ref{lem:nicetd}, 
	we obtain a nice tree decomposition of $G$ of width at most $5w+4$ in time $\mathcal{O}(c^w\cdot n)$ for some constant $c$.
	Let $(T, \cB=\{B_t\}_{t\in V(T)})$ be the resulting nice tree decomposition with root node $ro$. 
	For each node $t$ of $T$, let $G_t$ be the subgraph of $G$ induced by the union of all bags $B_{t'}$ where $t'$ is a descendant of $t$.
	Recall that $\cU_d$ is the class of all biconnected $d$-labeled $\cP$-block graphs, where each $H$ in $\cU_d$ has a labeling $L_H$.
	Note that $\abs{\cU_d}\le 2^{\binom{d}{2}}$.
	We start with enumerating all graphs in $\cU_d$ and their labelings.
	It takes time $2^{\mathcal{O}(d^2)}$.
	
	We define the following notation for every pair a node $t$ of $T$ and $X\subseteq B_t$:
\begin{enumerate}
\item Let $\Comp (t,X)$ be the set of all components of $G[B_t\setminus X]$.
\item Let $\Parti (t,X)$ be the set of all partitions of $\Comp (t,X)$.
\item Let $\Block (t,X)$ be the set of all non-trivial blocks of $G[B_t\setminus X]$.
\end{enumerate}
	For each node $t$ of $T$, $X\subseteq B_t$, and a function $L:B_t\setminus X\rightarrow [d]$, we define $\cF (t, X, L)$ as 
	the set of all pairs $(g,h)$ consisting of functions
	$g:\Block (t, X)\rightarrow\cU_d$ and $h:\Block (t,X)\rightarrow 2^{[d]}$.
	We say that $(g,h)$ is \emph{valid} if 
	\begin{itemize}
	\item $L$ is a $d$-labeling of $G[B_t\setminus X]$,
	\item for each $B\in \Block (t,X)$, $B$ is partially label-isomorphic to $g(B)$, and
	\item for each $B\in \Block (t,X)$, $L(V(B))\cap h(B)=\emptyset$.
	\end{itemize}
	Furthermore, for $i\in \{0, 1, \ldots, k\}$  and $(g,h)\in \cF(t, X, L)$, 
	let $c[t, (X, L, i, (g,h))]$ be the family
	of all partitions~$\cX$ in $\Parti (t,X)$ satisfying the following property: there exist $S\subseteq V(G_t)\setminus B_t$ with $\abs{S}=i$ and a $d$-labeling $L'$ of $G_t-(X\cup S)$ 	where
	\begin{itemize}
	\item $L =L'|_{B_t\setminus X}$, 
	\item $G_t-(X\cup  S)$ is a $\cP$-block graph,
	\item $(g, h)$ is a characteristic of $(G_t-(X\cup  S), B_t\setminus X)$, and 
	\item $\Inc (\Comp (t,X), \cX)\sim \Part (G_t-(X\cup  S), B_t\setminus X)$.
	\end{itemize}
	Such a pair $(S, L')$ will be called a \emph{partial solution} with respect to $(t, (X, L, i, (g,h)), \cX)$.
	It is easy to verify that $c[t, (X, L, i, (g,h)) ]=\emptyset$ if $(g,h)$ is not valid. 
	Let $\cM_t$ be the set of all possible tuples $(X,L, i, (g,h))$ at node $t$.

	The main idea of the algorithm is that instead of fully computing $c[t, M]$ for $M=(X,L, i, (g,h))\in \cM_t$, 
	we recursively enumerate a set $r[t, M]$ that represents $c[t, M]$.
	Formally, for a subset $r[t,M]\subseteq c[t,M]$, we denote $r[t,M]\equiv c[t,M]$ if
	\begin{itemize}
	\item for every $\cX\in c[t,M]$ and a partial solution $(S, L')$ with respect to $(t, M, \cX)$ and 
	$S_{out}\subseteq V(G)\setminus V(G_t)$ where $(G-(S\cup X\cup S_{out}), B_t\setminus X)$ is a $d$-labeled $\cP$-block graph respecting $(g,h)$
	(considering $G-(S\cup X\cup S_{out})$ as the sum $(G_t-(S\cup X), B_t\setminus X)\oplus (G-(V(G_t)\setminus B_t)-(X\cup S_{out}), B_t\setminus X)$), 
	there exists $\cX_1\in r[t,M]$ and a partial solution $(S', L'')$ with respect to $(t, M, \cX_1)$ such that
	$(G-(S'\cup X\cup S_{out}), B_t\setminus X)$ is a $d$-labeled $\cP$-block graph respecting $(g,h)$. 
	\end{itemize}
	By the definition of $r[t,M]$, 
	the problem is a \YES-instance if and only if 
	there exists $(X, L, i, (g,h))\in \cM_r$ with $\abs{X}+i\le k$ such that $r[ro,(X, L, i, (g,h))]\neq \emptyset$.
	To decide whether the problem is a \YES-instance, 
	we enumerate $r[t,M]$ for all nodes $t$ and all $M\in \cM_t$.

	Whenever we update $r[t,M]$, we confirm that $\abs{r[t,M]}\le w\cdot 2^{w-1}$. 
	This is a consequence of Proposition~\ref{prop:repalgorithm}.
	We describe how to update families $r[t,M]$ depending on the type of node $t$, 
	and prove the correctness of each procedure.
	We fix such a tuple.
	For each leaf node $t$ and all $0\le i\le k$ and empty functions $L$, $g$, $h$, we assign $r[t,(\emptyset, L, i, (g, h))]:=\emptyset$. 
	We may assume that $t$ is not a leaf node.
	Let $M:=(X,L, i, (g,h))\in \cM_t$.
	We may assume $(g,h)$ is valid.

\medskip\medskip
\noindent\textbf{1) $t$ is an introduce node with child $t'$ and $B_t\setminus B_{t'}=\{v\}$:}
\medskip

	If $v\in X$, then $G_t-X=G_{t'}-(X\setminus \{v\})$ and $B_t\setminus X=B_{t'}\setminus (X\setminus \{v\})$.
	So, we can set $r[t,M]:=r[t', (X\setminus \{v\}, L, i, (g, h))]$.	
	We assume $v\notin X$, and let $L_{res}:=L|_{B_{t'}\setminus X}$. 

	For a pair $(g,h)\in \cF(t, X, L)$, 
	a pair $(g', h')\in \cF(t', X, L_{res})$ is called the \emph{restriction} of $(g,h)$ 
	if  
	\begin{itemize}
	\item for $B_1\in \Block (t', X)$ and $B_2\in \Block(t, X)$ with $V(B_1)\subseteq V(B_2)$, 
	\begin{itemize}
	\item $g'(B_1)=g(B_2)$,
	\item if $v\in V(B_2)$, 
	then every vertex in $g'(B_1)$ with label in $h'(B_1)$ is not adjacent to the vertex in $g'(B_1)$ with label $L(v)$,
	\end{itemize}
	\item for $B_1\in \Block (t', X)$ and $B_2\in \Block(t, X)$ with $V(B_1)\subseteq V(B_2)$ and $v\notin V(B_2)$, $h'(B_1)=h(B_2)$, and 
	\item for $B_2\in \Block (t,X)$ containing $v$, $h(B_2)=\bigcup_{B_1\in \Block(t', X), V(B_1)\subseteq V(B_2)} h(B_1)$.
	\end{itemize}
	
	\begin{claim}\label{claim:introduce} 
	For every $\cX\in \Parti (t,X)$, $\cX\in c[t,M]$ if and only if 
	there exist a restriction $(g',h')$ of $(g,h)$ and $\cY\in c[t',(X,L_{res}, i, (g',h'))]$ such that
	\begin{itemize}
	\item $v$ has neighbors on at most one component in each part of $\cY$ (that is, $\Inc(\Comp(t',X), \cY)\oplus \Part (G[B_t\setminus X], B_{t'}\setminus X)$ has no cycles), and 
	\item if $v$ has at least one neighbor in $G[B_t\setminus X]$, then $\cX$ is the partition obtained from $\cY$ by, 
	for parts $Y_1, \ldots, Y_m$ of $\cY$ containing components having a neighbor of $v$,
	removing all of $Y_1, \ldots, Y_m$ and adding a part that consists of all components of $G[B_t\setminus X]$ that are not contained in 
	parts of $\cY\setminus \{Y_1, \ldots, Y_m\}$;
	and otherwise, $\cX=\cY\cup \{\{v\}\}$.
	\end{itemize}
	\end{claim}

\begin{clproof}
	Suppose $\cX\in c[t,M]$ and let $(S,L_t)$ be a partial solution with respect to $(t, M, \cX)$.
	Observe \[ G_t-(X\cup S)= (G_{t'}-(X\cup S), B_{t'}\setminus X)\oplus (G[B_t\setminus X], B_{t'}\setminus X).\]
	Let $\cY\in \Parti (t',X)$ such that $\Inc(\Comp(t',X), \cY)\sim \Part (G_{t'}-(X\cup S), B_{t'}\setminus X)$.

	As 
	$G_t-(X\cup S)= (G_{t'}-(X\cup S), B_{t'}\setminus X)\oplus (G[B_t\setminus X], B_{t'}\setminus X)$
	and $G_t-(X\cup S)$ is chordal, by Proposition~\ref{prop:nocircuit}, 
	\begin{align*}
	\Inc(\Comp(t',X), \cY)\oplus \Part (G[B_t\setminus X], B_{t'}\setminus X)
	\end{align*} 
	has no cycles.
	The second condition holds by the definition of $\cY$.
	Since we can naturally obtain a restriction $(g',h')$ of $(g,h)$ for $(G_{t'}-(X\cup S), B_{t'}\setminus X)$, 
	this concludes the proof of the forward direction.
	
	For the converse, suppose 
	there exist $(g',h')$ and $\cY$ satisfying the assumption.
	Let $M_{res}:=(X,L_{res}, i, (g',h'))$, and $(S, L_{t'})$ be a partial solution with respect to $(t', M_{res}, \cY)$. 
	For convenience, we define that
	\begin{itemize}
	\item $H:=G_t-(X\cup S)$, 
	\item $H':=G_{t'}-(X\cup S)$,
	\item $L_t:V(H)\rightarrow [d]$ is the function obtained from $L_{t'}$ by further assigning $L_t(v):=L(v)$.
	\end{itemize}

	We claim that $(g,h)$ is a characteristic of $(H, B_t\setminus X)$. 
	Before checking the conditions of a characteristic, 
	we show that 
	if two blocks $D_1, D_2\in \Block (t',X)$ are contained in the same $(B_{t'}\setminus X)$-block of $H$, 
	then $g'(D_1)=g'(D_2)$.
	
	Let $D_1, D_2\in \Block (t', X)$. 
	If $D_1$ and $D_2$ are contained in the same $(B_{t'}\setminus X)$-block of $(H', B_{t'}\setminus X)$, 
	then $g'(D_1)=g'(D_2)$ because 
	$(g',h')$ is a characteristic of $(H', B_{t'}\setminus X)$.
	Also, if $D_1$ and $D_2$ are contained in the same $(B_{t'}\setminus X)$-block of  $(G_{t}[B_t\setminus X], B_{t'}\setminus X)$,
	then $g'(D_1)=g'(D_2)$ as $(g',h')$ is a restriction of $(g,h)$.
	By the assumption, 
	$\Inc(\Comp(t',X), \cY)\oplus \Part (G[B_t\setminus X], B_{t'}\setminus X)$ and equivalently,
	$\Part (H', B_{t'}\setminus X)\oplus \Part (G[B_t\setminus X], B_{t'}\setminus X)$ 
	have no cycles.
	Therefore, by Lemma~\ref{lem:gvlauesblock}, if $D_1$ and $D_2$ are contained in the same $(B_{t'}\setminus X)$-block of $H$, 
	then $g'(D_1)=g'(D_2)$.

	Let $B\in \Block (t,X)$ and $F$ be the $(B_t\setminus X)$-block of $H$ containing $B$.
\begin{enumerate}
\item (Coincidence condition)
	
	Let $B'\in \Block (H, B_t\setminus X)$ such that $B\neq B'$ and $B'$ is contained in $F$.
	First assume that $\abs{V(B)\setminus \{v\}}=1$ or $\abs{V(B')\setminus \{v\}}=1$.
	In this case, since $v$ has neighbors on at most one component in each part of $\cY$, 
	$B_t\setminus \{v\}$ and $B_{t'}\setminus \{v\}$ must be contained in the same component of $G[B_{t'}\setminus X]$; call it $C$.
	Then there is a path from $B_t\setminus \{v\}$ to $B_{t'}\setminus \{v\}$ in $C$, 
	and therefore $B$ and $B'$ are contained in the same block of $G[B_t\setminus X]$. 
	This contradicts the assumption that $B$ and $B'$ are distinct blocks in $\Block (H, B_t\setminus X)$. 
	
	Thus, both $B$ and $B'$ contain non-trivial blocks $U$ and $U'$ in $G[B_{t'}\setminus X]$ respectively,
	where $g'(U)=g'(U')$. This implies that $g(B)=g(B')$.

\item (Neighborhood condition)
	
	We need to show that $h(B)=\bigcup_{B_1\in \Block(t', X), V(B_1)\subseteq V(B)} h(B_1)$.
	Note that $(g',h')$ is a restriction of $(g,h)$.	
	If $B$ does not contain $v$, then $B$ is a block of $G[B_{t'}\setminus X]$, and $h(B)=h'(B)$.
	We assume that $B$ contains $v$. 
	
	It is easy to confirm that $h(B)\supseteq \bigcup_{B_1\in \Block(t', X), V(B_1)\subseteq V(B)} h(B_1)$, 
	since the block of $H'$ containing $B_1\in \Block(t', X)$ with $V(B_1)\subseteq V(B)$ has to be contained in $F$.
	To see that $h(B)\subseteq \bigcup_{B_1\in \Block(t', X), V(B_1)\subseteq V(B)} h(B_1)$, 
	let $z\in N_F(V(B))$, and choose a neighbor $w$ of $z$ in $B$.
	Since $v$ is introduced at the current node, $w\neq v$.
	Thus $F$ contains three vertices, and $F$ is $2$-connected.
	Also, if $V(B)=\{v,w\}$, then $w$ is a cut vertex in $F$, a contradiction.
	So $B$ also has at least $3$ vertices, and it is $2$-connected. In particular, $B-v$ is connected.
	
	We take a shortest path $P$ from $z$ to $V(B)$ in $F-w$.
	Let $p$ be the endpoint of $P$ at $V(B)$, and 
    let $Q$ be a path from $p$ to $w$ in $B-v$.
    Then $P\cup Q$ and $wz$ form a cycle and thus there exists an $S$-block $B_1$ of $H'$ with $V(B_1)\subseteq V(B)$, where $z$ is in the neighborhood of this block.
    This implies that 
	$h(B)\subseteq \bigcup_{B_1\in \Block(t', X), V(B_1)\subseteq V(B)} h(B_1)$.

 \item (Label-isomorphism condition)
 
	We prove that $F$ is partially label-isomorphic to $g(B)$.
	Let $F_1:=F\cap H'$, $F_2:=F\cap G[B_t\setminus X]$, and $U=V(F_1)\cap V(F_2)$.
	Since $\Part (H', B_{t'}\setminus X)\oplus \Part (G[B_t\setminus X], B_{t'}\setminus X)$ has no cycles, 
	by Lemma~\ref{lem:nocycle}, $\Part (F_1,U)\oplus \Part (F_2,U)$ has no cycles.
	To apply Proposition~\ref{prop:sumofchordalgraphs2}, 
	we verify that 
	$(F_1, U)$ and $(F_2, U)$ are block-wise $g(B)$-compatible.
	We observed that 
	if two non-trivial blocks $D_1$ and $D_2$ of $G[B_{t'}\setminus X]$ are contained in $F$, 
	then $g'(D_1)=g'(D_2)=g(B)$.
	
	Since $(g',h')$ is a characteristic of $(H', B_{t'}\setminus X)$, 
	$(F_1, U)$ is block-wise partially label-isomorphic to $g(B)$.
	Also, since $(g,h)$ is valid, $(F_2, U)$ is block-wise partially label-isomorphic to $g(B)$.
	As $(g',h')$ is a restriction of $(g,h)$, 
	for $B_1\in \Block (t', X)$ and $B_2\in \Block(t, X)$ with $V(B_1)\subseteq V(B_2)$ and $v\in V(B_2)$, 
	every vertex in $g'(B_1)$ with label in $h'(B_1)$ is not adjacent to the vertex in $g'(B_1)$ with label $L(v)$.
	Because of this condition, the second condition of being block-wise $g(B)$-compatible is also satisfied.
	
	By Proposition~\ref{prop:sumofchordalgraphs2}, $F$ is partially label-isomorphic to $g(B)$.

\item (Completeness condition)

	This follows from  the fact that $(g',h')$ is a restriction of $(g,h)$ and it is a characteristic of $(H', B_{t'}\setminus X)$.
\end{enumerate}
All together we conclude that $(g,h)$ is a characteristic of $(H, B_t\setminus X)$ and therefore $\cX\in c[t,M]$.
\end{clproof}

	When $v\notin X$, we update $r[t,M]$ as follows. 
	Set $\cK:=\emptyset$ at the beginning.
	For every $(g',h')\in \cF(t',X,L_{res})$, we test whether $(g',h')$ is a restriction of $(g,h)$.
	Assume $(g',h')$ is a restriction of $(g,h)$, otherwise, we skip it.
	Now, for each $\cY\in r[t',(X,L_{res}, i, (g',h'))]$,
	we check the two conditions for $(g',h')$ and $\cY$ in Claim~\ref{claim:introduce}, 
	and if they are satisfied, then we add the set $\cX$ described in Claim~\ref{claim:introduce} to $\cK$;
	otherwise, we skip it.
	Since $\abs{\cF(t',X,L_{res})}\le 2^{\mathcal{O}(wd^2)}$ and $\abs{r[t',(X,L_{res}, i, (g',h'))]}\le w\cdot 2^{w-1}$, the whole procedure can be done in time $2^{\mathcal{O}(wd^2)}$.
	After we do this for all possible candidates, we take a representative set of $\cK$ using Proposition~\ref{prop:repalgorithm}, 
	and assign the resulting set to $r[t,M]$. 	
	Since $\abs{\cK}\le 2^{\mathcal{O}(wd^2)}$, we can apply Proposition~\ref{prop:repalgorithm} in time $2^{\mathcal{O}(wd^2)}$. 
	Also, we have $\abs{r[t,M]}\le w\cdot 2^{w-1}$.

	We claim that $r[t,M]\equiv c[t,M]$.
	Let $G_{out}:=G-(V(G_t)\setminus B_t)$. 
	Let $\cX\in c[t,M]$ and
	$(S, L')$ be a partial solution with respect to $(t, M, \cX)$, and
	suppose there exists $S_{out}\subseteq V(G)\setminus V(G_t)$ where 
	\[G-(S\cup X\cup S_{out})=(G_{t}-(X\cup S), B_{t}\setminus X)\oplus (G_{out}-(X\cup S_{out}), B_{t}\setminus X)\] is a $d$-labeled $\cP$-block graph respecting $(g,h)$. 
	Note that every $(B_{t'}\setminus X)$-block of $G-(S\cup X\cup S_{out})$ is chordal as such a block is also a $(B_t\setminus X)$-block of $G-(S\cup X\cup S_{out})$.
	Since $G-(S\cup X\cup S_{out})$ is chordal, by Proposition~\ref{prop:nocircuit},
	$\Part (G_{t'}-(X\cup S), B_{t'}\setminus X)\oplus \Part (G_{out}-(X\cup S_{out}), B_{t'}\setminus X)$ has no cycles.
	Recall that $M_{res}:=(X,L_{res}, i, (g',h'))$.
	As $r[t',M_{res}]\equiv c[t', M_{res}]$, there exist $\cY\in r[t',M_{res}]$ and a partial solution $(S', L'')$ with respect to $(t', M_{res}, \cY)$ such that
	\begin{itemize}
	\item $\Inc(\Comp(t',X), \cY)\sim \Part (G_{t'}-(X\cup S'), B_{t'}\setminus X)$, and
	\item $\Part (G_{t'}-(X\cup S'), B_{t'}\setminus X)\oplus \Part (G_{out}-(X\cup S_{out}), B_{t'}\setminus X)$ has no cycles.
	\end{itemize}
	By \cref{thm:equivalence2}, 
	$G-(S'\cup X\cup S_{out})$ is also a $d$-labeled $\cP$-block graph respecting $(g,h)$. 

	By the update procedure, the partition $\cX_1$ where $\Inc(\Comp(t,X), \cX_1)\sim \Part (G_{t}-(X\cup S'), B_{t}\setminus X)$
	 is added to the set $\cK$,
	and there exist $\cX_2\in r[t,M]$ and a partial solution $(S'', L''')$ with respect to $(t, M, \cX_2)$ such that 
	$G-(S''\cup X\cup S_{out})$ is a $d$-labeled $\cP$-block graph respecting $(g,h)$.
	This shows that $r[t,M]\equiv c[t,M]$.

\medskip\medskip
\noindent\textbf{2) $t$ is a forget node with child $t'$ and $B_{t'}\setminus B_t=\{v\}$:}
\medskip

	A pair $(g', h')\in \cF(t', X, L')$ is called an \emph{extension} of $(g,h)$ if 
    \begin{itemize}
    		\item $(g', h')$ is valid and $L'$ is an extension of $L$ on $B_{t'}\setminus X$, 
        \item for $B_1\in \Block(t,X)$ and $B_2\in \Block (t',X)$ with $V(B_1)\subseteq V(B_2)$, 
            \begin{itemize}
                \item $g'(B_2)=g(B_1)$,
                \item if $v\notin V(B_2)$, then $h(B_1)=h'(B_2)$,
                \item if $v\in V(B_2)$, then $h(B_1)$ is the union of $\{L'(v)\}$ and the set of labels in $N_{g(B_1)}(A)$ that appear in $h'(B_2)$ where $A$ is the set of vertices in $g(B_1)$ with labels in $L(B_1)$.
                \end{itemize}
    \end{itemize}

We show the following.

	\begin{claim}\label{claim:forget}
	 For every $\cX\in \Parti (t,X)$, $\cX\in c[t,M]$ if and only if 
	one of the following holds:
	\begin{enumerate}[(i)]
	\item $\cX\in c[t',(X\cup \{v\},L, i-1, (g,h))]$, or
	\item there exist an extension $L_{ext}$ of $L$ on $B_{t'}\setminus X$, an extension $(g',h')$ of $(g,h)$ in $\cF(t',X, L_{ext})$ with respect to $L_{ext}$, and
	$\cY\in c[t',(X,L_{ext}, i, (g',h'))]$ such that 
	$\cX$ is the partition obtained from $\cY$ by replacing the component $U$ of $G[B_{t'}\setminus X]$ containing $v$ with 
	the components of $G[B_t\setminus X]$ contained in $U$.
	\end{enumerate}
	\end{claim}
	\begin{clproof}
	We first show the backward direction.
	If $\cX\in c[t',(X\cup \{v\},L, i-1, (g,h))]$, 
	then $\cX\in c[t,M]$, as we can put $v$ into the partial solution.
	Suppose statement (ii) holds.
	Then there exists a partial solution $(S, L')$ with respect to $(t, (X,L_{ext}, i, (g',h')), \cY)$. 
	It is not difficult to verify that $(g,h)$ is the characteristic of $(G_t-(X\cup S), B_{t}\setminus X)$ and 
	$\Inc(\Comp(t,X), \cX)\sim \Part (G_{t}-(X\cup S), B_{t}\setminus X)$. 
	Thus, $\cX\in c[t,M]$.

	For the other direction, suppose $\cX\in c[t,M]$, and 
	let $(S, L')$ be a partial solution with respect to $(t, M, \cX)$. 
	If $v\in S$, then $\cX\in c[t', (X\cup \{v\}, L, i-1, (g,h))]$, and the statement (i) holds.
	Thus, we may assume that $v\notin S$. 
	
	Let $L_{ext}:=L'|_{B_{t'}\setminus X}$ and $\cY\in \Parti (t',X)$ such that
	$\Inc (\Comp(t', X), \cY)\sim \Part (G_t-(X\cup S), B_{t'}\setminus X)$.
	Since $G_t-(X\cup S)=G_{t'}-(X\cup S)$, one can observe that $\cX$ is the partition obtained from $\cY$ 
	by replacing the component $U$ of $G[B_t\setminus X]$ containing $v$ with 
	the components of $B_{t'}\setminus X$ contained in $U$.
	We focus on showing that there exists an extension $(g',h')$ of $(g,h)$ in $\cF(t',X, L_{ext})$
	that is the characteristic of $(G_{t'}-(X\cup S), B_{t'}\setminus X)$.

	We construct $(g', h')$ as follows. 
	\begin{itemize}
	\item Suppose there is a block $B\in \Block (t',X)$ containing $v$. If there exists $B'\in \Block(t,X)$ where $B$ and $B'$ are contained in the same block of $G_t-(X\cup S)$, 
	then we let $g'(B)=g(B')$.
	Otherwise, we know that the block of $G_{t'}-(X\cup S)$ containing $v$ is label-isomorphic to a graph in $\cU_d$; let $g'(B)$ be this graph.
	\item For $B\in \Block (t',X)$ with $v\notin V(B)$, let $g'(B)=g(B)$.
	\item Also, for every $B\in \Block (t', X)$, let $h'(B)$ be the set of labels that appear in the neighbors of vertices of $B$ 
	that are in the block of $G_{t'}-(S\cup X)$ containing $B$ and are not in $B_{t'}\setminus X$.
	\end{itemize}
	Then $(g',h')$ is an extension of $(g,h)$, and $\cY\in c[t',(X, L_{ext}, i, (g',h'))]$.
	\end{clproof}

	We update $r[t,M]$ as follows. 
	Set $\cK:=\emptyset$.
	First, we add all partitions in $r[t',(X\cup \{v\},L, i-1, (g,h))]$ to $\cK$. 
	At the second step,
	for every extension $L_{ext}$ of $L$ on $B_{t'}\setminus X$ and every $(g',h')\in \cF(t',X,L_{ext})$, 
	we test whether $(g',h')$ is an extension of $(g,h)$. 
	Note that the condition for $h$ can be checked in polynomial in $d$ and $w$.
	In the case when $(g',h')$ is an extension of $(g,h)$ with respect to $L_{ext}$,
	for all partitions $\cY\in r[t',(X, L_{ext}, i, (g',h'))]$, we add the set $\cX$ satisfying the second statement in Claim~\ref{claim:forget} to $\cK$, 
	and otherwise, we skip this pair.
	This can be done in time $2^{\mathcal{O}(wd^2)}$.
 	After we do this for all possible candidates, we take a representative set of $\cK$ using Proposition~\ref{prop:repalgorithm}, 
	and assign the resulting set to $r[t,M]$. 
	Notice that $\abs{\cK}\le 2^{\mathcal{O}(wd^2)}$. By Proposition~\ref{prop:repalgorithm}, 
	the procedure of obtaining a representative set can be done in time $2^{\mathcal{O}(wd^2)}$, and we have $\abs{r[t,M]}\le w\cdot 2^{w-1}$.

	We claim that $r[t,M]\equiv c[t,M]$.
	Let $\cX\in c[t,M]$ and $(S, L')$ be a partial solution with respect to $(t, M, \cX)$ and 
	$S_{out}\subseteq V(G)\setminus V(G_t)$ where $G-(S\cup X\cup S_{out})$ is a $d$-labeled $\cP$-block graph respecting $(g,h)$. 
	Let $G_{out}:=G-(V(G_{t'})\setminus B_{t'})$. 
	The graph $G-(S\cup X\cup S_{out})$ can be seen as 
	$(G_{t'}-(X\cup S), B_{t'}\setminus X)\oplus (G_{out}-(X\cup S_{out}), B_{t'}\setminus X)$.

	Note that every $(B_{t'}\setminus X)$-block of $G-(S\cup X\cup S_{out})$ is chordal.
	Since $G-(S\cup X\cup S_{out})$ is chordal, by Proposition~\ref{prop:nocircuit},
	$\Part (G_{t'}-(X\cup S), B_{t'}\setminus X)\oplus \Part (G_{out}-(X\cup S_{out}), B_{t'}\setminus X)$ has no cycles.
	As $r[t',(X, L_{ext}, i, (g',h')]\equiv c[t', (X, L_{ext}, i, (g',h')]$, there exists $\cY\in r[t', (X, L_{ext}, i, (g',h')]$ and a partial solution $(S', L'')$ with respect to $(t', (X, L_{ext}, i, (g',h'), \cY)$ such that
	$\Inc(\Comp(t',X), \cY)\sim \Part (G_{t'}-(X\cup S'), B_{t'}\setminus X)$, and thus
	$\Part (G_{t'}-(X\cup S'), B_{t'}\setminus X)\oplus \Part (G_{out}-(X\cup S_{out}), B_{t'}\setminus X)$ has no cycles.
	By \cref{thm:equivalence2}, 
	$G-(S'\cup X\cup S_{out})$ is a $d$-labeled $\cP$-block graph respecting $(g',h')$ for some extension $(g',h')$ of $(g,h)$. 
	By the procedure, the partition $\cX_1$ where $\Inc(\Comp(t,X), \cX_1)\sim \Part (G_{t}-(X\cup S'), B_{t}\setminus X)$
	 is added to the set $\cK$,
	and there exists $\cX_2\in r[t,M]$ and a partial solution $(S'', L''')$ with respect to $(t, M, \cX_2)$ such that 
	$G-(S''\cup X\cup S_{out})$ is a $d$-labeled $\cP$-block graph respecting $(g,h)$.
	This shows that $r[t,M]\equiv c[t,M]$.

\medskip\medskip
\noindent\textbf{3) $t$ is a join node with two children $t_1$ and $t_2$:}
\medskip

We show the following.

	\begin{claim}\label{claim:join} 
	For every $\cX\in \Parti (t,X)$, $\cX\in c[t,M]$ if and only if 
	there exist integers $i_1, i_2$ with $i_1+i_2=i$, 
	$(g,h_1)\in \cF(t_1,X,L)$, $(g,h_2)\in \cF(t_2,X, L)$, $\cX_1\in c[t_1,(X,L,i_1, (g,h_1))]$, and $\cX_2\in c[t_2,(X,L,i_2, (g,h_2))]$  such that
	\begin{itemize}
	\item $\Inc (\Comp(t,X), \cX_1\cup \cX_2)$ has no cycles, 
	\item $\cX=\cX_1\uplus \cX_2$, and
	\item for each $B\in \Block (t,X)$, $h_1(B)\cap h_2(B)=\emptyset$ and $h(B)=h_1(B)\cup h_2(B)$, 
	and for $\ell_1\in h_1(B)$ and $\ell_2\in h_2(B)$, the vertices with labels $\ell_1$ and $\ell_2$ in $g(B)$ are not adjacent.
	\end{itemize}
	\end{claim}
	\begin{clproof}
	The forward direction is straightforward.
	For the converse direction, suppose there exist integers $i_1, i_2$ with $i_1+i_2=i$, and
	$(g,h_1)$, $(g,h_2)$, and partitions $\cX_1, \cX_2$ as specified in the statement.
	For each $j\in \{1,2\}$, let $M_j:=(X,L,i_j, (g,h_j))$ and $(S_j, L_j)$ be a partial solution with respect to $(t_j,M_j, \cX_j)$.
	Furthermore, let $H_j:=G_{t_j}-(X\cup S_j)$, $H:=H_1\cup H_2$, and $L_H:=L_1\oplus L_2$.

	We claim that $(g,h)$ is a characteristic of $(H, B_t\setminus X)$. 
	
\begin{enumerate}
\item (Coincidence condition)
	
	Let $i\in \{1,2\}$. Since $(g,h_i)$ is a characteristic of $H_i$, 
	if $B_1, B_2\in \Block (t,X)$ are contained in the same $(B_t\setminus X)$-block of $H_i$, 
	$g(B_1)=g(B_2)$.
	Since  
	$\Inc (\Comp(t,X), \cX_1\cup \cX_2)$ and equivalently, $\Part (H_1, B_t\setminus X)\oplus \Part (H_2, B_t\setminus X)$ have no cycles, 
	by Lemma~\ref{lem:gvlauesblock}, 
	if $B_1, B_2\in \Block (t,X)$ are contained in the same $(B_t\setminus X)$-block of $H$, 
	then we have $g(B_1)=g(B_2)$.

\item (Neighborhood condition)

	This follows from the assumption that $h(B)=h_1(B)\cup h_2(B)$ for each $B\in \Block (t, X)$. 

 \item (Label-isomorphism condition)
 
	Let $B\in \Block(t,X)$ and $F$ be the $(B_t\setminus X)$-block of $H$ containing $B$.
	We show that $F$ is partially label-isomorphic to $g(B)$. Let $U:=V(F)\cap (B_t\setminus X)$.
	
Since $\Part(H_1, B_t\setminus X)\oplus \Part(H_2, B_t\setminus X)$ has no cycles, by Lemma~\ref{lem:nocycle}, 
	$\Part (F\cap H_1,U) \oplus \Part (F\cap H_2, U)$ has no cycles.
	Since each $(g,h_j)$ is a characteristic of $(H_j, B_{t_j}\setminus X)$, 
	$(F\cap H_1, U)$ and $(F\cap H_2, U)$ are block-wise partially label-isomorphic to $g(B)$.	
	Moreover, $(F\cap H_1, U)$ and $(F\cap H_2, U)$ are block-wise $g(B)$-compatible, 
	because of the assumption that for each $B\in \Block (t,X)$, $h_1(B)\cap h_2(B)=\emptyset$ and $h(B)=h_1(B)\cup h_2(B)$, 
	and for $\ell_1\in h_1(B)$ and $\ell_2\in h_2(B)$, the vertices with labels $\ell_1$ and $\ell_2$ in $g(B)$ are not adjacent.
	By Proposition~\ref{prop:sumofchordalgraphs2}, $F$ is partially label-isomorphic to $g(B)$.

\item (Completeness condition)
	
	This follows from the fact that each $(g,h_j)$ is a characteristic of $(H_j, B_{t_j}\setminus X)$.

\end{enumerate}
This proves that $(g,h)$ is a characteristic of $(H, B_t\setminus X)$. That is, $(S_1\cup S_2, L_1\oplus L_2)$ is a partial solution with respect to $(t,M,\cX)$, and thus we have $\cX\in c[t,M]$.
\end{clproof}

	We update $r[t,M]$ as follows. Set $\cK:=\emptyset$.
	We fix integers $i_1, i_2$ with $i_1+i_2=i$, $(g, h_1)\in \cF(t_1, X, L)$ and $(g,h_2)\in \cF(t_2, X, L)$.
	We can check in time $\mathcal{O}(wd^2)$ the condition that 
	\begin{itemize}
	\item for each $B\in \Block (t,X)$, $h_1(B)\cap h_2(B)=\emptyset$ and $h(B)=h_1(B)\cup h_2(B)$, 
	and for $\ell_1\in h_1(B)$ and $\ell_2\in h_2(B)$, the vertices with labels $\ell_1$ and $\ell_2$ in $g(B)$ are not adjacent.
	\end{itemize}
	If these pairs do not satisfy this condition, then we skip them. We assume that these pairs satisfy this condition.
	For $\cX_1\in r[t_1, M_1]$ and $\cX_2\in r[t_2, M_2]$, 
	we test whether $\Inc (\Comp(t,X), \cX_1\cup \cX_2)$ has no cycles and $\cX=\cX_1\uplus \cX_2$.
	We can check this in time $\mathcal{O}(w)$.
	If they satisfy the two conditions, then we add the partition $\cX$ to the set $\cK$, and otherwise, we do not add it.
	After we do this for all possible candidates, we take a representative set of $\cK$ using Proposition~\ref{prop:repalgorithm}, 
	and assign the resulting set to $r[t,M]$. 
	The total running time is $k\cdot 2^{\mathcal{O}(wd^2)}$ because $\abs{\cF(t_j, X, L)}\le 2^{\mathcal{O}(wd^2)}$ and $\abs{ r[t_j,M_j]}\le w\cdot 2^{w-1}$ for each $j\in \{1,2\}$.
	We have $\abs{r[t,M]}\le w\cdot 2^{w-1}$.

	We claim that
	$r[t,M]\equiv c[t,M]$.
	Let $\cX\in c[t,M]$ and 
	$(S, L')$ be a partial solution with respect to $(t, M, \cX)$ and 
	$S_{out}\subseteq V(G)\setminus V(G_t)$ where $G-(S\cup X\cup S_{out})$ is a $d$-labeled $\cP$-block graph respecting $(g,h)$. 
	Let $H_{out}:=G-(V(G_{t})\setminus B_{t})-(X\cup S_{out})$, and for each $j\in \{1,2\}$, let $S_j=V(H_j)\cap S$.
	Note that every $(B_{t}\setminus X)$-block of $G-(S\cup X\cup S_{out})$ is chordal.

	We first consider $G-(S\cup X\cup S_{out})$ as the sum
	$(H_1, B_{t}\setminus X)\oplus (H_2\cup H_{out}, B_{t}\setminus X)$.
	Since $G-(S\cup X\cup S_{out})$ is chordal, by Proposition~\ref{prop:nocircuit},
	$\Part (H_1, B_{t}\setminus X)\oplus \Part (H_2\cup H_{out}, B_{t}\setminus X)$ has no cycles.
	As $r[t_1,M_1]\equiv c[t_1, M_1]$, 
	there exists $\cY_1\in r[t_1,M_1]$ and a partial solution $(S_1', L_1)$ with respect to $(t_1, M_1, \cY_1)$ such that
	$\Inc(\Comp(t,X), \cY_1)\sim \Part (G_{t}-(X\cup S_1'), B_{t}\setminus X)$, and
	\[\Part (G_{t_1}-(X\cup S_1'), B_{t}\setminus X)\oplus \Part (H_2\cup H_{out}, B_{t}\setminus X)\] has no cycles.
	By \cref{thm:equivalence2}, 
	$G-(S_1'\cup S_2\cup X\cup S_{out})$ is a $d$-labeled $\cP$-block graph respecting $(g,h)$. 
	Let $H_1':=G_{t_1}-(X\cup S_1')$.
	In a similar manner, 
	we consider $G-(S_1'\cup S_2\cup X\cup S_{out})$  as the sum 
	$(H_1'\cup H_{out}, B_{t}\setminus X)\oplus (H_2, B_{t}\setminus X)$.
	Since $G-(S_1'\cup S_2\cup X\cup S_{out})$ is chordal, by Proposition~\ref{prop:nocircuit},
	$\Part (H_1'\cup H_{out}, B_{t}\setminus X)\oplus \Part (H_2, B_{t}\setminus X)$ has no cycles.
	As $r[t_2,M_2]\equiv c[t_2, M_2]$, 
	there exist $\cY_2\in r[t_2,M_2]$ and a partial solution $(S_2', L_2)$ with respect to $(t_2, M_2, \cY_2)$ such that
	$\Inc(\Comp(t,X), \cY_2)\sim \Part (G_{t}-(X\cup S_2'), B_{t}\setminus X)$, and
	\[\Part (H_1\cup H_{out}, B_{t}\setminus X)\oplus \Part (G_{t_2}-(X\cup S_2'), B_{t}\setminus X)\] has no cycles.
	By \cref{thm:equivalence2}, 
	$G-(S_1'\cup S_2'\cup X\cup S_{out})$ is a $d$-labeled $\cP$-block graph respecting $(g,h)$. 
	Thus the partition $\cX_1=\cY_1\uplus \cY_2$ which satisfies $\Inc(\Comp(t,X), \cX_1)\sim \Part (G_{t}-(X\cup S_1'\cup S_2'), B_{t}\setminus X)$
	 is added to the set $\cK$.
	And there exists $\cX_2\in r[t,M]$ and a partial solution $(S'', L''')$ with respect to $(t, M, \cX_2)$ such that 
	$G-(S''\cup X\cup S_{out})$ is a $d$-labeled $\cP$-block graph respecting $(g,h)$.
	This shows that $r[t,M]\equiv c[t,M]$.

\medskip
	
	\textbf{Total running time.} 
	We denote $\abs{V(G)}$ by $n$. Note that the number of nodes in $T$ is $\mathcal{O}(wn)$ by Lemma~\ref{lem:nicetd}.
	For fixed $t\in V(T)$, there are at most $2^{w+1}$ possible choices for $X\subseteq B_t$, and 
	for fixed $X\subseteq B_t$, there are at most $d^{w+1}$ possible functions $L$.
	Furthermore, the size of $\cF(t,X,L)$ is bounded by $2^{\mathcal{O}(wd^2)}$.
	Thus, there are $\mathcal{O}(n\cdot k\cdot \max(2,d)^{w+1}\cdot 2^{\mathcal{O}(wd^2)})$ tables.
	
	In summary, the algorithm runs in time
	$\mathcal{O}(n\cdot k\cdot \max(2,d)^{w+1})\cdot 2^{\mathcal{O}(wd^2)}\cdot k=2^{\mathcal{O}(wd^2)}k^2n$.
\end{proof}

	We finish this section with a few remarks regarding \textsc{Bounded $\cP$-Component Vertex Deletion}.
	For this problem, we think of graphs as labeled graphs where each component consists of vertices with distinct labels from $1$ to $d$.
	Let $\Comp (G,S)$ be the set of components of $G[S]$.
	For such a graph $(G,S)$, we define a `characteristic' as a pair $(g, h)$ of functions 
	$g:\Comp (G,S) \rightarrow \cU_d$ and $h:\Comp (G,S) \rightarrow 2^{[d]}$
	satisfying the following,
	for $C\in \Comp (G,S)$ and the component $H$ of $G$ containing $C$, 
	\begin{enumerate}[(a)]
	\item (label-isomorphism condition) $H$ is partially label-isomorphic to $g(C)$, 
	\item (coincidence condition) for every $C'\in \Comp (G,S)$ where $C'$ is contained in $H$, $g(C')=g(C)$, 
	\item (neighborhood condition) $h(C)=L(N_{H}(V(C)))$, and
	\item (completeness condition) for every $w\in V(H)\setminus S$, 
	$H[N_H[w]]$ is label-isomorphic to $g(C)[N_{g(C)}[z]]$ where $z$ is the vertex in $g(C)$ with label $L(w)$.
	\end{enumerate}
	Then, by following similar, but simpler, arguments,
	one can also prove that \textsc{Bounded $\cP$-Component Vertex Deletion} can be solved in time $2^{\mathcal{O}(wd^2)}k^2n$.
	We omit the details.

\mainthma*

\section{Lower bound for fixed $d$}\label{sec:lowerbound}
We showed that \BPCVD and \BPBVD admit single-exponential time algorithms parameterized by treewidth, when $\cP$ is a class of chordal graphs. 
We now establish that, assuming the ETH, this is no longer the case when $\cP$ contains a graph that is not chordal.

In the \kkIS problem, one is given a graph $G=([k] \times [k],E)$ over the $k^2$ vertices of a \emph{$k$-by-$k$ grid}.
We denote by $\p{i}{j}$ with $i,j \in [k]$ the vertex of $G$ in the $i$-th \emph{row} and $j$-th \emph{column}.
The goal is to find an independent set of size $k$ in $G$ that contains exactly one vertex in each row.
The \PkkIS problem is similar but with the additional constraint that the independent set should also contain exactly one vertex per column.

\begin{theorem}
\begin{enumerate}[(1)]
\item     Let $\cP$ be a block-hereditary class of graphs that is polynomial-time recognizable.
	If $\cP$ contains the cycle graph on $\ell \ge 4$ vertices, then \BPBVD is not solvable in time $2^{o(w\log w)}n^{\cO(1)}$ on graphs with $n$ vertices and treewidth at most $w$ 
	even for fixed $d=\ell$, unless the ETH fails.
\item 
Let $\cP$ be a hereditary class of graphs that is polynomial-time recognizable.
If $\cP$ contains the cycle graph on $\ell \ge 4$ vertices, then \BPCVD is not solvable in time $2^{o(w\log w)}n^{\cO(1)}$ on graphs with $n$ vertices and treewidth at most $w$ even for fixed $d=\ell$, unless the ETH fails.
\end{enumerate} 
\end{theorem}

    \begin{proof}
      We reduce from \PkkIS which, like \PkkC, cannot be solved in time $2^{o(k \log k)}k^{\cO(1)}$ unless the ETH fails \cite{LokshtanovMS11}.
Let $G=([k] \times [k],E)$ be an instance of \PkkIS.
We assume that $\forall h,i,j \in [k]$ with $h \neq i$, $\p{i}{j}\p{h}{j} \in E$.
Adding these edges does not change the \YES- and \NO-instances, but has the virtue of making \PkkIS equivalent to \kkIS.
We also assume that $\forall h,i,j \in [k]$, $\p{i}{j}\p{i}{h} \in E$, since at most one of $\p{i}{j}$ and $\p{i}{h}$ can be in a given solution.
Let $m := |E| = \cO(k^4)$ be the number of edges of $G$.

\medskip

\textbf{Outline.}
We build two almost identical graphs $G'=(V',E')$ and $G''=(V',E'')$ with treewidth %
at most $(3d+4)k+6d-5=\cO(k)$, and $((3d-2)k^2+2k)m$ vertices, such that the following three conditions are equivalent:
\begin{enumerate}
  \item $G$ has an independent set of size $k$ with one vertex per row of $G$.\label{cond1}
  \item There is a set $S \subseteq V'$ of size at most $(3d-2)k(k-1)m$ such that each component of $G'-S$ has size at most $d$ and belongs to $\cP$.\label{cond2}
  \item There is a set $S \subseteq V'$ of size at most $(3d-2)k(k-1)m$ such that each block of $G''-S$ has size at most $d$ and belongs to $\cP$.\label{cond3}
\end{enumerate}

The overall construction of $G'$ and $G''$ will display $m$ \emph{almost} copies of the encoding of an \emph{edgeless} $G$ arranged in a cycle.
Each copy embeds one distinct edge of $G$.
The point of having the information of $G$ distilled edge by edge in $G'$ and $G''$ is to control the treewidth.
This general idea originates from a paper of Lokshtanov et al.~\cite{LokshtanovMS11a}.

\begin{figure}[h!]
\centering
\begin{tikzpicture}

\foreach \i in {1,2,3}{
\begin{scope}[xshift=3 * \i cm]
\node (s\i) at (0,0) [draw,rounded corners,thick,minimum width=0.6cm,minimum height=0.6cm] {$S^{e_\i}$} ;
\node (h\i) at (1.5,0) [draw,thick,minimum width=1.2cm,minimum height=1.2cm] {$H^{e_\i}$} ;
\end{scope}
\draw (s\i) -- (h\i) ;
}
\draw (h1) -- (s2) ;
\draw (h2) -- (s3) ;
\draw[dashed] (h3.east) --++(0.4,0) ;

\begin{scope}[xshift=12.5 cm]
\node (sm) at (0,0) [draw,rounded corners,thick,minimum width=0.6cm,minimum height=0.6cm] {$S^{e_m}$} ;
\node (hm) at (1.5,0) [draw,thick,minimum width=1.2cm,minimum height=1.2cm] {$H^{e_m}$} ;
\end{scope}
\draw (sm) -- (hm) ;
\draw[dashed] (sm.west) --++(-0.4,0) ;

\draw[rounded corners] (s1.north) --++(0,0.6) --++(11,0) -- (hm.north) ;

\draw[very thick,red] (4.5,-0.4) --++(0.5,0.3) ; 
\draw[very thick,red] (7,-0.2) --++(0.2,0.6) ; 
\draw[very thick,red] (10.7,-0.4) --++(0.3,0.5) ; 
\draw[very thick,red] (13.5,-0.3) --++(0.5,-0.2) ; 

\end{tikzpicture}
\caption{A high-level schematic of $G'$ and $G''$. The $H^{e_i}$s only differ by a constant number of edges (in red/light gray) that encode their edge $e_i$ of $G$.}
\label{fig:overall-schema}
\end{figure} 

\medskip

\textbf{Construction.}
We first describe $G'$.
As a slight abuse of notation, a gadget (and, more generally, a subpart of the construction) may refer to either a subset of vertices or to an induced subgraph.    
For each $e=\p{i^e}{j^e}\p{i'^e}{j'^e} \in E$, we detail the internal construction of $H^e$ and $S^e$ of \cref{fig:overall-schema} and how they are linked to one another.
Each vertex $v=\p{i}{j}$ of $G$ is represented by a gadget $H^e(v)$ on $3d-2$ vertices in $G'$: a path on $d-3$ vertices whose endpoints are $v^e_{-a}$ and $v^e_{-b}$, an isolated vertex $v^e_+$, and two disjoint cycles of length $d$.
Observe that if $d=4$, then $v^e_{-a}$ and $v^e_{-b}$ is the same vertex.
We add all the edges between $H^e(\p{i}{j})$ and $H^e(\p{i}{j'})$ for $i,j,j' \in [k]$ with $j \neq j'$.
We also add all the edges between $H^e(\p{i^e}{j^e})$ and $H^e(\p{i'^e}{j'^e})$.
We call $H^e$ the graph induced by the union of every $H^e(v)$, for $v \in V(G)$.
The \emph{row/column selector} gadget $S^e$ consists of a set $S^e_r$ of $k$ vertices with one vertex $r^e_i$ for each row index $i \in [k]$, and a set $S^e_c$ of $k$ vertices with one vertex $c^e_j$ for each column index $j \in [k]$.
The gadget $S^e$ forms an independent set of size $2k$.
We arbitrarily number the edges of $G$: $e_1, e_2, \ldots, e_m$.
For each $h \in [m]$ and $v=\p{i}{j} \in V$, we link $v^{e_h}_{-a}$ to $r^{e_h}_i$ (the row index of $v$) and $v^{e_h}_{-b}$ to $c^{e_h}_j$ (the column index of $v$).
We also link, for every $h \in [m-1]$, $v^{e_h}_+$ to $r^{e_{h+1}}_i$ and to $c^{e_{h+1}}_j$, and $v^{e_m}_+$ to $r^{e_1}_i$ and to $c^{e_1}_j$.
That concludes the construction (see \cref{fig:lowerbound-tw-overall}).
To obtain $G''$ from $G'$, we add the edges $c^{e_h}_jc^{e_h}_{j+1}$ for every $h \in [m]$ and $j \in [k-1]$. 
We ask for a deletion set $S$ of size $s := (3d-2)k(k-1)m$.

\begin{figure}
\centering
\resizebox{300pt}{!}{
\begin{tikzpicture}
\def\m{2}
\def\k{3}
\def\o{0.2}

\foreach \l in {1,...,\m}{ 

\begin{scope}[yshift=6.5 * \l cm]
\foreach \i in {1,...,\k}{
\node[draw,circle] (rr\l\i) at (7 + \i,-2) {} ;
\node[draw,circle] (cc\l\i) at (12.5 + \i,-2) {} ;
}
\foreach \i [count=\h from 1] in {2,...,\k}{
\draw[very thick,opacity=0.25] (cc\l\h) -- (cc\l\i) ;
}
\node [draw,rectangle,very thick,rounded corners,fit=(rr\l1)(rr\l\k)] (rrr\l) {} ;
\node [draw,rectangle,very thick,rounded corners,fit=(cc\l1)(cc\l\k)] (ccc\l) {} ;

\node[left of=rr\l1] {row index~~~~~~~} ;
\node[right of=cc\l\k] {~~~~~~~~~~~column index} ;
\node (rsel\l) at (9,-2.6) {$S^{e_\l}_r$} ;
\node (csel\l) at (14.5,-2.6) {$S^{e_\l}_c$} ;

\foreach \j in {1,...,\k}{
\begin{scope}[xshift=5.5 * \j cm]
  \foreach \i in {1,...,\k}{
    \begin{scope}[xshift=1.6 * \i cm]
      \node[draw,circle] (v\i\j\l1a) at (-0.35,0) {} ;
      \node[draw,circle] (v\i\j\l1b) at (0.35,0) {} ;
      
      \node[draw,circle] (v\i\j\l2) at (-0.35,0.45) {} ;
      \node[draw,circle] (v\i\j\l3) at (0.35,0.45) {} ;
      \node[draw,circle] (v\i\j\l4) at (-0.35,0.95) {} ;
      \node[draw,circle] (v\i\j\l5) at (0.35,0.95) {} ;

      \node[draw,circle] (v\i\j\l2b) at (-0.35,1.45) {} ;
      \node[draw,circle] (v\i\j\l3b) at (0.35,1.45) {} ;
      \node[draw,circle] (v\i\j\l4b) at (-0.35,1.95) {} ;
      \node[draw,circle] (v\i\j\l5b) at (0.35,1.95) {} ;

      \node[draw,circle] (v\i\j\l6) at (0,2.4) {} ;

      \draw[dashed, thick] (v\i\j\l1a) -- (v\i\j\l1b) ;    
      \draw[densely dotted, very thick] (v\i\j\l2) -- (v\i\j\l3) ;
      \draw (v\i\j\l3) -- (v\i\j\l5) -- (v\i\j\l4) -- (v\i\j\l2) ;
      \draw[densely dotted, very thick] (v\i\j\l2b) -- (v\i\j\l3b) ;
      \draw (v\i\j\l3b) -- (v\i\j\l5b) -- (v\i\j\l4b) -- (v\i\j\l2b) ;

      \draw (cc\l\i) -- (v\i\j\l1b) ;          
      \draw (v\i\j\l1a) -- (rr\l\j) ;

      \node [draw,rectangle,rounded corners,fit=(v\i\j\l1a)(v\i\j\l1b)(v\i\j\l2)(v\i\j\l3)(v\i\j\l4)(v\i\j\l5)(v\i\j\l2b)(v\i\j\l3b)(v\i\j\l4b)(v\i\j\l5b)(v\i\j\l6)] (r\i\j\l) {} ;
    \end{scope}
}
\draw[very thick] (r1\j\l) to [bend right] (r2\j\l) ;
\draw[very thick] (r2\j\l) to [bend right] (r3\j\l) ;
\draw[very thick] (r1\j\l) to [bend left] (r3\j\l) ;
\node [draw,rectangle,very thick,fit=(r1\j\l)(r\k\j\l)] (b\j\l) {} ;
\end{scope}
}
\node[below of=v11\l1a] {$H^{e_\l}$} ;
\node[draw,rectangle,fit=(b1\l)(b\k\l)] (d\l) {} ;
\end{scope}
}

\node[above of=v2226] {\textbf{$\vdots$}} ;

\foreach \j in {1,...,\k}{
  \foreach \i in {1,...,\k}{
    \draw (v\i\j16) -- (cc2\i) ;
    \draw (v\i\j16) -- (rr2\j) ;
}
}

\draw(r212) edge [very thick,red,bend right=20] (r122) ;
\draw(r211) edge [very thick,red,bend right=20] (r221) ;

\end{tikzpicture}
}
\caption{The overall picture of $G'$ and $G''$ with $k=3$. Dotted edges are subdivided $d-4$ times. In particular, if $d=4$, they are simply edges.
Dashed edges are subdivided $d-5$ times.
In particular, if $d=4$, the two endpoints are in fact a single vertex. 
Edges between two boxes link each vertex of one box to each vertex of the other box. The gray edges in the column selectors $S^{e_h}_c$ are only present in $G''$.}
\label{fig:lowerbound-tw-overall}
\end{figure}

\medskip

\textbf{Treewidth of $G'$ and $G''$.} 
We claim that the pathwidth, and hence treewidth, of $G'$ and $G''$ are bounded by $(3d+4)k+6d-5$.
For any edge $e \in E$, we set $H(e) := H^e(\p{i^e}{j^e}) \cup H^e(\p{i'^e}{j'^e})$.
For any $h \in [m-1]$, we set $\tilde{S}_h := S^{e_1} \cup S^{e_h} \cup S^{e_{h+1}}$, and  $\tilde{S}_m := S^{e_1} \cup S^{e_m}$.
For each $e \in E$, and $i \in [k]$, $H^e(i)$ denotes the union of the $H^e(v)$ for all vertices $v$ of the $i$-th row.
Here is a path decomposition of $G'$ and $G''$ where the bags contain no more than $(3d+4)k+6d-4$ vertices:

\begin{center}
$\tilde{S}_1 \cup H(e_1) \cup H^{e_1}(1) \rightarrow \tilde{S}_1 \cup H(e_1) \cup H^{e_1}(2) \rightarrow \ldots \rightarrow \tilde{S}_1 \cup H(e_1) \cup H^{e_1}(k) \rightarrow$ \\
$\tilde{S}_2 \cup H(e_2) \cup H^{e_2}(1) \rightarrow \tilde{S}_2 \cup H(e_2) \cup H^{e_2}(2) \rightarrow \ldots \rightarrow \tilde{S}_2 \cup H(e_2) \cup H^{e_2}(k) \rightarrow$\\ 
$\vdots$\\
$\tilde{S}_m \cup H(e_m) \cup H^{e_m}(1) \rightarrow \tilde{S}_m \cup H(e_m) \cup H^{e_m}(2) \rightarrow \ldots \rightarrow \tilde{S}_m \cup H(e_m) \cup H^{e_m}(k)$. 
\end{center} 
As, for any $h \in [m]$, $|\tilde{S}_h| \leqslant 6k$, $|H(e_h)|=2(3d-2)$, and $|H^{e_h}(i)|\leqslant (3d-2)k$ for any $i \in [k]$, the size of a bag is bounded by $\max_{h \in [m],i \in [k]}|\tilde{S}_h \cup H(e_h) \cup H^{e_h}(i)| \leqslant 6k+2(3d-2)+(3d-2)k=(3d+4)k+6d-4$.

\medskip

\textbf{Correctness.}
We first show \ref{cond1} $\Rightarrow$ \ref{cond2}.
Let us assume that there is an independent set $I := \{v_1=\p{1}{j_1},v_2=\p{2}{j_2},\ldots,v_k=\p{k}{j_k}\}$ in $G$.
We define the deletion set $S \subseteq V'$ as follows.
For each $e \in E$ and $i \in [k]$, we delete all of $H^e(i)$ except $H^e(v_i)$.
The cardinality of $S$ adds up to a total of $(|H^e(i)|-|H^e(v_i)|)mk=((3d-2)k-3d+2)mk=(3d-2)k(k-1)m=s$ vertices.
We claim that all the components of $G'-S$ are isomorphic to $C_d$, and belong to $\cP$ since $d \geqslant 4$. 
First, we observe that the $C_d$s inside any $H^e(v_i)$, for $e \in E$ and $i \in [k]$, are isolated in $G'-S$.
Indeed, $H^e(v_i)$ is the only remaining $H^e(v)$ from $H^e(i)$.
So, it might only be linked to $H^e(v_j)$ with some $j \neq i \in [k]$.
But this would imply that $v_iv_j \in E$, contradicting that $I$ is an independent set.
Besides those $C_d$s contained in the $H^e(v_i)$s, we claim that the rest of $G'-S$ is $mk$ disjoint $C_d$s formed with the vertices ${v_p}_+^{e_{h-1}}$, $r^{e_h}_p$, $c^{e_h}_{j_p}$, and the path $P_{v_p}^{e_h}$ between ${v_p}_{-a}^{e_h}$ and ${v_p}_{-b}^{e_h}$, for any $h \in [m]$ and $p \in [k]$ (with the convention that $e_0=e_m$).
Indeed, let us recall that $\{j_1,j_2,\ldots,j_k\}=[k]$.
Therefore, $\{{v_p}_+^{e_{h-1}}, r^{e_h}_p, c^{e_h}_{j_p}\} \cup P_{v_p}^{e_h}$ is a family of $mk$ pairwise disjoint sets of size $d$.
The vertices $r^{e_h}_p$ and $c^{e_h}_{j_p}$ have degree $2$ in $G'-S$ since $I$ contains only one vertex in the $p$-th row of $G$, and $I$ contains only one vertex in the $j_p$-th column; and in both cases this vertex is $v_p$. 
The vertex ${v_p}_+^{e_{h-1}}$ and the vertices of $P_{v_p}^{e_h}$ also have degree $2$ in $G'-S$.
Therefore, $G'-S$ is a disjoint union of $C_d$s.
The implication \ref{cond1} $\Rightarrow$ \ref{cond3} is derived similarly.
We now claim that, with the same deletion set $S$, all the blocks of $G''-S$ are isomorphic to $C_d$ or $K_2$.
As $\mathcal P$ is a hereditary class that contains the induced cycle of length $d \geqslant 4$, it holds that $K_2 \in \mathcal P$. 
We still have the property that the $C_d$s within any $H^e(v_i)$ are isolated in $G''-S$.
Now, the slight difference is that $\{{v_p}_+^{e_{h-1}}, r^{e_h}_p, c^{e_h}_{j_p}\} \cup P_{v_p}^{e_h}$ induces $m$ disjoint $\mathcal C_{k,d}$s in $G''-S$, where $\mathcal C_{k,d}$ is the graph obtained by linking each of the $k$ vertices of a path to the two endpoints of a path on $d-1$ vertices.
Informally, $\mathcal C_{k,d}$ corresponds to $k$ $C_d$s attached to different vertices of a path on $k$ vertices. 
In this case, the path consists of the vertices $c^{e_h}_1, c^{e_h}_2, \ldots, c^{e_h}_k$.
Finally, we observe that the blocks of $\mathcal C_{k,d}$ are $k$ $C_d$s and $k-1$ $K_2$s.

We now show that \ref{cond2} $\Rightarrow$ \ref{cond1} and \ref{cond3} $\Rightarrow$ \ref{cond1}.
We assume that there is a set $S \subseteq V'$ of size at most $s$ such that all the blocks of $G''-S$ (resp.~$G'-S$) have size at most $d$.
We note that this is implied by \ref{cond3} (resp.~by a weaker assumption than \ref{cond2}).
The first property we show on $S$ is that, for any $e \in E$ and $i \in [k]$, $|H^e(i) \cap S| \geqslant (3d-2)(k-1)$.
In other words, there are at most $3d-2$ vertices of $H^e(i)$ remaining in $G''-S$ (or $G'-S$).
Assume, for the sake of contradiction, that $H^e(i)-S$ contains at least $3d-1$ vertices.
Observe that $H^e(i)-S$ cannot contain at least one vertex from three distinct $H^e(u)$, $H^e(v)$, and $H^e(w)$ (with $u$, $v$ and $w$ in the $i$-th row of $G$), since then $H^e(i)-S$ would be $2$-connected (and of size $>d$).
For the same reason, $H^e(i)-S$ cannot contain at least two vertices in $H^e(u)$ and at least two vertices in another $H^e(v)$.
Therefore, the only way of fitting $3d-1$ vertices in $H^e(i)-S$ is the $3d-2$ vertices of an $H^e(u)$ plus one vertex from some other $H^e(v)$.
But then, this vertex of $H^e(v)$ would form, together with one $C_d$ of $H^e(u)$, a $2$-connected subgraph of $G''-S$ (or $G'-S$) of size $d+1$.
Now, we know that $|H^e(i) \cap S| \geqslant (3d-2)(k-1)$.
As there are precisely $mk$ sets $H^e(i)$ in $G'$ (and they are disjoint), it further holds that $|H^e(i) \cap S| = (3d-2)(k-1)$, since otherwise $S$ would contain strictly more than $s=(3d-2)k(k-1)m$ vertices.
Thus, $H^e(i)-S$ contains exactly $3d-2$ vertices.
By the previous remarks, $H^e(i)-S$ can only consist of the $3d-2$ vertices of the same $H^e(u)$ or $3d-3$ vertices of $H^e(u)$ plus one vertex from another $H^e(v)$.
In fact, the latter case is not possible, since the vertex of $H^e(v)$ would form, with at least one remaining $C_d$ of the $3d-3$ vertices of $H^e(u)$, a $2$-connected subgraph of $G''-S$ (or $G'-S$) of size $d+1$. 
Note that this is why we needed two disjoint $C_d$s in the construction instead of just one.
So far, we have proved that, assuming \ref{cond2} or \ref{cond3}, for any $e \in E$ and $i \in [k]$, $H^e(i) \cap S=H^e(v(i,e))$ for some vertex $v(i,e)$ of the $i$-th row of $G$, and for any $e \in E$, $S^e \cap S=\emptyset$.

The second part of the proof consists of showing that $v(i,e)$ does not depend on $e$.
Formally, we want to show that there is a $v_i$ such that, for any $e \in E$, $v(i,e)=v_i$.
Observe that it is enough to derive that, for any $h \in [m]$, $v(i,e_h)=v(i,e_{h+1})$ (with $e_{m+1}=e_1$).
Let $j \in [k]$ (resp.~$j' \in [k]$) be the column of $v(i,e_h)$ (resp.~$v(i,e_{h+1})$) in $G$.
We first assume \ref{cond2}.
For any $h \in [m]$, ${v(i,e_h)}_+^{e_h}$, $r^{e_{h+1}}_i$, $c^{e_{h+1}}_{j'}$, $c^{e_{h+1}}_j$ plus the path $P_{v(i,{e_{h+1}})}^{e_{h+1}}$ (between ${v(i,{e_{h+1}})}_{-a}^{e_{h+1}}$ and ${v(i,{e_{h+1}})}_{-b}^{e_{h+1}}$) induces a path (in particular, a connected subgraph) of size $d+1$ in $G''-S$, unless $j=j'$ (with $e_{m+1}=e_1$). 
Therefore, $j=j'$. 
As $v(i,e_h)$ and $v(i,e_{h+1})$ have the same column $j$ and the same row $i$ in $G$, $v(i,e_h)=v(i,e_{h+1})$.

Now, we assume \ref{cond3}.
For any $h \in [m]$, ${v(i,e_h)}_+^{e_h}$, $r^{e_{h+1}}_i, {v(i,e_{h+1})}_{-a}^{e_{h+1}}, {v(i,e_{h+1})}_{-b}^{e_{h+1}}, c^{e_{h+1}}_{j'}, c^{e_{h+1}}_{j'+1},$ $ \dotsc,$ $ c^{e_{h+1}}_{j-1}, c^{e_{h+1}}_j$ if $j \geqslant j'$ (resp.~$c^{e_{h+1}}_{j'-1}, \dotsc,$ $ c^{e_{h+1}}_{j+1}, c^{e_{h+1}}_j$ if $j \leqslant j'$) plus the path between ${v(i,e_{h+1})}_{-a}^{e_{h+1}}$ and ${v(i,e_{h+1})}_{-b}^{e_{h+1}}$ induces a cycle (that is, a $2$-connected subgraph) of length at least $d+1$ in $G''-S$, unless $j=j'$ (with $e_{m+1}=e_1$). 
Again, $j=j'$; and the vertices $v(i,e_h)$ and $v(i,e_{h+1})$ have the same column and the same row in $G$, which implies that $v(i,e_h)=v(i,e_{h+1})$.
In both cases (\ref{cond2} or \ref{cond3}), we can now safely define $v_i := v(i,e)$. 

We finally claim that $\{v_1,v_2,\ldots,v_k\}$ is an independent set in $G$ (and for each $i \in [k]$, $v_i$ is in the $i$-th row).
Indeed, if there were an edge $e=v_iv_j \in E$ for some $i \neq j \in [k]$, then $H^e(v_i) \cup H^e(v_j)$ would induce a $2$-connected subgraph of size $2(3d-2)>d$ (since $d \geqslant 4$) in $G''-S$ (or $G'-S$).

That finishes the proof that \ref{cond1} $\Leftrightarrow$ \ref{cond2} $\Leftrightarrow$ \ref{cond3}.
Therefore, for any fixed integer $d \geqslant 4$, an algorithm running in time $2^{o(w \log w)}|V'|^{\cO(1)}$ for either \BPCVD or \BPBVD on graphs of treewidth $w$ with $C_d \in \mathcal P$ would also solve \PkkIS in time $$2^{o(((3d+4)k+6d-5) \log ((3d+4)k+6d-5))}(((3d-2)k^2+2k)m)^{\cO(1)}=2^{o(k \log k)}k^{\cO(1)},$$ which contradicts the ETH.
\end{proof}

\section{Hardness and lower bounds, when $d$ is not fixed}\label{sec:w1hardness}

In this section, we prove \cref{hardness-intro}.  %
Our first reduction is from the following problem:

\ProblemQ{\MC}{$k$}{A graph $G$, a positive integer $k$, and a partition $(V_1,V_2,\dotsc,V_k)$ of $V(G)$.}
{Is there a $k$-clique $X$ of $G$ such that $|X \cap V_i| = 1$ for each $i \in [k]$?}

\noindent
We call a set $V_i$, for some $i \in [k]$, a \emph{color class}.
The problem \MC is known to be $W[1]$-complete (see, for example, \cite{Cygan15}),
and it is clear that this remains true under the assumption that there are no edges between vertices of the same color class.
Moreover, %
we may assume
that each color class has the same size,
and between every distinct pair of color classes we have the same %
number of edges%
~\cite{Fellows2011}.
We say that $X \subseteq V(G)$ is a \emph{multicolored $k$-clique} if $X$ is a $k$-clique such that $|X \cap V_i| = 1$ for each $i \in [k]$.

\begin{theorem}
  \label{hardness}
  \BPCVD is $W[1]$-hard parameterized by the combined parameter $(w,k)$, when $\cP$ contains all chordal graphs.
\end{theorem}
    Before proving this theorem, we describe the reduction used in the proof.
Given an instance $(G,k,(V_1,\dotsc,V_k))$ of \MC, where each color class has size~$t$, we construct a graph $G'$ such that $G$ has a multicolored $k$-clique if and only if there exists a set $S \subseteq V(G')$ of size at most $k'$ such that each component of $G'-S$ consists of at most $d$ vertices,
where $k' = 3 \binom{k+1}{2} -6$ and $d = 3t^2+3t+3$, and
the treewidth of $G'$ is bounded above by $54k-69$. 
Each component of $G'-S$ is a chordal graph, so we obtain a reduction to \BPCVD whenever $\cP$ contains all chordal graphs.
We may assume that $k \geq 2$.

Let $V_i = \{v_i^1, v_i^2, \dotsc, v_i^t\}$, for each $i \in [k]$.
For $i,j \in [k]$ with $i < j$, we denote the set of edges in $G[V_i \cup V_j]$ by $E_{i,j}$, and
we may assume that %
$|E_{i,j}| = p$, say. %
We construct $G'$ from several gadgets; namely, an ``edge-encoding gadget'' $G_{i,j}$ for each $i,j \in [k]$ with $i<j$, which represents the set $E_{i,j}$, linked together by copies of one of the ``propagator gadgets'', $H_i$ or $\tilde{H_i}$, which collectively represent the color class~$V_i$ %
for some $i \in [k]$.
We also have a gadget $G_{i,i}$, for each $i \in [2,k-2]$, which ensures that the vertex selection in the $H_i$ gadgets also propagates to the $\tilde{H_i}$ gadgets.

Each gadget encodes a sequence of $z+1$ integers $X=\left<x_0,x_1,\dotsc,x_{z}\right>$, where $x_0 \ge 3$, and $x_s-x_{s-1} \ge 3$ for each $s \in [z]$.
We denote such a gadget $G(X)$ and call it a \emph{gadget of $G'$ of order~$z$}.
It is constructed as follows.
First, set $$(d_0,d_1,d_2,\dotsc,d_z) := (x_0, x_1-x_0, x_2-x_1, \dotsc, x_z-x_{z-1}).$$
Note that $d_q \ge 3$ for every $q \in [0,z]$.
For each $q \in [0,z]$, we now define a graph $P_q$ which resembles a ``thickened path''.
For $q \in [1,z-1]$, let $P_q$ be the graph on the vertex set $\{w_{q,1},w_{q,2},\dotsc,w_{q,d_q-1}\}$ with edges between distinct $w_{q,d}$ and $w_{q,d'}$ if and only if $|d-d'| \in [2]$.
For $q \in \{0,z\}$, let $P_q$ be the graph on the vertex set $\{w_{q,1},w_{q,2},\dotsc,w_{q,d_q}\}$ with edges between distinct $w_{q,d}$ and $w_{q,d'}$ if and only if $|d-d'| \in [3]$.
For each $q \in [z]$, we add a vertex $u_q$ adjacent to $w_{q-1,1}$, $w_{q-1,2}$, $w_{q,1}$, and $w_{q,2}$.
The resulting graph $G(X)$ consists of $d_0 + \big(\sum_{q \in [z-1]} (d_q - 1)\big) + d_z + z  = \big(\sum_{q \in [0,z]} d_q\big)+1 = x_{z}+1$
vertices, and, for $q \in [z]$, the graph obtained by deleting $u_q$ has two components: one of size $x_q$, and the other of size $x_{z}- x_q$.
Let $B := \{w_{0,1}, w_{0,2}, w_{0,3}\}$ and $D := \{w_{z,1}, w_{z,2}, w_{z,3}\}$.
Since we will use several copies of this gadget, 
we usually refer to $P_q$ as $P_q(G(X))$,
a vertex $v \in V(G(X))$ as $v(G(X))$,
and $B$ or $D$ as $B(G(X))$ or $D(G(X))$ respectively;
but we sometimes omit the ``$(G(X))$'' when there is no ambiguity.
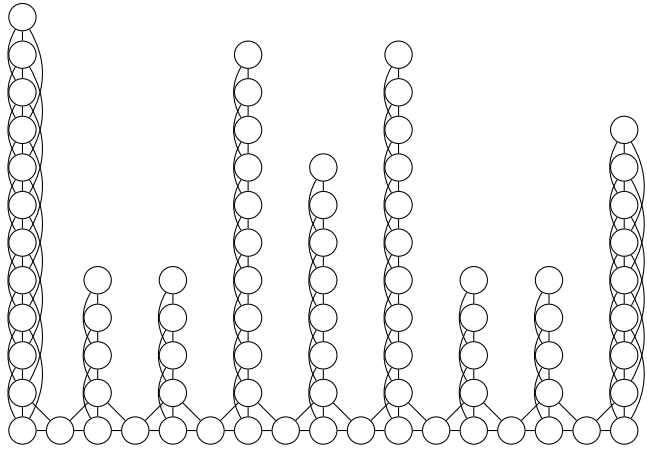
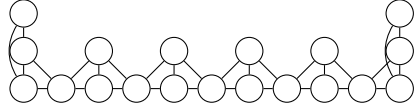
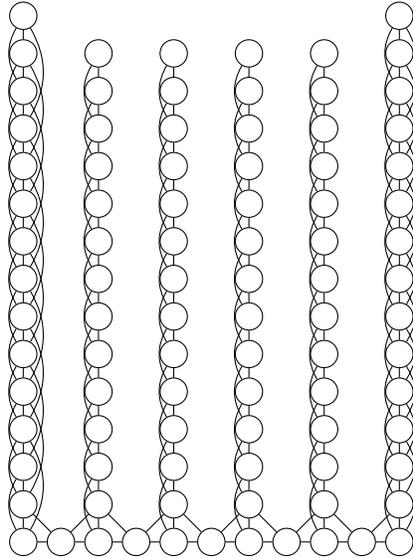
\begin{figure}
\captionsetup[subfigure]{width=0.66\textwidth}
\begin{minipage}{\linewidth}
\centering
  \begin{tikzpicture}
\def\d{9}
\pgfmathsetmacro{\u}{\d-1}
\foreach \i in {1,...,\u}{
\node[draw,circle] (a1\i) at (\i,0) {} ;
\node[draw,circle] (a2\i) at (\i+0.5,0) {} ;
\node[draw,circle] (p\i0) at (\i,0.5) {} ;
\draw (a1\i) -- (a2\i) -- (p\i0) -- (a1\i) ;
}
\node[draw,circle] (a1\d) at (\d,0) {} ;
\node[draw,circle] (p\d0) at (\d,0.5) {} ;
\draw (a1\d) -- (p\d0) ;

\foreach \i [count=\j from 2] in {1,...,\u}{
\draw (a1\j) -- (a2\i) -- (p\j0) ;
}

\foreach \i in {1,...,\u}{
\pgfmathsetmacro{\h}{3 * mod(\i * \i+1,5)+3}
\foreach \s in {1,...,\h}{
\node[draw,circle] (p\i\s) at (\i,0.5 * \s + 0.5) {} ;
}
\foreach \s [count=\j from 0] in {1,...,\h}{
\draw (p\i\s) -- (p\i\j) ;
}
\foreach \s [count=\j from 0] in {2,...,\h}{
\draw (p\i\s) [bend right=30] to (p\i\j) ;
}
\draw(p\i1) [bend right=30] to (a1\i) ;
}

\foreach \s in {1,...,7}{
\node[draw,circle] (p\d\s) at (\d,0.5 * \s + 0.5) {} ;
}
\foreach \s [count=\j from 0] in {1,...,7}{
\draw (p\d\s) -- (p\d\j) ;
}
\foreach \s [count=\j from 0] in {2,...,7}{
\draw (p\d\s) [bend right=30] to (p\d\j) ;
}
\draw(p\d1) [bend right=30] to (a1\d) ;

\foreach \i in {1}{
\pgfmathsetmacro{\h}{3 * mod(\i * \i+1,5)+3}
\foreach \s [count=\j from 0] in {3,...,\h}{
\draw (p\i\s) [bend left=30] to (p\i\j) ;
}
\draw(p\i2) [bend left=30] to (a1\i) ;
}

\foreach \s [count=\j from 0] in {3,...,7}{
\draw (p\d\s) [bend left=30] to (p\d\j) ;
}
\draw(p\d2) [bend left=30] to (a1\d) ;

\node[draw,circle] (p110) at (1,5.5) {} ;
\draw (p110) -- (p19) ;
\draw (p110) [bend right=30] to (p18) ;
\draw (p110) [bend left=30] to (p17) ;

\end{tikzpicture}
\subcaption{The edge encoding gadget $G_{i,j}$ (with $t=5$) for the edges $\{v_i^1v_j^4, v_i^2v_j^1, v_i^2v_j^3, v_i^3v_j^2, v_i^3v_j^5, v_i^4v_j^4, v_i^5v_j^1, v_i^5v_j^3\}$, encoded as $\langle 12, 18, 24, 36, 45, 57, 63, 69, 78\rangle$.}
\label{twlb-buildingblock-edges}
\end{minipage}

\begin{minipage}{\linewidth}
    \centering
\begin{tikzpicture}
\def\d{6}
\pgfmathsetmacro{\u}{\d-1}
\foreach \i in {1,...,\u}{
\node[draw,circle] (a1\i) at (\i,0) {} ;
\node[draw,circle] (a2\i) at (\i+0.5,0) {} ;
\node[draw,circle] (a3\i) at (\i,0.5) {} ;
\draw (a1\i) -- (a2\i) -- (a3\i) -- (a1\i) ;
}
\node[draw,circle] (a1\d) at (\d,0) {} ;
\node[draw,circle] (a3\d) at (\d,0.5) {} ;
\draw (a1\d) -- (a3\d) ;

\foreach \i [count=\j from 2] in {1,...,\u}{
\draw (a1\j) -- (a2\i) -- (a3\j) ;
}

\node[draw,circle] (p11) at (1,1) {} ;
\node[draw,circle] (p61) at (6,1) {} ;

\draw (p11) -- (a31) ;
\draw (p61) -- (a36) ;

\draw (p11) [bend right=30] to (a11) ;
\draw (p61) [bend right=30] to (a16) ;
\end{tikzpicture}
    \subcaption{A propagator gadget $H_j$ (with $t=5$), which will be linked to edge encoding gadgets $G_{i,j}$ with $i \le j$.}
\label{twlb-buildingblock-vertices}
\end{minipage}

\begin{minipage}{\linewidth}
    \centering
\begin{tikzpicture}
\def\d{6}
\def\h{13}
\pgfmathsetmacro{\u}{\d-1}
\pgfmathsetmacro{\v}{\h-1}
\foreach \i in {1,...,\u}{
\node[draw,circle] (a1\i) at (\i,0) {} ;
\node[draw,circle] (a2\i) at (\i+0.5,0) {} ;
\node[draw,circle] (a3\i) at (\i,0.5) {} ;
\draw (a1\i) -- (a2\i) -- (a3\i) -- (a1\i) ;
}
\node[draw,circle] (a1\d) at (\d,0) {} ;
\node[draw,circle] (a3\d) at (\d,0.5) {} ;
\draw (a1\d) -- (a3\d) ;

\foreach \i [count=\j from 2] in {1,...,\u}{
\draw (a1\j) -- (a2\i) -- (a3\j) ;
}

\foreach \r in {1,\d}{
\foreach \s in {1,...,\h}{
\node[draw,circle] (p\r\s) at (\r,0.5 * \s + 0.5) {} ;
}
\draw (a3\r) -- (p\r1) ;

\foreach \s [count=\j from 1] in {2,...,\h}{
\draw (p\r\s) -- (p\r\j) ;
}
}

\foreach \i in {2,...,\u}{ 
\foreach \s in {1,...,\v}{
\node[draw,circle] (p\i\s) at (\i,0.5 * \s + 0.5) {} ;
}
\draw (a3\i) -- (p\i1) ;

\foreach \s [count=\j from 1] in {2,...,\v}{
\draw (p\i\s) -- (p\i\j) ;
}
}

\foreach \i in {1,...,\d}{
\foreach \s [count=\j from 1] in {3,...,\v}{
\draw (p\i\s) [bend right=30] to (p\i\j) ;
}
\draw (p\i1) [bend right=30] to (a1\i) ;
\draw (p\i2) [bend right=30] to (a3\i) ;
}

\foreach \i in {1,\d}{
\foreach \s [count=\j from 1] in {4,...,\h}{
\draw (p\i\s) [bend left=30] to (p\i\j) ;
}
\draw (p\i2) [bend left=30] to (a1\i) ;
\draw (p\i3) [bend left=30] to (a3\i) ;
\draw (p\i13) [bend right=30] to (p\i11) ;
}

\end{tikzpicture}
\subcaption{A propagator gadget $\tilde{H_i}$ (with $t=5$), which will be linked to edge encoding gadgets $G_{i,j}$ with $i \le j$.}
\label{twlb-buildingblock-vertices2}
\end{minipage}
\caption{The different uses of the gadgets: the edge encoding gadget and the two kinds of propagator gadgets.}

\label{twlb-buildingblock}
\end{figure}

We now describe the \emph{edge encoding gadget} $G_{i,j}$, for some $i,j \in [k]$ with $i < j$; an example is given in \cref{twlb-buildingblock-edges}.
We can uniquely describe an edge between a vertex in $V_i$ and a vertex in $V_j$ by an ordered pair $(a,b)$, representing the edge $v_i^av_j^b$, where $a, b \in [t]$.
We define an injective function~$\phi$ %
from such a pair to an integer in $\{3,6,\dotsc,3t^2\}$, as given
by $(a,b) \mapsto 3t(a-1) + 3b$.
Thus, the set $\{\phi(a,b) : v_i^av_j^b \in E_{i,j}\}$ uniquely describes the set $E_{i,j}$.
Let %
$(f_{i,j}^0,f_{i,j}^1,\dotsc,f_{i,j}^{p-1})$
be the sequence obtained after ordering the elements of this set %
in increasing order, and let $f_{i,j}^{p} = 3t^2+3$.
Note that $f_{i,j}^0 \ge 3$, and $f_{i,j}^q-f_{i,j}^{q-1} \ge 3$ for each $q \in [p]$.
Finally, we set $G_{i,j} := G\left(\left<f_{i,j}^0,f_{i,j}^1,\dotsc,f_{i,j}^{p}\right>\right)$.

  \sloppy
We define the \emph{propagator gadgets} as
$H_i := G(\left<3,6,\dotsc,3(t+1)\right>)$ and $\tilde{H_i} := G(\left<3t,6t,\dotsc,3(t+1)t\right>)$; see \cref{twlb-buildingblock-vertices,twlb-buildingblock-vertices2}.
Note that these gadgets have size $3(t+1)+1$ and $3t(t+1)+1$, respectively.
For each color class~$V_i$, where $i \in [2,k-1]$, we will take $i$ copies of the gadget $H_i$, and $k-i+1$ copies of $\tilde{H_i}$; whereas for
$i=1$ (or $i=k$), we take $k-1$ copies of $\tilde{H_i}$ (or $H_i$, respectively) only.
Let $\mathcal{H}_i$ denote the set containing the copies of $H_i$, and let $\tilde{\mathcal{H}_i}$ denote the copies of $\tilde{H_i}$.
Note that $|\mathcal{H}_i \cup \tilde{\mathcal{H}}_i| = k+1$ when $i \in [2,k-1]$, and $|\mathcal{H}_i \cup \tilde{\mathcal{H}}_i| = k-1$ when $i \in \{1,k\}$.

Finally, for each $i \in [2,k-2]$, we have a special gadget $G_{i,i} := G\left(\left<\phi(1,1),\phi(2,2),\dotsc,\phi(t,t)\right>\right)$.  Intuitively, this gadget is used to ensure %
the vertex selected in each $H_i \in \mathcal{H}_i$ is the same as in each $\tilde{H_i} \in \tilde{\mathcal{H}_i}$.
However, we also consider $G_{i,i}$ an edge encoding gadget, since it is treated as one in the construction.

\fussy

\begin{figure}[hbt]
\centering
\resizebox{250pt}{!}{
\begin{tikzpicture}

\node[draw,rectangle] (a11) at (0,0) {$G_{1,2}$} ;
\node[draw,rectangle] (a12) at (2,0) {$\tilde{H}_1$} ;
\node[draw,rectangle] (a13) at (4,0) {$G_{1,3}$} ;
\node[draw,rectangle] (a14) at (6,0) {$\tilde{H}_1$} ;
\node[draw,rectangle] (a15) at (8,0) {$G_{1,4}$} ;
\node[draw,rectangle] (a16) at (10,0) {$\tilde{H}_1$} ;

\node[draw,rectangle] (a21) at (0,-2) {$G_{2,2}$} ;
\node[draw,rectangle] (a22) at (2,-2) {$\tilde{H}_2$} ;
\node[draw,rectangle] (a23) at (4,-2) {$G_{2,3}$} ;
\node[draw,rectangle] (a24) at (6,-2) {$\tilde{H}_2$} ;
\node[draw,rectangle] (a25) at (8,-2) {$G_{2,4}$} ;
\node[draw,rectangle] (a26) at (10,-2) {$\tilde{H}_2$} ;

\foreach \k in {1,2}{
\foreach \i [count=\j from 2] in {1,...,5}{
  \draw (a\k\i) -- (a\k\j) ;
}
}

\node[draw,rectangle] (b11) at (0,-1) {$H_2$} ;
\node[draw,rectangle] (b12) at (4,-1) {$H_3$} ;
\node[draw,rectangle] (b13) at (8,-1) {$H_4$} ;

\draw (a11) -- (b11) -- (a21) ;
\draw (a13) -- (b12) -- (a23) ;
\draw (a15) -- (b13) -- (a25) ;

\node[draw,rectangle] (a33) at (4,-4) {$G_{3,3}$} ;
\node[draw,rectangle] (a34) at (6,-4) {$\tilde{H}_3$} ;
\node[draw,rectangle] (a35) at (8,-4) {$G_{3,4}$} ;
\node[draw,rectangle] (a36) at (10,-4) {$\tilde{H}_3$} ;

\draw (a33) -- (a34) -- (a35) -- (a36) ;

\node[draw,rectangle] (b21) at (0,-3) {$H_2$} ;
\node[draw,rectangle] (b22) at (4,-3) {$H_3$} ;
\node[draw,rectangle] (b23) at (8,-3) {$H_4$} ;

\draw (b21) -- (a21) ;
\draw (a33) -- (b22) -- (a23) ;
\draw (a35) -- (b23) -- (a25) ;

\node[draw,rectangle] (b32) at (4,-5) {$H_3$} ;
\node[draw,rectangle] (b33) at (8,-5) {$H_4$} ;

\draw (a33) -- (b32) ;
\draw (a35) -- (b33) ;

\draw (a11.west) --++(-0.3,0) --++(0,0.4) --++(11.4,0) --++(0,-0.4) -- (a16.east) ;
\draw (a21.west) --++(-0.3,0) --++(0,0.4) --++(11.4,0) --++(0,-0.4) -- (a26.east) ;
\draw (a33.west) --++(-0.3,0) --++(0,0.4) --++(7.4,0) --++(0,-0.4) -- (a36.east) ;

\draw (a11.north) --++(0,0.5) --++(0.6,0) --++(0,-4.4) --++(-0.6,0) -- (b21.south) ;
\draw (a13.north) --++(0,0.5) --++(0.6,0) --++(0,-6.4) --++(-0.6,0) -- (b32.south) ;
\draw (a15.north) --++(0,0.5) --++(0.6,0) --++(0,-6.4) --++(-0.6,0) -- (b33.south) ;

\end{tikzpicture}
}
\caption{The overall picture with $k=4$.}
\label{twlb-overall}
\end{figure}
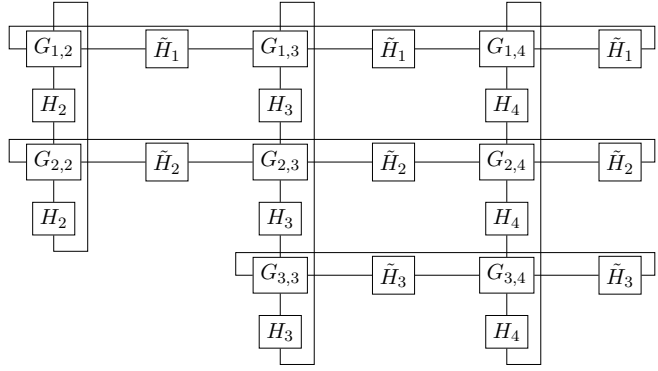

In order to describe how these gadgets are joined together in $G'$, as shown in \cref{twlb-overall}, we require some terminology.
Given some $G_{i,j}$ and $G_{i,j'}$ with $i,j,j' \in [k]$, %
we say we \emph{connect $G_{i,j}$ to $G_{i,j'}$ using $\tilde{H_i}$} to describe adding all nine edges between $D(G_{i,j})$ and $B(\tilde{H_i})$, and all nine edges between $D(\tilde{H_i})$ and $B(G_{i,j'})$.
In this case, we also say $\tilde{H_i}$ \emph{connects from $G_{i,j}$} and \emph{connects to $G_{i,j'}$}.
Given some $G_{i,j}$ and $G_{i',j}$ with $i,i',j \in [k]$, the operation of \emph{connecting $G_{i,j}$ to $G_{i',j}$ using $H_j$} is defined analogously.
We give the following cyclic ordering to the edge encoding gadgets: %
$(G_{1,2},G_{1,3},\dotsc,G_{1,k},G_{2,2},G_{2,3},\dotsc,G_{2,k},\dotsc,G_{k-1,k-1},G_{k-1,k})$.
For each $G_{i,j}$, we connect this gadget to the next gadget $G_{i,j'}$ in the cyclic ordering that matches on the first index using one of the copies of $\tilde{H_i}$, and also connect it to the next gadget $G_{i',j}$ in the ordering that matches on the second index using one of the copies of $H_j$.
For example, we connect $G_{1,3}$ to $G_{1,4}$ using a copy of $\tilde{H_1}$, and connect $G_{1,3}$ to $G_{2,3}$ using a copy of $H_3$.
This completes the construction.
    \begin{proof}[Proof of \cref{hardness}]
Observe that each vertex $v \in V(G')$ is contained in precisely one gadget, and so each vertex of $G'$ inherits either a `$u$' label or a `$w$' label from its gadget.
In what follows, whenever we refer to an edge encoding gadget $G_{i,j}$, or a propagator gadget $\tilde{H_i}$ or $H_j$, it is for some $i \in [1,k-1]$ and $j \in [2,k]$ with $i \leq j$.

\textbf{Treewidth.}
We now describe a path decomposition of $G'$ that illustrates that its pathwidth, and hence treewidth, is at most $54k-69$.

First, observe that for a gadget $H := G(\left<x_0,x_1,\dotsc,x_{z}\right>)$, there is a path decomposition where %
each bag has size at most~$4$.
By adding $B(H) \cup D(H)$ to every bag, we obtain a path decomposition where each bag has size at most~$10$; %
we denote this path decomposition by $\mathbf{P}(H)$.
Note that $H$ is only linked to other gadgets in $G'$ by edges with one end in either $B(H)$ or $D(H)$.

Recall that the edge encoding gadgets are joined together using propagator gadgets with respect to the cyclic ordering
$$(G_{1,2},G_{1,3},\dotsc,G_{1,k},G_{2,2},G_{2,3},\dotsc,G_{2,k},\dotsc,G_{k-1,k-1},G_{k-1,k}).$$
Consider an auxiliary multigraph $F$ on the vertex set $\{G_{i,j} : %
i \in [1,k-1], j \in [2,k], i \le j\}$
where there is an edge between $G_{i,j}, G_{i',j'} \in V(F)$ %
whenever
the gadget $G_{i,j}$ is connected to $G_{i',j'}$ using some propagator gadget in $G'$.
(Formally, there is an edge for $i=i'$ and $|j-j'| \in \{1,k-i,k-2\}$, or $j=j'$ and $|i-i'| \in \{1,j-1,k-2\}$.)

We now show that $F$ has pathwidth at most $3k-5$.
Let $\mathcal{G}_1 = \{G_{1,j} : j \in [2,k]\}$ and, for $i \in [2,k-1]$, let $\mathcal{G}_i = \{G_{i,j} : j \in [i,k]\}$.
Then $(\mathcal{G}_1 \cup \mathcal{G}_2 \cup \mathcal{G}_3, \mathcal{G}_1 \cup \mathcal{G}_3 \cup \mathcal{G}_4, \dotsc, \mathcal{G}_1 \cup \mathcal{G}_{k-2} \cup \mathcal{G}_{k-1})$ is a path decomposition for $F$ where
the largest bag, the first one, has size~$3k-4$.
We denote this path decomposition $\mathbf{P}(F)$.

We extend this to a path decomposition of $G'$ by
replacing each bag of $\mathbf{P}(F)$
with a path, which is in turn constructed from several concatenated ``subpaths'', one for each gadget.
Suppose, for some $i,j \in [k]$ with $i \le j$, we have that $\tilde{H_i}$ and $H_j$ connect to $G_{i,j}$ in $G'$, and $\tilde{H_i'}$ and $H_j'$ connect from $G_{i,j}$ in $G'$; then we denote
$X_{i,j} = D(\tilde{H_i}) \cup D(H_j) \cup B(G_{i,j}) \cup D(G_{i,j}) \cup B(\tilde{H_i'}) \cup B(H_j')$.
Let $Z \subseteq [k] \times [k]$ such that $\bigcup_{(i,j) \in Z} G_{i,j}$ is a bag of the path decomposition of $F$.
From this bag, we construct a path %
where each bag contains $Q=\bigcup_{(i,j) \in Z} X_{i,j}$.  %
The subpaths of this path are as follows.
For each $(i,j) \in Z$ we have a subpath obtained from $\mathbf{P}(G_{i,j})$ by adding $Q$ to each bag.
Every edge of $F$ is contained in some bag of the path decomposition, and corresponds to a propagator gadget $H$ %
of $G'$.
For each such $H$, we have a subpath obtained from $\mathbf{P}(H)$ by adding $Q$ to each bag.
These subpaths are then concatenated together, end to end, to create the path that replaces the bag $\bigcup_{(i,j) \in Z} G_{i,j}$ in $\mathbf{P}(F)$. 
After doing this for each bag, we obtain a path decomposition of $G'$.

Note that $|Z| \leq 3k-4$, %
and $|X_{i,j}| = 18$, for any $(i,j) \in Z$.  So $|Q| \leq 18(3k-4)$.  A path decomposition $\mathbf{P}(H)$, for some gadget $H$, has bags with size at most~$10$, but each bag meets $Q$ in precisely the elements $B(H) \cup D(H)$.  So the pathwidth of $G'$ is at most $18(3k-4)+4-1 = 54k-69$.

\textbf{Correctness ($\Rightarrow$).}
First, let $X$ be a multicolored $k$-clique in $G$; we will show that $G'$ has a set $S \subseteq V(G')$ such that $|S| = %
3 \binom{k+1}{2} -6$
and each component of $G'-S$ has at most $d$ vertices, where $d = 3t^2+3t+3$.
Let $\gamma(i)$ be the index of the unique vertex in $X \cap V_i$ for each $i \in [k]$; that is, $X \cap V_i = \{v_i^{\gamma(i)}\}$.
For each $H \in \mathcal{H}_i \cup \tilde{\mathcal{H}_i}$, we add the vertex $u_{\gamma(i)}(H)$ to $S$; there are %
$(k-2)(k+1)+2(k-1)=k(k+1)-4$ such gadgets, so this many
vertices are added to $S$ so far.
For each pair $i,j \in k$ with $i < j$,
there is some $q \in [p]$ such that $\phi(\gamma(i),\gamma(j)) = f_{i,j}^q$; we add the vertex $u_q(G_{i,j})$ to $S$.
For $i \in [2,k-2]$, we also add the vertex $u_{\gamma(i)}(G_{i,i})$ to $S$.
Now $|S| = k(k+1)-4 + \binom{k}{2} +k-2 = %
3\binom{k+1}{2}-6$.

We now consider the size of the components of $G'-S$.
We %
first analyze the size of the components of a gadget $G_{i,j}$, $\tilde{H_i}$ or $H_j$ after deleting $S$.
Note that $S$ meets the vertex set of one of these gadgets in precisely one vertex, and the deletion of this vertex splits the gadget into two components.
The two components of $G_{i,j} - u_q$ have $f_{i,j}^q = 3t(\gamma(i)-1)+3\gamma(j)$ and $f_{i,j}^{p} - f_{i,j}^q = 3t^2+3 - (3t(\gamma(i)-1)+3\gamma(j))$ vertices.
The two components of $\tilde{H_i} - u_{\gamma(i)}$ have $3t\gamma(i)$ and $3t(t+1-\gamma(i))$ vertices, while
the two components of $H_j - u_{\gamma(j)}$ have $3\gamma(j)$ and $3(t+1-\gamma(j))$ vertices.
These gadgets are joined in such a way that the size of a component of $G'-S$ is 
\begin{align*}
  & \big[3t(\gamma(i)-1)+3\gamma(j)\big] + 3t(t+1-\gamma(i)) + 3(t+1-\gamma(j)) \\
  &= 3t^2+3t+3 \\
  &= \big[3t^2+3 - \big(3t(\gamma(i)-1)+3\gamma(j)\big)\big] + 3t\gamma(i) + 3\gamma(j),
\end{align*}
as required.
Finally, observe that the only cycles in each component are contained in a gadget, and each gadget has no chordless cycles, so each component is a chordal graph.

\textbf{($\Leftarrow$).}
Suppose $G'$ has a set $S \subseteq V(G')$ with $|S| \le %
3 \binom{k+1}{2} -6$
such that each component of $G'-S$ has at most $d$ vertices, where $d = 3t^2+3t+3$.
We call any such set $S$ a \emph{solution}.

First, we show, loosely speaking, that we may assume each vertex in $S$ is a `$u$' vertex of its gadget, not a `$w$' vertex.
Let
$H$ be a gadget of $G'$ of order~$s$.
There are two cases to consider: the first is when, for some $r \in [1,s-1]$, we have that $S \cap V(P_r(H)) \neq \emptyset$.
Suppose %
$P_r(H)$ contains a pair of adjacent vertices $w$ and $w'$ such that $\{w,w'\} \cap S \neq \emptyset$.  If $w \in S$ and $w' \notin S$,
then, in $G'-(S \setminus \{w\})$, only the component containing $w'$ can have size more than $d$, and $|V(P_r(H))| \le 3t^2 < d$, so replacing $w'$ in $S$ with $u_{r-1}(H)$ or $u_r(H)$ also gives a solution.
If $\{w,w'\} \subseteq S$, %
then $(S \setminus \{w,w'\}) \cup \{u_{r-1}(H),u_r(H)\}$ is also a solution.
So we may assume that $V(P_r(H)) \cap S = \emptyset$ for each $r \in [1,s-1]$.

Now we consider the second case;
let $G_{i,j}$ be an edge encoding gadget, %
let $H \in \mathcal{H}_i$ and $\tilde{H} \in \tilde{\mathcal{H}_j}$ connect from $G_{i,j}$, and
let $J$ be the set of vertices $V(P_y(G_{i,j})) \cup V(P_z(H)) \cup V(P_z(\tilde{H}))$, for $(y,z) \in \{(p,0),(0,k+1)\}$.
Observe that $G'[J]$ is connected and $|J| \leq d$; intuitively, these are the vertices involved in the ``join'' of multiple gadgets in $G'$.
We show that if $J \cap S \neq \emptyset$, then there is some solution~$S'$ with $J \cap S' = \emptyset$.
Let $U := N_{G'}(J)$, %
so $|U|=3$.
If $|J \cap S| \geq 3$, then $(S\setminus J) \cup U$ is a solution.
Moreover, if $|U\setminus S| \le |J\cap S|$, then $(S\setminus J) \cup U$ is again a solution.
Assuming otherwise, we can pick $U' \subseteq U\setminus S$ such that $|U'| = |J\cap S|$.
If $G'[(J \cup U)\setminus S]$ is connected, then $S' = (S \setminus J) \cup U'$ is a solution.  But since $|J \cap S| \leq 2$, it follows, by the construction of $G'$, that $G'[J \setminus S]$ is connected.  Thus, in the exceptional case, the deletion of $J \cap S$ disconnects some $u \in U \setminus S$ from $G'[J\setminus S]$.  But in this case, if we ensure that $U'$ is chosen to contain $u$, then we still obtain a solution $S' = (S \setminus J) \cup U'$.

Next, we claim that %
each edge encoding gadget $G_{i,j}$ or propagator gadget $\tilde{H_i} \in \tilde{\mathcal{H}_i}$, has at least one vertex in $S$.
Consider the subgraph $D_{i,j}$ of $G'$ induced by $V(G_{i,j}) \cup V(\tilde{H_i}) \cup V(H_j)$, where $\tilde{H_i}$ and $H_j$ connect from $G_{i,j}$.
Recall that $G_{i,j}$ consists of $3t^2+3+1$ vertices, $\tilde{H_{i}}$ consists of $3t^2+3t+1$ vertices, $H_{j}$ consists of $3t+3+1$ vertices, and hence $D_{i,j}$ has size $2d+3$.
If $V(\tilde{H_{i}}) \cap S$ is empty, then the connected subgraph of $D_{i,j}-S$ containing $V(\tilde{H_i})$ also contains $P_p(G_{i,j})$, which has size at least $3$, so this connected subgraph contains
at least $3t^2 + 3t + 1 + 3 = d+1$ vertices; a contradiction.
Similarly, if $V(G_{i,j}) \cap S$ is empty, then the connected subgraph of $D_{i,j}-S$ containing $V(G_{i,j})$ also contains at least $3t$ vertices of $V(\tilde{H_i})$, so at least $d + 1$ in total; a contradiction.
So $|V(\tilde{H_{i}}) \cap S|, |V(G_{i,j}) \cap S| \geq 1$, as claimed.

Now we claim that each component of $G'-S$ has size exactly~$d$.
Pick $S' \subseteq S$ such that %
$|V(G_{i,j}) \cap S'| = 1$ for each edge encoding gadget $G_{i,j}$, 
and $|V(\tilde{H_i}) \cap S'| = 1$ for each $\tilde{H_i} \in \tilde{\mathcal{H}_i}$.
It follows that $|S'| = 2\big(\binom{k+1}{2}-2\big)$, and %
$|S \setminus S'| = \binom{k+1}{2}-2$.
Now, for distinct propagator gadgets $H, H' \in \bigcup \mathcal{H}_i$, there is no path in $G'-S'$ between a vertex in $H$ and a vertex in $H'$, so $G'-S'$ has at least $\binom{k+1}{2}-2$ components, one for each $H \in \bigcup \mathcal{H}_i$.
In fact, for every vertex $v$ of $G'-S$, there exists a vertex $v' \in V(H)$ for some $H \in \bigcup \mathcal{H}_i$ such that there is a path from $v$ to $v'$, so 
$G'-S'$ has precisely $\binom{k+1}{2}-2$ components.
Moreover, since $S$ consists only of `$u$' vertices, the deletion of each vertex in $S \setminus S'$ further increases the number of components by one.
As $|V(G')| = (2d+3) \big(\binom{k+1}{2}-2\big)$, so
$|V(G'-S)| = 2d \big(\binom{k+1}{2}-2\big)$,
and each of the $2\big(\binom{k+1}{2}-2\big)$ components of $G'-S$ has size at most $d$,
these components must have size precisely $d$, as claimed.

Next we show that each gadget $H_j \in \mathcal{H}_j$ %
also has at least one vertex in $S$.
Suppose we have some $H_j$ for which $S \cap V(H_j) = \emptyset$.
We calculate the size, modulo $3$, of the component $C$ of $G'-S$ that contains $H_j$.  Since the size of $V(C) \cap V(\tilde{H_i})$ or $V(C) \cap V(G_{i,j})$
is congruent to $0 \pmod 3$, and $|V(H_j)| \equiv 1 \pmod 3$,
we deduce that $|V(C)| \equiv 1 \pmod 3$; a contradiction.
So $|S \cap V(H_j)| \geq 1$ for every $H_j \in \mathcal{H}_j$ with $j \in [2,k]$.
Since $|S| = 3\binom{k}{2}$, it follows that each gadget meets $S$ in precisely one vertex.

Finally, suppose $u_{q}(G_{i,j}) \in S$, for some %
$q \in [p]$.
Then $\phi(a,b) = f_{i,j}^q$, for some $a,b \in [t]$.
Let $\tilde{H_i} \in \mathcal{H}_i$ and $H_j \in \mathcal{H}_j$ be the propagators that connect from $G_{i,j}$.
Now, the component of $G'-S$ containing $3t^2+3 - (3t(a-1)+3b)$ vertices of $G_{i,j} - u_q$ also contains $3ta'$ vertices of $\tilde{H_i}$, and $3b'$ vertices of $H_j$, for some $a',b' \in [t]$.
So %
$$3t^2 + 3ta' - 3t(a-1) + 3b' -3b +3 = 3t^2+3t+3.$$
Working modulo $t$, we deduce that $3(b'-b+1) \equiv 3 \pmod t$, hence $b = b'$.
It then follows that $3t(a'-(a-1)) = 3t$, so $a = a'$.
Thus $u_a(\tilde{H_i}), u_b(H_j) \in S$.

On the other hand, if for some $a,b \in [t]$ we have $u_a(\tilde{H_i}), u_b(H_{j}) \in S$, where $\tilde{H_i}$ and $H_{j}$ connect to $G_{i,j}$, then the component of $G'-S$ containing vertices from these three gadgets contains $3t(t+1-a)$ vertices from $\tilde{H_i}$,
as well as $3(t+1-b)$ vertices from $H_j$, and
$3t(a'-1) + 3b'$ from $G_{i,j}$ for some $a',b' \in [t]$.
Since this component has a total of $3t^2+3t+3$ vertices, working modulo $t$ we deduce that $3b' + 3 - 3b \equiv 3 \pmod t$, so $b=b'$.
It follows that 
$3t(a-a'+1) = 3t$, so $a=a'$.
Thus, $u_{q}(G_{i,j}) \in S$ for $q \in [p]$ such that $\phi(a,b) = f_{i,j}^q$.

We deduce that for every $l \in [k]$, there exists some $\gamma(l)$ such that
$V(\tilde{H}) \cap S = \{u_{\gamma(i)}\}$ for every $\tilde{H} \in \tilde{\mathcal{H}_{i}}$, 
$V(H) \cap S = \{u_{\gamma(j)}\}$ for every $H \in \mathcal{H}_j$, and
$V(G_{i,j}) \cap S = \{u_q\}$ for $q \in [p]$ such that $f_{i,j}^q=\phi(\gamma(i),\gamma(j))$.
It follows that each $v_i^{\gamma(i)}v_j^{\gamma(j)}$ is an edge of $G$, and $X=\{v_i^{\gamma(i)} : i \in [k]\}$ is a multicolored $k$-clique in $G$, as required.
\end{proof}

\Cref{hardness} implies that \BPCVD has no algorithm running in time $f(w)n^{\cO(1)}$,
assuming $\textrm{FPT} \neq W[1]$.
However, we can say something stronger, assuming the ETH holds.
Since, in the parameterized reduction in the previous proof, the treewidth of the reduced instance $G'$ has linear dependence on $k$, a $f(w)n^{o(w)}$-time algorithm for %
this problem
would lead to a $f(k)n^{o(k)}$-time algorithm for \MC.
But, assuming the ETH holds, no such algorithm for \MC exists~\cite{Lokshtanov2011a}.
So we have the following:

\begin{theorem}
  Unless the ETH fails, there is no %
  $f(w)n^{o(w)}$-time algorithm for %
  \BPCVD when $\cP$ contains all chordal graphs.
\end{theorem}

Furthermore, Marx~\cite{Marx2010} showed that, assuming the ETH holds, \SI has no $f(k)n^{o(k / \log k)}$-time algorithm, where $k$ is the number of edges in the smaller graph.
By reducing from \SI, instead of \MC, we obtain a lower bound %
with the combined parameter treewidth and solution size.

\begin{theorem}
    \label{sireduc}
  Unless the ETH fails, there is no 
  $f(k')n^{o(k'/\log k')}$-time algorithm for %
  \BPCVD,
  where $k' = w + k$,
  when $\cP$ contains all chordal graphs.
\end{theorem}
\begin{proof}
  Let $(G,H)$ be a \SI instance where the task is to find if $G$ has a subgraph isomorphic to $H$.
Let $k:=|V(H)|$ and $t:=|V(G)|$, and suppose $V(G) = \{v^a : a \in [t]\}$ and $V(H) = \{v_i : i \in [k]\}$.
Let $V_i = \{v_i^a : a \in [t]\}$ for each $i \in [k]$, and
let $G^+$ be the graph on the vertex set $\bigcup_{i \in [k]} V_i$
with an edge $v_i^av_j^b$ if and only if $i \neq j$ and $v^av^b$ is an edge of $G$.
Now the task is to select $|E(H)|$ edges of $G^+$ that induce a \emph{multicolored} subgraph of $G^+$; that is, the vertex set of this edge-induced subgraph meets each $V_i$ in exactly one vertex.

\sloppy
We construct $G'$ from $G^+$ using a similar construction as in the proof of \cref{hardness}, but we only have an edge encoding gadget $G_{i,j}$ for $1 \leq i < j \leq k$ when $v_iv_j$ is an edge in $H$.
More specifically, we take the subsequence of $(G_{1,2},G_{1,3},\dotsc,G_{1,k},G_{2,2},G_{2,3},\dotsc,G_{2,k},\dotsc,G_{k-1,k-1},G_{k-1,k})$ consisting of each $G_{i,j}$ for which $v_iv_j \in E(H)$, as well as $G_{i,i}$ for all $i \in [2,k-1]$, and, as before, connect
each $G_{i,j}$ to the next $G_{i,j'}$ in the cyclic ordering that matches on the first index using a copy of $\tilde{H_i}$, and also connect it to the next gadget $G_{i',j}$ in the ordering that matches on the second index using a copy of $H_j$.
Note that $p=|E_{i,j}| = 2|E(G)|$. %

\fussy
By a routine adaptation of \cref{hardness}, it is easy to see that %
$\tw(G') = \cO(k)$, and that $G$ has a subgraph isomorphic to $H$ if and only if $G'$ has a set $S \subseteq V(G')$ of size at most $k'$ such that each component of $G'-S$ has size at most~$d$.
Now the parameter in the reduced instance is $k'' := \tw(G') + k' = \cO(|V(H)|) + \cO(|V(H)|^2) = \cO(|E(H)|)$.
Thus, an $f(k'')n^{o(k''/\log k'')}$-time algorithm for %
\BPCVD
would lead to an algorithm for \SI running in time $f(|E(H)|)n^{o(|E(H)|/\log |E(H)|)}$. %
But there is no algorithm for \SI with this running time unless the ETH fails~\cite{Marx2010}.
\end{proof}

\def\ocirc#1{\ifmmode\setbox0=\hbox{$#1$}\dimen0=\ht0 \advance\dimen0
  by1pt\rlap{\hbox to\wd0{\hss\raise\dimen0
  \hbox{\hskip.2em$\scriptscriptstyle\circ$}\hss}}#1\else {\accent"17 #1}\fi}

\end{document}